\documentclass[english]{llncs}

\usepackage{times}
\usepackage{comment}
\usepackage{hyperref}
\usepackage{refcount}
\hypersetup{
    colorlinks=true,
    pdfborder={0 0 0},
}

\makeatletter
\def\namedlabel#1#2{\begingroup
    #2%
    \def\@currentlabel{#2}%
    \phantomsection\label{#1}\endgroup
}
\makeatother

\usepackage{url}
\usepackage{wrapfig}

\usepackage{caption}
\usepackage{subcaption}
\usepackage{amsmath,amsthm}
\usepackage{nccmath}
\usepackage{amssymb}
\usepackage{wrapfig}
\usepackage{graphicx}
\graphicspath{{./Figs/}}
\usepackage{paralist}
\usepackage[active]{srcltx}
\usepackage{xparse}
\usepackage{tabularx}
\usepackage{longtable}
\usepackage{ifthen}
\usepackage{ifmtarg}
\usepackage[usenames,dvipsnames]{xcolor}
\usepackage{relsize}
\usepackage[stable]{footmisc}

\usepackage[lined,algoruled]{algorithm2e}
\usepackage{calc}
\newcolumntype{L}[1]{>{\raggedright\arraybackslash}p{#1}} % linksb�ndig mit Breitenangabe
\newcolumntype{C}[1]{>{\centering\arraybackslash}p{#1}} % zentriert mit Breitenangabe
\newcolumntype{R}[1]{>{\raggedleft\arraybackslash}p{#1}} % rechtsb�ndig mit Breitenangabe

\newcommand{\refer}[1]{
\sbox0{#1}%
  \ifdim\wd0=0pt
    {\textcolor{CadetBlue}{[ref needed]}}% if #1 is empty
  \else
    {\textcolor{CadetBlue}{[ref: #1]}}% if #1 is not empty
  \fi
}

\theoremstyle{definition}

\newtheorem*{deferredproof}{Proof}

\usepackage{fp}

\newlength{\TOarg} \newlength{\TOunit}
{\catcode`p=12 \catcode`t=12 \gdef\TOnum#1pt{#1}}
\newcommand\TOop[2]{\setlength{\TOarg}{#2}%
   \FPdiv\TOres{\expandafter\TOnum\the\TOarg}{\expandafter\TOnum\the\TOunit}%
   \FPround\TOres\TOres{#1}}

\begin{document}

% FUNCTIONS
	%GRAPH
	\newcommand{\deltaOut}[1]{\delta^+_{#1}}
	\newcommand{\deltaIn}[1]{\delta^-_{#1}}
	\newcommand{\head}{\textnormal{head}}
	\newcommand{\tail}{\textnormal{tail}}
	\newcommand{\negate}[1]{\overline{#1}}

% FONTS
	\newcommand{\fontMacro}[1]{\mathsf{#1}}
	\newcommand{\fontVariable}[1]{\mathrm{#1}}
	\newcommand{\fontParameter}[1]{\mathbf{#1}}
	\newcommand{\fontParameterElement}[1]{\textnormal{#1}}
	\newcommand{\fontConstraintName}[1]{\textsc{#1}}
	%Shortcuts:
	\newcommand{\foM}[1]{\fontMacro{#1}}
	\newcommand{\foV}[1]{\fontVariable{#1}}
	\newcommand{\foP}[1]{\fontParameter{#1}}
	\newcommand{\foPE}[1]{\fontParameterElement{#1}}
	\newcommand{\foCN}[1]{\fontConstraintName{#1}}

\newcommand{\nat}{\mathbb{N}}
\newcommand{\psnat}{\mathbb{N}_{> 0}}
\newcommand{\reals}{\mathbb{R}}
\newcommand{\preals}{\mathbb{R}_{\geq 0}}
\newcommand{\nreals}{\mathbb{R}_{\leq 0}}

\newcommand{\psreals}{\mathbb{R}_{> 0}}
\newcommand{\nsreals}{\mathbb{R}_{< 0}}

\newcommand{\integers}{\mathbb{Z}}
\newcommand{\pint}{\mathbb{Z}_{\geq 0}}
\newcommand{\nint}{\mathbb{Z}_{\leq 0}}

\newcommand{\psint}{\mathbb{Z}_{> 0}}
\newcommand{\nsint}{\mathbb{Z}_{< 0}}

%network
\newcommand{\G}{G}
\newcommand{\VG}{V_{G}}
\newcommand{\EG}{E_{G}}
\newcommand{\uE}{u_{E}}
\newcommand{\cE}{c_{E}}

%VSA related
\newcommand{\Tree}{\mathcal{T}_{\G}}
\newcommand{\Vtree}{V_{\mathcal{T}}}
\newcommand{\Etree}{E_{\mathcal{T}}}
\newcommand{\fPath}{\pi}

%extended network
\newcommand{\Gext}{G_{\textnormal{ext}}}
\newcommand{\Vext}{V_{\textnormal{ext}}}
\newcommand{\Eext}{E_{\textnormal{ext}}}
\newcommand{\EextMinus}{E^{-}_{\textnormal{ext}}}

\newcommand{\Enagg}{E^{R}_\textnormal{ext}}
\newcommand{\EextSm}{E^{S^-}_\textnormal{ext}}
\newcommand{\EextSp}{E^{S^+}_\textnormal{ext}}
\newcommand{\EextTp}{E^{T^+}_\textnormal{ext}}
\newcommand{\EextR}{E^{r}_\textnormal{ext}}

\newcommand{\cEext}{c_{\Eext}}
\newcommand{\uEext}{u_{\Eext}}
\newcommand{\sSource}{\textnormal{o}^+}
\newcommand{\sSink}{\textnormal{o}^-}
\newcommand{\sSinkS}{\textnormal{o}^-_S}
\newcommand{\sSinkR}{\textnormal{o}^-_r}

%root related
\newcommand{\eR}{r}
\newcommand{\uR}{u_{r}}

%Steiner node related
\newcommand{\sS}{S}
\newcommand{\cS}{c_S}
\newcommand{\uS}{u_S}

%Termina related
\newcommand{\sT}{T}

% VARIABLES

%flow related
\newcommand{\f}{f}
\newcommand{\fopt}{\hat{f}}
\newcommand{\fFopt}{\hat{f}_F}

%Steiner node related
\newcommand{\xS}{x}

% CONSTANTS
\newcommand{\cu}{\fontParameter{u}}
\newcommand{\cc}{\fontParameter{c}}

% solution related

\newcommand{\solX}{\hat{\xS}}
\newcommand{\solf}{\hat{\f}}
\newcommand{\solS}{\hat{S}}
\newcommand{\solTree}{\hat{\mathcal{T}}_G}
\newcommand{\solT}{\hat{T}}
\newcommand{\unT}{\bar{T}}
\newcommand{\solt}{\bar{t}}
\newcommand{\solVtree}{\hat{V}_\mathcal{T}}
\newcommand{\solEtree}{\hat{E}_\mathcal{T}}
\newcommand{\solfPath}{\hat{\pi}}

\newcommand{\fGext}{G_{\textnormal{ext}}^f}
\newcommand{\fpGext}{G_{\textnormal{ext}}^{f'}}
\newcommand{\fVext}{V_{\textnormal{ext}}^f}
\newcommand{\fEext}{V_{\textnormal{ext}}^f}

\newcommand{\fsolVext}{V_{\textnormal{ext}}^{\solf}}
\newcommand{\fsolEext}{E_{\textnormal{ext}}^{\solf}}

\newcommand{\FeasibleCVSAP}{\mathcal{F}_{\textnormal{CVSAP}}}
\newcommand{\FeasibleACVSAP}{\mathcal{F}_{\textnormal{A-CVSAP}}}
\newcommand{\FeasibleMCVSAP}{\mathcal{F}_{\textnormal{M-CVSAP}}}
\newcommand{\FeasibleIP}{\mathcal{F}_{\textnormal{IP}}}
\newcommand{\FeasibleLP}{\mathcal{F}_{\textnormal{LP}}}

\newcommand{\cCVSAP}{C_{\textnormal{CVSAP}}}
\newcommand{\cIP}{C_{\textnormal{IP}}}

%MCF related

\newcommand{\GextM}{G_{\textnormal{\tiny{MCF}}}}
\newcommand{\VextM}{V_{\textnormal{\tiny{MCF}}}}
\newcommand{\EextM}{E_{\textnormal{\tiny{MCF}}}}
\newcommand{\EextMT}{E^{T}_{\textnormal{\tiny{MCF}}}}
\newcommand{\EextMS}{E^{S}_{\textnormal{\tiny{MCF}}}}

\newcommand{\EnaggM}{E^{R}_\textnormal{\tiny{MCF}}}
\newcommand{\EextSmM}{E^{S^-}_\textnormal{\tiny{MCF}}}
\newcommand{\EextSpM}{E^{S^+}_\textnormal{\tiny{MCF}}}
\newcommand{\EextTpM}{E^{T^+}_\textnormal{\tiny{MCF}}}
\newcommand{\EextRM}{E^{r}_\textnormal{\tiny{MCF}}}

\newcommand{\cEextM}{c_{\Eext}}
\newcommand{\uEextM}{u_{\Eext}}

\newcommand{\solP}{\hat{p}}
\newcommand{\solfS}{\hat{f}^s}
\newcommand{\solfT}{\hat{f}^T}

\newcommand{\cIPM}{C_{\textnormal{\tiny{MCF}}}}

\newcommand{\FeasibleMCF}{\mathcal{F}_{\textnormal{MCF}}}

\newcommand{\priority}{p}
\newcommand{\order}{o}

\title{The Constrained Virtual Steiner Arborescence Problem:\\ \vspace{4pt} \smaller[1]{Formal Definition, Single-Commodity Integer Programming Formulation and Computational Evaluation}\thanks{A conference version of this paper will appear in the proceedings of OPODIS 2013, \copyright Springer LNCS. It will be available at link.springer.com.} }
%The Virtual Steiner Arborescence Problem}

\author {
   Matthias Rost, Stefan Schmid\\
\texttt{\{mrost,stefan\}@net.t-labs.tu-berlin.de}
}

\institute{Telekom Innovation Laboratories (T-Labs) \& TU Berlin, Germany}

\date{}

\maketitle \thispagestyle{empty}

\begin{abstract}
As the Internet becomes more virtualized and software-defined, new functionality is introduced in the network core: the distributed resources available in ISP central offices, universal nodes, or datacenter middleboxes can be
used to process (e.g., filter, aggregate or duplicate) data. Based on this new networking paradigm, we formulate the Constrained Virtual Steiner Arborescence
Problem (CVSAP) which asks for optimal locations to perform in-network processing, in order to jointly minimize processing costs and network traffic while
respecting link and node capacities. 

We prove that CVSAP cannot be approximated (unless $P = \mathit{NP}$), and accordingly, develop the exact algorithm VirtuCast to compute (near) optimal solutions. VirtuCast consists of: (1) a compact single-commodity flow Integer Programming (IP) formulation; (2) a flow decomposition algorithm to reconstruct individual routes from the IP solution. The compactness of the IP formulation allows for computing lower bounds even on large instances quickly, speeding up the algorithm. We rigorously prove VirtuCast's correctness.

To complement our theoretical findings, we have implemented VirtuCast and present an extensive computational evaluation, showing that, using VirtuCast, realistically sized instances can be solved (close to) optimality. We show that VirtuCast significantly improves upon naive multi-commodity formulations and also initiate the study of primal heuristics to generate feasible solutions during the branch-and-bound process.
\end{abstract}

\section{Introduction} 
\label{sec:introduction}

Multicast and aggregation are two fundamental functionalities offered by many communication networks. In order to efficiently distribute content (e.g., live TV) to multiple receivers, a multicast solution should duplicate the content as close to the receivers as possible. Analogously, in aggregation applications such as distributed network monitoring, the monitoring data may be filtered or aggregated along the path to the observer, to avoid redundant transmissions over physical links.
Efficient multicasting and aggregation is a mature research field, and many important theoretical and practical results have been obtained over the last decades. Applications range from IPTV~\cite{BLTJ:BLTJ20217} over sensor networks~\cite{stacy,in-net-sensor} to fiber-optical transport~\cite{BLTJ:BLTJ20217}.

This paper is motivated by the virtualization trend in today's Internet, and in particular by network (function) virtualization \cite{ETSIwhitepaperNFV} 
and software-defined networking, e.g., OpenFlow~\cite{simple}. 
In virtualized environments, resources can be allocated or leased flexibly at the locations where they are most useful or cost-effective: Computational and storage resources available at middleboxes in datacenters~\cite{camdoop} or in distributed (micro-) datacenters in the wide-area network~\cite{gigascope,dcc13,ratul} can for example be used for in-network processing, e.g., to reduce traffic during the MapReduce shuffle phase~\cite{naas}.

Such distributed resource networks open new opportunities on how services can be deployed.
In the context of aggregation and multicasting, for example, a new degree of freedom arises:
the sites (i.e., the number and locations) used for the data processing, becomes subject to optimization.

This paper initiates the study of how to efficiently allocate in-network processing functionality in order to jointly minimize
network traffic and computational resources.
Importantly, for many of these problem variants, classic Steiner Tree models ~\cite{voss2006SteinerChapter} are no longer applicable~\cite{molnar2011hierarchies}.
Accordingly, we coin our problem the \emph{Constrained Virtual Steiner Arborescence Problem (CVSAP)}, as the goal is to install a set of processing nodes and to connect all terminals via them to a single root.

\begin{figure}[t]
	\begin{subfigure}[t]{0.28\textwidth}
		\centering
		\includegraphics[scale=0.6]{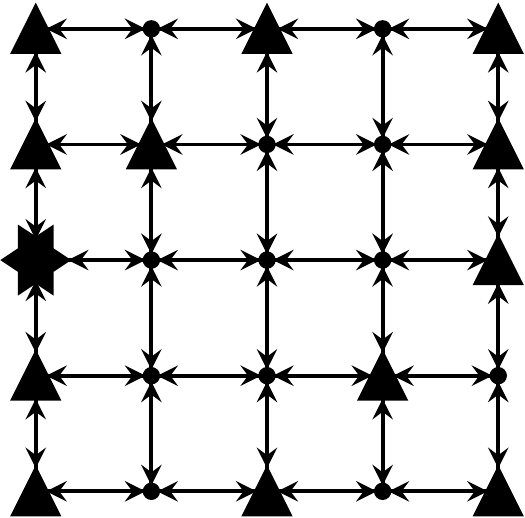}	
		\caption{$5 \times 5$ Grid Topology}
		\label{fig:example-idea-topology}
	\end{subfigure}
	\qquad
	\begin{subfigure}[t]{0.28\textwidth}
			\centering
		\includegraphics[scale=0.6]{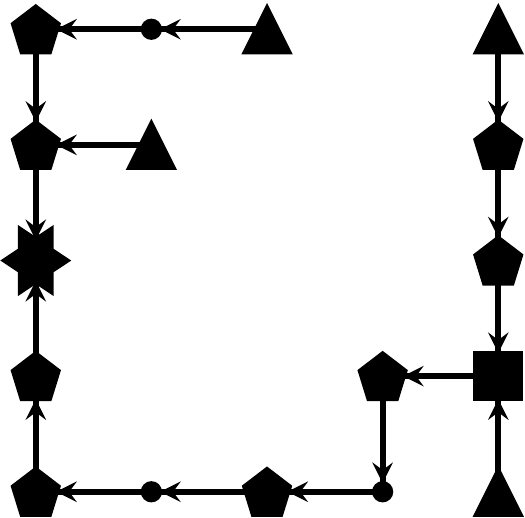}	
		\caption{Steiner Arborescence}
		\label{fig:example-idea-steiner}
	\end{subfigure}		
	\qquad
	\begin{subfigure}[t]{0.28\textwidth}
			\centering
		\includegraphics[scale=0.6]{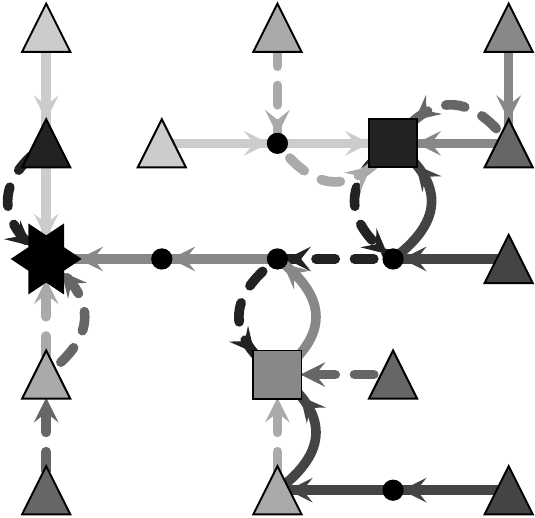}	
		\caption{Virtual Arborescence}
		\label{fig:example-idea-paths}
	\end{subfigure}
	\caption{An aggregation example on a $5 \times 5$ grid. Terminals are depicted as triangles while the receiver is depicted as star. The terminals must establish a path towards the receiver, while multiple data streams may be aggregated by activated processing locations. Such processing locations are pictured as squares or, in case that a processing location is collocated with a terminal, pentagons. In Figure~(\subref{fig:example-idea-paths}), equally colored and dashed edges represent paths, originating at the node with the same color.}
	\label{fig:example-idea}
\end{figure}

\subsubsection{Example.} To illustrate our model, consider the aggregation example depicted in Figure~\ref{fig:example-idea}. The terminals must connect to the receiver, while processing functionality can be placed on nodes to aggregate any number of incoming data flows into a single one. Assuming no costs for placing processing functionality, the problem reduces to the Steiner Arborescence Problem and the optimal solution, depicted in Figure~\ref{fig:example-idea-steiner}, uses $16$ edges and $9$ processing locations.
However, assuming unit edge costs and activation costs of $5$ for processing locations, this solution is suboptimal. Figure~\ref{fig:example-idea-paths} depicts a solution which only uses $2$ processing locations and $26$ edges: Terminals in the first column directly connect to the receiver, while the remaining terminals use one of the two processing nodes. Note that we allow for nested processing of flows: the upper processing node forwards its aggregation result to the lower processing node, from where the result is then forwarded to the receiver.

\subsubsection{Contribution.}

This paper presents the first concise graph-theoretic formulation of the \emph{Constrained Virtual Steiner Arborescence Problem (CVSAP)} which captures the trade-off between traffic optimization benefits and in-network processing costs arising in virtualized environments, and which also generalizes many classic in-network processing problems related to multicasting and aggregation.
We prove that CVSAP cannot be approximated unless $\mathit{NP} \subseteq \mathit{P}$ holds and therefore focus on obtaining provably good solutions for CVSAP in non-polynomial time. To this end we introduce the algorithm VirtuCast, which is based on Integer Programming (IP) and allows to obtain optimal solutions. The advantage of VirtuCast lies in the fact that even for large problem instances, when optimal solutions cannot be computed in reasonable time, our approach bounds the gap to optimality as lower bounds are computed on the fly. 

VirtuCast consists of two components: a single-commodity IP formulation which can be solved by branch-and-cut methods and a decomposition algorithm to construct the routing scheme. Our IP formulation not only uses a smaller number of variables compared to alternative multi-commodity IP formulations, but also yields good linear relaxations in practice.

Our main contribution is the constructive proof that any solution to our IP formulation can be decomposed to yield a valid routing scheme connecting all terminals via processing nodes to the root. This is intriguing, as the single-commodity flow in the network is not restricted to directed acyclic graphs (cf.~Figure~\ref{fig:example-idea-paths}). In fact, as already shown in \cite{molnar2011hierarchies}, forbidding directed acyclic graphs (DAGs) may yield suboptimal solutions. Rather, we allow for the iterative processing of flows, such that processing nodes may be connected to other processing nodes.

To complement our theoretical findings, we present an extensive computational study. As our results indicate, using VirtuCast realistically sized instances can be solved (close to) optimality. Adapting well-known separation techniques from Steiner Tree literature, we improve the runtime of VirtuCast by the order of a magnitude. To compute feasible solutions to CVSAP in the course of the branch-and-bound process, we developed a novel primal heuristic and investigate its performance. To demonstrate the advantages of our single-commodity formulation over multi-commodity formulations, we have implemented a simple multi-comodity formulation and show its computational inferiority.

\subsubsection{Overview.}

In Section~\ref{sec:problem-definition} we formally introduce CVSAP and its aggregation and multicast versions and show its general inapproximability. We continue by presenting our algorithm VirtuCast, that relies on a single-commodity IP formulation, in Section~\ref{sec:ilp-formulation}. For a later computational comparison, we introduce a multi-commodity Mixed Integer Program for CVSAP in Section~\ref{sec:mcf-formulation}. We present our implementation of VirtuCast and our primal heuristic to generate feasible solutions in the course of the branch-and-bound process in Section~\ref{sec:implementation}. Section~\ref{sec:results} presents our extensive computational evaluation, containing a validation of our implementation choices, an analysis of VirtuCast's performance on realistically sized instances, an analysis of our primal heuristic's performance as well as the computational comparison with the multi-commodity flow formulation. We conclude this paper with summarizing related work in Section~\ref{sec:related-work}.

\section{The Constrained Virtual Steiner Arborescence Problem} 
\label{sec:problem-definition}

The Constrained Virtual Steiner Arborescence Problem (CVSAP) considers multicast routing and optimal in-network aggregation problems in which processing locations can be \emph{chosen} to reduce traffic. As using (or leasing) in-network processing capabilities comes at a certain cost (e.g., the corresponding resources cannot be used by other applications), there is a trade-off between additional processing and traffic reduction. In contrast to the classic Steiner Tree Problems~\cite{voss2006SteinerChapter}, our model distinguishes between nodes that merely forward traffic and nodes that may actively \emph{process flows}. Informally, the task is to construct a minimal cost spanning arborescence on the set of  active processing nodes, sender(s) and receiver(s), such that edges in the arborescence correspond to paths in the original graph. As edges in the arborescence represent logical links (i.e., routes) between nodes, we refer to the problem as \emph{Virtual} Steiner Arborescence Problem. Based on the notion of virtual edges, the underlying paths may overlap and may use both the (resource-constrained) nodes and edges in the original graph multiple times (cf. Figure~\ref{fig:example-idea-paths}). We naturally adopt the notion of \emph{Steiner nodes} in our model, and refer to processing nodes contained in the virtual arborescence as \emph{active Steiner nodes}. The following notations will be used throughout this paper.

\subsubsection{Notation.}
In a directed graph $\G=(\VG,\EG)$ we denote by $\mathcal{P}_\G$ the set of all simple, directed paths in $\G$. Given a set of simple paths $\mathcal{P}$, we denote by $\mathcal{P}[e]$ the subset of paths contained in $\mathcal{P}$ that contain edge $e$. We use the notation $P=\langle v_1,v_2,\dots, v_n \rangle$ to denote the directed path $P$ of length $|P|=n$ where $P_i \triangleq v_i \in \VG$ for $1\leq i \leq n$ and $(v_i,v_{i+1}) \in \EG$ for $1\leq i < n$.
We denote the set of outgoing and incoming edges, restricted on a subset $F\subseteq \EG$, for node $v \in \VG$ by $\delta^+_F(v)=\{(v,u) \in F\}$ and $\delta^-_F(v)=\{(u,v) \in F\}$ and set $\delta^+_F(W)=\{(v,u) \in F|v\in W, u\notin W\}$ and respectively $\delta^-_F(W)=\{(u,v)\in F|v\in W, u\notin W\}$. We abridge $\foM{f}((\foM{y},\foM{z}))$ to $\foM{f}(\foM{y},\foM{z})$ for functions defined on tuples.

\subsubsection{Formal Problem Statement.}
Our general problem definition presented henceforth captures both the multicast and the aggregation scenarios. We model the physical infrastructure as capacitated, directed network $\G = (\VG,\EG,\cE,\uE)$ with integral capacities on the edges $\uE: \EG \to \mathbb{N}$ and real-valued, positive edge costs $\cE: \EG \to \mathbb{R^+}$. On top of this network, we define an abstract request $R_G=(r,S,T,\uR,\cS,\uS)$, where $T \subseteq \VG$ defines the set of terminals that need to be connected with the root $r \in \VG \setminus T$, for which an integral capacity $\uR \in \mathbb{N}$ is given.  The set $S \subseteq \VG \setminus (\{r\} \cup T)$ denotes the set of possible \emph{Steiner sites}, i.e. nodes at which Steiner nodes may be \emph{activated}. Steiner sites are attributed with a positive cost that is incurred upon using it as active Steiner node $\cS : S \to \mathbb{R^+}$, and an integral capacity $\uS: S \to \mathbb{N}$. It should be noted that we require the sets $S$ and $T$ to be disjoint for terminological reasons. A node $v \in S \cup T$ can easily be modeled by introducing a new node $v_T \in T$ and letting $v \in S$ such that $v_T$ is only  connected to $v$ with $\cE(v,v_T)=0$ and $\uE(v,v_T)=1$.

In the aggregation scenario the terminals represent nodes holding data that needs to be forwarded to the root (the single receiver) while data may be aggregated at active Steiner nodes. Contrary, in the multicast scenario the root represents the single sender that must stream (the same) data to each of the terminals, while active Steiner nodes may duplicate and reroute the stream. The capacities on the root and on the Steiner sites limit the degree in the Virtual Arborescence, formally introduced next.

\begin{definition}[Virtual Arborescence] 
\label{def:VA}
Given a directed graph $G=(\VG,\EG)$ and a root $r\in \VG$, a \emph{Virtual Arborescence (VA)} on $G$ is defined as $\mathcal{T}_G=(\Vtree,\Etree,r,\fPath)$ where $\{r\} \subseteq \Vtree \subseteq \VG$, $\Etree \subseteq \Vtree \times \Vtree$ and $\fPath: \Etree \to \mathcal{P}_G $ maps each edge in the arborescence on a simple directed path $P \in \mathcal{P}_G$ such that
\begin{description}
\item[\quad (\namedlabel{def:VA1}{\textnormal{VA-1}})] $(\Vtree, \Etree, r)$ is an arborescence root at $r$ with edges either directed towards or away from $r$,
\item[\quad (\namedlabel{def:VA2}{\textnormal{VA-2}})] for all $(u,v) \in \Etree$ the directed path $\fPath(u,v)$ connects $u$ to $v$ in $G$. \hfill $\qed$
\end{description} 
\end{definition}

A link $(v,w)\in \Etree$ represents a logical connection between nodes $v$ and $w$ while the function $\fPath(v,w)=P$ defines the route taken to establish this link. Note that the directed path $P$ must, \emph{pursuant} to the orientation $(v,w)$ of the logical link in the arborescence, start with $v$ and end at $w$. Figure~\ref{fig:example-idea-paths} illustrates our definition of the VA:  equally colored and dashed paths represent edges of the Virtual Arborescence. Using the concept of Virtual Arborescence, we can concisely state the problem we are attending to.
\begin{definition}[Constrained Virtual Steiner Arborescence Problem] 
\label{def:CVSAP}
Given a directed capacitated network $G=(\VG,\EG,\cE,\uE)$ and a request $R_G =(r,S,T,\uR,\cS,\uS)$ as above,  the \emph{Constrained Virtual Steiner Arborescence Problem (CVSAP)} asks for a minimal cost Virtual Arborescence $\mathcal{T}_G=(\Vtree,\Etree,r,\fPath)$  satisfying the following conditions:
\begin{description}
\item[\quad (\namedlabel{def:CVSAP1}{\textnormal{CVSAP-1}})] $\{r\} \cup T \subseteq \Vtree$ and  $ \Vtree \subseteq \{r\} \cup S \cup T$,
\item[\quad (\namedlabel{def:CVSAP2}{\textnormal{CVSAP-2}})] for all $t \in T$ holds  $\delta^+_{\Etree}(t)+\delta^-_{\Etree}(t) = 1$,
\item[\quad (\namedlabel{def:CVSAP3}{\textnormal{CVSAP-3}})] for the root $\delta^+_{\Etree}(r)+\delta^-_{\Etree}(r) \leq \uR$ holds,
\item[\quad (\namedlabel{def:CVSAP4}{\textnormal{CVSAP-4}})]  for all $s \in S\cap \Vtree$ holds $\delta^-_{\Etree}(s) + \delta^+_{\Etree}(s) \leq \uS(s) + 1$ and
\item[\quad (\namedlabel{def:CVSAP5}{\textnormal{CVSAP-5}})] for all $e\in \EG$ holds $|\left(\fPath(\Etree)\right)[e]| \leq \uE(e)$.
\end{description}
Any VA $\Tree$ satisfying \ref{def:CVSAP1} - \ref{def:CVSAP5} is said to be a feasible solution.
The cost of a Virtual Arborescence is defined to be
\begin{align*}
\phantomsection
 \cCVSAP(\mathcal{T}_G) = \sum\limits_{e\in \EG} \cE(e) \cdot |\left(\fPath(\Etree)\right)[e]| + \sum\limits_{s \in S \cap \Vtree} \cS(s) ~.
\end{align*}
where $|\left(\fPath(\Etree)\right)[e]|$ is the number of times an edge is used in different paths \hfill $\qed$
\end{definition}

In the above definition, \ref{def:CVSAP1} states that terminals and the root must be included in $\Vtree$, whereas non Steiner sites are excluded. We identify $\Vtree \setminus (\{r\} \cup T)$ with the set of active Steiner nodes. Condition \ref{def:CVSAP2} states that terminals must be leaves in $\mathcal{T}_G$ and \ref{def:CVSAP3} and \ref{def:CVSAP4} enforce degree constraints in $\mathcal{T}_G$. The term $\fPath(\Etree)$ in Condition \ref{def:CVSAP5} determines the set of all used paths and consequently $\fPath(\Etree)[e]$ yields the set of paths that use $e \in \Etree$. As $\fPath$ is injective and maps on simple paths, Condition \ref{def:CVSAP5} enforces that edge capacities are not violated.

\begin{definition}[Multicast / Aggregation CVSAP] 
\label{def:A-M-CVSAP}
The CVSAP constitutes the \emph{Constrained Virtual Steiner Aggregation Problem (A-CVSAP)} in case that the edges of $\Etree$ are oriented towards $r$, or conversely constitutes the \emph{Constrained Virtual Steiner Multicast Problem (M-CVSAP)} if edges are oriented away from $r$. We denote the set of feasible solutions by $\FeasibleACVSAP$ and $\FeasibleMCVSAP$ respectively.
\end{definition}
It is important to note that Definitions~\ref{def:VA} and \ref{def:A-M-CVSAP} together imply that in A-CVSAP each terminal is connected to the root and conversely that in M-CVSAP the root is connected to all terminals. Furthermore, the following remark establishes that the degree constraints (see \ref{def:CVSAP3} and \ref{def:CVSAP4}) effectively constrain the number of incoming connections (A-CVSAP) and respectively constrains the number of outgoing connections (M-CVSAP).

\begin{remark} 
\label{rem:alternativeCVSAP2}
Note that in case of M-CVSAP the Condition~\ref{def:CVSAP2} amd \ref{def:CVSAP3} reduce to $\forall s \in S \cap V_{\mathcal{T}}. \delta^+_{\Etree}(s) \leq \uS(s)$ and  $\delta^+_{\Etree}(r) \leq \uR$ respectively. Analogously, Conditions~\ref{def:CVSAP2} and \ref{def:CVSAP3} simplify to $\forall s \in S \cap V_{\mathcal{T}}. \delta^-_{\Etree}(s) \leq \uS(s)$ and $\delta^-_{\Etree}(r) \leq \uR$ in A-CVSAP.
\end{remark}

Obviously, as the single difference between A-CVSAP and M-CVSAP is the orientation of the edges, the problems can be reduced on each other.
\begin{remark}[Equivalence of A-CVSAP and M-CVSAP]
The problems A-CVSAP and M-CVSAP can be reduced on each other by inverting the orientation of the edges in the underlying network $G$ and adapting $\cE$ and $\uE$ accordingly.
\end{remark}
Based on the above fact, it suffices to give an algorithm for either one of the two problems. As our IP approach presented in Section~\ref{sec:ilp-formulation} can be more intuitively understood in the aggregation scenario, we restrict our further discussion to the case of A-CVSAP.
Lastly, we give the following remark relating the directed CVSAP to its undirected counterpart, namely the Constrained Virtual Steiner Tree Problem.
\begin{remark}[Constrained Virtual Steiner Tree Problem]
Analogously to the definition of VA the concept of a (rooted) Virtual Tree can be introduced, in which (undirected) virtual edges are mapped on undirected simple paths (see \ref{def:VA1}), and the orientation constraint (see \ref{def:VA2}) is dropped. Substituting $\delta^+_{\Etree}(\cdot)+\delta^-_{\Etree}(\cdot)$ by its undirected counterpart $\delta_{\Etree}(\cdot)$ in  Definition~\ref{def:CVSAP} of CVSAP then directly translates to the definition of the Constrained Virtual Steiner Tree Problem (CVSTP).
\end{remark}

The following theorem motivates our approach in Section~\ref{sec:ilp-formulation}, namely to search for provably good solutions in non-polynomial time.

\begin{theorem}
Checking whether a feasible solution for CVSAP exists is NP-complete. Thus, unless $\mathit{NP} \subseteq \mathit{P}$ holds, there cannot exist an (approximation) algorithm yielding a feasible solution in polynomial time.
\end{theorem}
\begin{proof}
We give a reduction on the decision variant of set cover. Let $U$ denote the universe of elements and let $\mathcal{S}\subseteq 2^U$ denote a family of sets covering $U$. To check whether a set cover using at most $k$ many sets exists, we construct the following CVSAP instance. We introduce a terminal $t_u$ for each element $u \in U$ and a Steiner site $s_S$ for each $S \in \mathcal{S}$. A terminal $t_u$ is connected by a directed link to each Steiner site $s_S$ iff. $u \in S$. Each Steiner site $s_S$ is connected to the root $r$. We set the capacity of the root to $k$ and capacities of Steiner sites to $|U|$. It is easy to check that there exists a feasible solution to this CVSAP instance iff. there exists a set cover of less than $k$ elements.
\end{proof}

\section{VirtuCast Algorithm} 
\label{sec:ilp-formulation}

In this section we present the Algorithm VirtuCast to solve CVSAP. VirtuCast first computes a solution for a single-commodity flow Integer Programming formulation and then constructs the corresponding Virtual Arborescence. Even though our IP formulation can be used to compute the optimal solution for any CVSAP instance, feasible solutions to our IP formulation already yield feasible solutions to CVSAP. This allows to derive near-optimal solutions \emph{during} the solution process.

\subsection{The IP Model} 
\label{sec:ip-model}

 Our IP (see~\ref{alg:MIP}) is  based on an \emph{extended graph} containing a single super source $\sSource$ and two distinct super sinks $\sSinkS$ and $\sSinkR$ (see Definition~\ref{def:extended-graph}). While $\sSinkR$  may only receive flow from the root $r$, all possible Steiner sites $s \in S$ connect to $\sSinkS$. Distinguishing between these two super sinks is necessary, as we will require activated Steiner nodes to not \emph{absorb all} incoming flow, but forward at least one unit of flow towards $\sSinkR$, which will indeed ensure connectivity.

\begin{definition}[Extended Graph]
\label{def:extended-graph}
Given a directed network $G = (\VG,\EG,\cE,\uE)$ and a request $R_G=(r,S,T,\uR,\cS,\uS)$ as introduced in Section~\ref{sec:problem-definition}, we define the \emph{extended graph} $\Gext = (\Vext,\Eext)$ as follows
\begin{description}
\item[\quad \namedlabel{def:EXT1}{\textnormal{(EXT-1)}}] $\Vext \triangleq \VG \cup \{ \sSource, \sSinkS, \sSinkR \}~,$ 
\item[\quad \namedlabel{def:EXT2}{\textnormal{(EXT-2)}}] $\Eext \triangleq \EG \cup \{(r,\sSinkR)\} \cup \EextSm \cup \EextSp \cup \EextTp ~,$
\end{description}
where $\EextSm \triangleq S \times \{\sSinkS\}$, $\EextSp \triangleq \{ \sSource \} \times \sS$ and $\EextTp \triangleq \{ \sSource \} \times \sT$. We define $\Enagg \triangleq \Eext \setminus \EextSm$. \hfill $\qed$
\end{definition}

\subsubsection{Further Notation.}
To clearly distinguish between variables and constants, we typeset constants in bold font: instead of referring to $\cE,\cS$ and $\uE,\uR,\uS$ we use $\cc_\foM{y}$ and $\cu_\foM{y}$, where $\foM{y}$ may either refer to an edge or a Steiner site. Similarly, we use $\cu_\foM{y}$ where $\foM{y}$ may either refer to an edge, the root or Steiner node. We abbreviate $\sum_{\foM{y}\in \foM{Y}}f_\foM{y}$ by $f(\foM{Y})$. We use $\foM{Y} + \foM{y}$ to denote $\foM{Y} \cup \{\foM{y}\}$ and $\foM{Y} - \foM{y}$ to denote $\foM{Y} \setminus \{\foM{y}\}$ for a set $\foM{Y}$ and a singleton $\foM{y}$.
For $\f \in \pint^{\Eext}$ we define the flow-carrying subgraph $\Gext^{\f} \triangleq (\fVext,\fEext)$ with $\fVext \triangleq \Vext$ and $\fEext \triangleq {\{ e | e \in \Eext \wedge \f(e) \geq 1\}}$.\\

The IP formulation~\ref{alg:MIP} uses an integral single-commodity flow and we define a flow variable $\f_e \in \mathbb{Z}_{\geq 0}$ for each edge $e \in \Eext$ in the extended graph (see~\ref{IP:VarF}).  As we use an aggregated flow formulation, that does not model routing decisions explicitly, we show in Section~\ref{sec:flow-decomposition} how this single-commodity flow can be decomposed into paths for constructing an actual solution for CVSAP.

Whether a Steiner site $s \in \sS$ is activated is decided by the binary variable $\xS_s \in \{0,1\}$ (see~\ref{IP:VarX}). Constraint~\ref{IP:FlowToTerminals} forces each terminal $t \in \sT$ to send a single unit of flow. As flow conservation is enforced on all original nodes $v \in \VG$ (see~\ref{IP:FlowConservation}), all flow originating at $\sSource$ must be forwarded to one of the super sinks $\sSinkR$ or $\sSinkS$, while not violating link capacities (see~\ref{IP:CapEdge}).

%printAndIncrement

\newcommand{\tagIt}{\refstepcounter{IPnumber}\tag{IP-\theIPnumber}}
\newcommand{\tagItStar}{\refstepcounter{IPnumber}\tag{IP-$\textnormal{\theIPnumber}^\star$}}

\newcounter{IPnumber}

\setlength{\algomargin}{8pt}

%for algorithmic
%\floatname{algorithm}{Integer Program}
\begin{algorithm}[h!]
\SetAlgoRefName{IP-A-CVSAP}
\SetAlgoCaptionSeparator{}
\SetAlgorithmName{Integer Program}{}{{}}
%for algorithmic
	%\renewcommand{\thealgorithm}{IP-A-CVSAP}

\newcommand{\spaceIt}{\qquad\quad\quad}
\newcommand{\miniSpace}{\hspace{1.5pt}}

\caption{}
\label{alg:MIP}
\BlankLine
%\centering
\hspace{-30pt}
\begin{fleqn}[0pt]
\begin{alignat}{5}
\phantomsection	 \textnormal{minimize~}   \quad & \hspace{12pt}\cIP(\xS,\f) &\miniSpace=\miniSpace&\sum \limits_{e \in \EG} \cc_e \f_e + \sum \limits_{s \in S} \cc_s \xS_s &  \label{IP:obj} \tag{IP-OBJ} \\
\phantomsection   \textnormal{subject to} \quad & ~~\f(\deltaOut{\Eext}(v))  &\miniSpace=\miniSpace& \f (\deltaIn{\Eext}(v))	&	\quad \forall~ v \in \VG \label{IP:FlowConservation} \tagIt \\
\phantomsection                       \quad & \f(\deltaOut{\Enagg} (W)) &\miniSpace\geq\miniSpace&\xS_s & \quad \forall~ W \subseteq \VG, s \in W \cap S \neq \emptyset \label{IP:CutSteiner} \tagIt \\
\phantomsection                       \quad & \f(\deltaOut{\Enagg} (W)) &\miniSpace\geq\miniSpace&1 & \quad \forall~ W\subseteq \VG, T\cap W \neq \emptyset \label{IP:CutTerminal} \tagItStar \\
\phantomsection                       \quad & \spaceIt \f_e &\miniSpace\geq\miniSpace& \xS_s  & \quad \forall~ e=(s,\sSinkS) \in \EextSm \label{IP:EnforceAbsorption} \tagItStar \\
\phantomsection                       \quad & \spaceIt \f_e &\miniSpace\leq\miniSpace& \cu_s \xS_s & \quad \forall~ e=(s,\sSinkS) \in \EextSm \label{IP:CapSteiner} \tagIt \\
\phantomsection                       \quad & \qquad~ \f_{(r,\sSinkR)} &\miniSpace\leq\miniSpace& \cu_r &  \label{IP:CapRoot} \tagIt \\
\phantomsection                       \quad & \spaceIt \f_e &\miniSpace\leq\miniSpace& \cu_e & \quad \forall~ e \in \EG \label{IP:CapEdge} \tagIt \\
\phantomsection                       \quad & \spaceIt \f_e &\miniSpace=\miniSpace& 1  		& \quad \forall~ e \in \EextTp \label{IP:FlowToTerminals} \tagIt \\
\phantomsection                       \quad & \spaceIt \f_e &\miniSpace=\miniSpace& \xS_s 		& \quad \forall~ e=(\sSource, s) \in \EextSp \label{IP:FlowToSteiner} \tagIt \\
\phantomsection                       \quad & \spaceIt \xS_s &\miniSpace \in \miniSpace&\{0,1\} &\forall~ s \in S \label{IP:VarX} \tagIt \\
\phantomsection                       \quad & \spaceIt \f_e &\miniSpace \in \miniSpace& \pint &\forall~ e \in \Eext \label{IP:VarF} \tagIt
\end{alignat}
\end{fleqn}
\SetAlgoRefName{IP-A-CVSAP}
\end{algorithm}

As the definition of A-CVSAP requires that each terminal $t\in \sT$ establishes a path to $r$, we need to enforce connectivity; otherwise active Steiner nodes would simply absorb flow by directing it towards $\sSinkS$. To prohibit this, we adopt well-known \emph{Connectivity Inequalities}~\cite{lucena2004strong} and \emph{Directed Steiner Cuts}~\cite{koch1998solving}. Our Connectivity Inequalities~\ref{IP:CutSteiner} state that each set of nodes containing a Steiner site $s \in S$ must emit at least one unit of flow in $\Enagg$, if $s$ is activated. As $\Enagg$ does not contain edges towards $\sSinkS$, this constraint therefore enforces that there exists a path in $\Gext^{\f}$ from each activated Steiner node $s$ to the root $r$.

Analogously, the Directed Steiner Cuts~\ref{IP:CutTerminal} enforce that there exists a path from each terminal $t \in \sT$ towards $r$ in $\Gext^{f}$. These directed Steiner cuts constitute valid inequalities which are implied by \ref{IP:FlowConservation}  and \ref{IP:CutSteiner} (see Lemma~\ref{lem:chooseP-is-correct}). These Directed Steiner Cuts can strengthen the model by improving the LP-relaxation during the branch-and-cut process (see Lemma~\ref{lem:IP3StrengthensFormulation} in Section~\ref{sec:appendix:proofs} for the proof). As they are not needed for proving the correctness and could technically be removed, we mark them with a  $~^\star$ (star).

As a Steiner node $s \in \sS$ is activated iff.~$\xS_s = 1$, Constraint~\ref{IP:FlowToSteiner} requires activated Steiner nodes to receive one unit of flow while being able to maximally absorb $\cu_s$ many units of flow by forwarding it to $\sSinkS$ (see~\ref{IP:CapSteiner}). Furthermore, by~\ref{IP:CapSteiner} inactive Steiner sites may not absorb flow at all. The Constraint~\ref{IP:EnforceAbsorption} requires active Steiner nodes to at least absorb one unit of flow. This is a valid inequality, as activating a Steiner site $s\in S$ incurs non-negative costs. We introduce this constraint here, as it specifies a condition that is used in a proof later on. Constraint~\ref{IP:CapRoot} defines an upper bound on the amount of flow that the root may receive and the objective function~\ref{IP:obj} mirrors the CVSAP cost function (see Definition~\ref{def:CVSAP}). 
We denote with $\FeasibleIP = \{(\xS, \f) \in \{0,1\}^S \times \pint^{\Eext} | \textnormal{\ref{IP:FlowConservation} -~\ref{IP:VarF}}\}$ the set of feasible solutions to~\ref{alg:MIP}.

\subsection{Flow Decomposition} 
\label{sec:flow-decomposition}

Given a feasible solution $(\solX,\solf)\in \FeasibleIP$ for~\ref{alg:MIP}, Algorithm~\ref{alg:decompose} constructs a feasible solution $\solTree \in \FeasibleACVSAP$ for CVSAP. Similarly to well-known algorithms for computing flow decompositions for simple s-t flows (see e.g.~\cite{ahujaNetworkFlows}), our algorithm iteratively deconstructs the flow into paths from the super source $\sSource$ to the super sinks $\sSinkS$ or $\sSinkR$, which are successively removed from the network. However, as~\ref{alg:MIP} does not pose a simple flow problem, we constantly need to ensure that Connectivity Inequalities~\ref{IP:CutSteiner} hold after removing flow in $\Gext^{\solf}$. We first present~\ref{alg:decompose} in more detail and prove its correctness. A short runtime analysis is contained in Section~\ref{sec:runtime-analysis}.

\subsubsection{Synopsis of Algorithm.}
Algorithm~\ref{alg:decompose} constructs a feasible VA $\solTree$ given a solution $(\solX,\solf)\in \FeasibleIP$. In Line~\ref{alg:decomp:initVA}, $\solTree$ is initialized without any edges but containing all the nodes the final solution will consist of, namely the root $r$, the terminals $T$ and the activated Steiner nodes $\{s\in \sS|\xS_s \geq 1\}$.
In Line~\ref{alg:decomp:beginWhile} a terminal node $t \in \solT$ is selected for which a path is constructed to either an active Steiner node or to the root itself (Lines~\ref{alg:decomp:chooseP}-\ref{alg:decomp:endIfSteinerCut}). In Line~\ref{alg:decomp:chooseP} a path $P$, connecting $t$ to the root $r$ in the flow network $\Gext^{\solf}$, is chosen (see Lemma~\ref{lem:chooseP-is-correct} for the proof of existence for such a path). Note that by definition of $\Gext^{\solf}$ all edges contained in $P$ carry at least one unit of flow.
Within the loop beginning in Line~\ref{alg:decomp:beginFor}, the flow on path $P$ is iteratively decremented (see  Line~\ref{alg:decomp:reduceFlowP}) as long as the Connectivity Inequality~\ref{IP:CutSteiner} is not violated. In case it is violated, we revert the reduction of flow (see Line~\ref{alg:decomp:repairFlow}) and select a path towards the super sink $\sSinkS$ starting at the current node $P_j$ (see Line~\ref{alg:decomp:choosePrimedP}). Such a path must exist according to Lemma~\ref{lem:IPCutSteinerInvariant}. The path $P$ is accordingly redirected in Line~\ref{alg:decomp:extendPath} .

The path construction (in Lines~\ref{alg:decomp:chooseP}~to~\ref{alg:decomp:extendPath}) terminates once the flow from the second last node $P_{|P|-1}$ towards the last node $P_{|P|}$ has been reduced. By construction, the path $P$ leads from the super source $\sSource$ via the terminal $t \in \solT$ towards the super sink $\sSinkR$ or $\sSinkS$. If $P$ terminates in $\sSinkS$ via Steiner node $s = P_{|P|-1} \in \solS$ such that $(s,\sSinkS)$ carries no flow anymore, $s$ itself becomes a terminal (see Lines~\ref{alg:decomp:ifFlowReachedS} and~\ref{alg:decomp:SbecomesT}). Otherwise, $P$ terminates in $\sSinkR$ and $P_{|P|-1}=r$ holds. Lastly, in Line~\ref{alg:decomp:includeInTree} the (virtual) edge $(t,P_{|P|-1})$ is added to $\solEtree$ and $\solfPath(t,P_{|P|-1})$ is set accordingly to the truncated path $P$, where head, tail and any cycles are removed (function \texttt{simplify}).

\newcommand{\LET}{\textbf{let~}}
\newcommand{\CHOOSE}{\textbf{choose~}}
\newcommand{\SET}{\textbf{set~}}
\newcommand{\AND}{\textbf{and~}}
\newcommand{\ST}{\textbf{such~that~}}
\newcommand{\WHERE}{\textbf{where~}}

\begin{algorithm}[h!]

\newcommand{\nX}{}
\LinesNumbered
\DontPrintSemicolon
\SetAlgoRefName{Decompose}
\SetAlgoCaptionSeparator{}
\SetAlgorithmName{Algorithm}{}{{}}
\SetKwInOut{Input}{Input}\SetKwInOut{Output}{Output}
\SetKwFunction{Simplify}{simplify}

\Input{Network $\G = (\VG,\EG,\cE,\uE)$, Request $R_\G = (r,\sS,\sT,\uR,\cS,\uS)$, \\~Solution $(\solX, \solf)\in \FeasibleIP$ to~\ref{alg:MIP}}
\Output{Feasible Virtual Arborescence $\solTree$ for CVSAP}
\SetSideCommentRight
%\Function{Decompose}{Network $\G$, Request $R_G$, Solution $(\hat{x},\hat{f})\in \{0,1\}^S \times \pint^{\Eext}$}
%\Input{$(\hat{x},\hat{f})$ satisfies~\ref{IP:FlowConservation} -~\ref{IP:VarF}}
%\Output{$\solTree$ solves A-CVSAP}
    \BlankLine
    %\Begin{
\nX      \SET $\solS \triangleq \{s \in \sS | \xS_s \geq 1\}$ \AND $\solT \triangleq \sT$
																						\label{alg:decomp:initSHatTHat}\;			
																						
\nX      \SET $\solTree \triangleq (\solVtree,\solEtree,r,\solfPath)$ \WHERE $\solVtree \triangleq \{r\} \cup \solS \cup \solT,~ \solEtree \triangleq \emptyset$ \AND $\solfPath: \solEtree \to \mathcal{P}_G$			\label{alg:decomp:initVA}\;

\nX      \While{$\solT \neq  \emptyset$}{												
																						\label{alg:decomp:beginWhile}
\nX      	\LET $t \in \solT$ \AND $\solT \gets \solT - t$														
																						\label{alg:decomp:chooseTerminal} \;
																						
\nX      	 \CHOOSE $P \triangleq \langle \sSource, t, \dots, \sSinkR \rangle \in \Gext^{\solf}$  															   			 \label{alg:decomp:chooseP}\;
\nX      	 \For{$j = 1$ \KwTo $|P|-1$}{												
																						\label{alg:decomp:beginFor}
\nX      		\SET $\solf(P_j,P_{j+1}) \gets \solf(P_j,P_{j+1}) - 1$						
																						\label{alg:decomp:reduceFlowP} \;
																						
\nX      		\If{ \textnormal{Constraint~\ref{IP:CutSteiner}} \textnormal{is violated} with respect to $\solf$ \textnormal{and} $\solS$}{ 																 \label{alg:decomp:checkIP2}

\nX      			\CHOOSE $W \subseteq \VG$ \ST $ W \cap \solS \neq \emptyset$ \AND $\solf(\deltaOut{\Enagg} (W)) = 0$
      																					\label{alg:decomp:chooseW}\;
      																					
\nX      			\CHOOSE $P' \triangleq \langle P_j,\dots ,\sSinkS \rangle \in \Gext^{\solf}$ \ST $P_i \in W$ for $1\leq i<m$  																			 \label{alg:decomp:choosePrimedP}\;
      																		
\nX      			\SET $\solf(P_j,P_{j+1}) \gets \solf(P_j,P_{j+1}) + 1$ \AND $\solf(P'_1,P'_2) \gets \solf(P'_1,P'_2) - 1$
																						\label{alg:decomp:repairFlow}\;
																						
\nX     			\SET $P \gets \langle P_1,\dots,P_{j-1},P_j=P'_1,P'_2,\dots,P'_m \rangle$
																						\label{alg:decomp:extendPath}\;
      		} \label{alg:decomp:endIfSteinerCut}
      	} \label{alg:decomp:endFor}
\nX     	\If{$P_{|P|} = \sSinkS$ \textnormal{\AND} $\solf(P_{|P|-1},P_{|P|})=0$}{
																						\label{alg:decomp:ifFlowReachedS}
																						
\nX     		\SET $\solS \gets \solS - P_{|P|-1}$ \AND $\solX(P_{|P|-1}) \gets 0$ \AND $\solT \gets \solT + P_{|P|-1}$
																						\label{alg:decomp:SbecomesT}\;
     	} \label{alg:decomp:endIfFlowReachedS}
\nX      	\SET $\solEtree \gets \solEtree + (t,P_{|P|-1})$ \AND $\solfPath(t,P_{|P|-1}) \triangleq $~\Simplify{$\langle P_2, \dots, P_{|P|-1}\rangle$}
																						\label{alg:decomp:includeInTree}\;
     }\label{alg:decomp:endWhile}
     %\KwRet{$\solTree$}
    %}
%\EndFunction
\caption{}
\label{alg:decompose}
\end{algorithm}

\subsubsection{Proof of Correctness} 
\label{sec:proof}
We will now formally prove the correctness of Algorithm~\ref{alg:decompose}, thereby showing that \ref{alg:MIP} can be used to compute (optimal) solutions to CVSAP. We use an inductive argument similar to the one used for proving the existence of flow decompositions (see~\cite{ahujaNetworkFlows}). we assume that all constraints of \ref{alg:MIP} hold and show that for any terminal $t \in \sT$ a path towards the root or to an active Steiner node can be constructed, such that decrementing the flow along the path by one unit does again yield a feasible solution to \ref{alg:MIP}, in which $t$ has been removed from the set of terminals (see Theorem~\ref{thm:induction-step} below). During the course of this induction,  the well-definedness of the \CHOOSE operations is shown.

\begin{theorem}
\label{thm:induction-step}
Assuming that the constraints of \ref{alg:decompose} hold with respect to $\solS,\solT,\solf,\solX$ before executing Line~\ref{alg:decomp:chooseTerminal}, then the constraints of \ref{alg:decompose} will also hold in Line~\ref{alg:decomp:includeInTree} with respect to then reduced problem $\solS,\solT,\solf,\solX$. 
\end{theorem}

To prove the above theorem, we use the following Lemmas \ref{lem:chooseP-is-correct} through \ref{lem:IPCutSteinerInvariant}.

\begin{lemma}
\label{lem:chooseP-is-correct}
Assuming that \ref{IP:FlowConservation} and \ref{IP:CutSteiner} hold, there exists a path $P = \langle \sSource, t, \dots , \sSinkR \rangle \in \Gext^{\solf}$ in Line~\ref{alg:decomp:chooseP}.
\end{lemma}
\begin{proof}
Note that initially (i.e. in Line~\ref{alg:decomp:initSHatTHat}) $\solf (\sSource, v) = 1$ holds for $v \in \solS \cup \solT$ by  \ref{IP:FlowToTerminals} and \ref{IP:FlowToSteiner}. This flow will only be reduced once, as a node $t \in \solT$ will only be handled once when it is removed from $\solT$ in Line~\ref{alg:decomp:chooseTerminal}, and similarly, a node $s \in \solS$ will only be moved once into $\solT$ in Line~\ref{alg:decomp:SbecomesT}. By flow conservation (see \ref{IP:FlowConservation}), there must exist a path from $t$ to either $\sSinkR$ or $\sSinkS$. However, as we assume \ref{IP:CutSteiner} to hold, there exists a path from each $s \in \solS$ to $\sSinkR$ and we conclude that such a path $P = \langle \sSource, t, \dots , \sSinkR \rangle \in \Gext^{\solf}$ must exist.
\end{proof}

\begin{lemma}
\label{lem:limited-flow-conservation}
Assuming that \ref{IP:FlowConservation} has held in Line~\ref{alg:decomp:chooseP},  $\f(\deltaOut{\Eext}(v)) - \f(\deltaIn{\Eext}(v)) = \delta_{v,P_{j+1}}$ holds for all  $v \in \VG$ during construction of $P$ (Lines~\ref{alg:decomp:checkIP2}-\ref{alg:decomp:endIfSteinerCut}), where $\delta_{x,y} \in \{0,1\}$ and $\delta_{x,y} = 1$ iff.~ $x=y$.
\end{lemma}
\begin{proof}
We prove this statement by an inductive argument assuming for now that \CHOOSE operations in Lines~\ref{alg:decomp:chooseW} and \ref{alg:decomp:choosePrimedP} are well-defined.

After the first execution of Line~\ref{alg:decomp:reduceFlowP}, $\f(\deltaOut{\Eext}(P_2=t)) - \f(\deltaIn{\Eext}(P_2=t)) = 1$ holds, while for no other node $v \in \VG$ flow on adjacent edges were changed, and therefore  $\f(\deltaOut{\Eext}(v)) - \f(\deltaIn{\Eext}(v)) = 1$ holds.
Furthermore, the reduction of flow on edge $(\sSource,P_2=t)$ cannot violate \ref{IP:CutSteiner}, such that our claim holds until Line~\ref{alg:decomp:endIfSteinerCut} and therefore for the base case $j=1$.

Assuming that $\f(\deltaOut{\Eext}(v)) - \f(\deltaIn{\Eext}(v)) = \delta_{v,P_{j+1}}$ has held for $j = n$, it is easy to check that it will continue to hold for $j' = n+1$, as either in Line~\ref{alg:decomp:reduceFlowP} or in Line~\ref{alg:decomp:repairFlow} the outgoing flow from node $P_{j'}$ towards node $P_{j'+1}$ is reduced such that $\f(\deltaOut{\Eext}(v)) - \f(\deltaIn{\Eext}(v)) = \delta_{v,P_{j'+1}}$ indeed holds for all $v \in \VG$.
\end{proof}

\newcommand{\tagSubLem}{\refstepcounter{IPCutSteinerInvariantCounter}{\alph{IPCutSteinerInvariantCounter}}}
\newcounter{IPCutSteinerInvariantCounter}

\begin{lemma}
\label{lem:IPCutSteinerInvariant}
Assuming that connectivity inequalities \ref{IP:CutSteiner} have held before executing Line~\ref{alg:decomp:reduceFlowP}, these inequalities will hold again at Line~\ref{alg:decomp:endIfSteinerCut}.
\end{lemma}
\begin{proof}
We only have to consider the case in which the Constraint~\ref{IP:CutSteiner} was violated after executing Line~\ref{alg:decomp:reduceFlowP}. Assume therefore that \ref{IP:CutSteiner} is violated in Line~\ref{alg:decomp:checkIP2}. The \CHOOSE operation in Line~\ref{alg:decomp:chooseW} is well-defined, as \ref{IP:CutSteiner} is violated. Let $W \subseteq \VG$ be any violated set with $\solS \cap W \neq \emptyset$. To prove this lemma, we prove the following four statements:
\begin{description}
\item[\quad \namedlabel{lem:IPCutSubA}{\textnormal{(a)}}] $P_j$ is contained in $W$ while $P_{j+1}$ is not contained in $W$.

\item[\quad \namedlabel{lem:IPCutSubB}{\textnormal{(b)}}]
$\solf(P_j,P_{j+1})=0$ holds in Lines~\ref{alg:decomp:chooseW}-\ref{alg:decomp:choosePrimedP}.

\item[\quad \namedlabel{lem:IPCutSubC}{\textnormal{(c)}}]
Before flow reduction in Line~\ref{alg:decomp:reduceFlowP}, there existed a path \\$P'' = \langle s, \dots, P_j,P_{j+1}, \dots, \sSinkR \rangle \in \Gext^{\solf}$ for $s \in \solS \cap W$.

\item[\quad \namedlabel{lem:IPCutSubD}{\textnormal{(d)}}]
There exists a path $P' = \langle P_j,\dots ,\sSinkS \rangle$ with $P'_i \in W$ for $1\leq i<|P'|$ in $\Gext^{\solf}$.

\end{description}

Considering \ref{lem:IPCutSubA}, note that edge $(P_j,P_{j+1})$ is by definition only included in $\delta^+_{\Eext}(W)$ if $P_j\in W$ and $P_j \notin W$. Thus, assuming that either $P_j$ is not contained in $W$ or assuming that $P_{j+1}$ is contained in $W$, we can conclude that edge $(P_j,P_{j+1})$ is not contained in $\delta^+_{\Eext}(W)$. However, in this case the connectivity inequality \ref{IP:CutSteiner} must have been violated even before flow was reduced. This contradicts our assumption that connectivity inequalities \ref{IP:CutSteiner} have held beforehand, therefore proving \ref{lem:IPCutSubA}.

The correctness of \ref{lem:IPCutSubB} directly follows from \ref{lem:IPCutSubA}, as by \ref{lem:IPCutSubA} $(P_j,P_{j+1})\in \delta^+_{\Eext}(W)$ holds. As $\solf(\delta^+_{\Eext}(W)) = 0$ holds by definition of $W$ and flow may not be negative, we derive the second statement.

We now prove the statement \ref{lem:IPCutSubC}. As connectivity inequalities \ref{IP:CutSteiner} are assumed to have held \emph{before} the flow reduction in Line~\ref{alg:decomp:reduceFlowP}, for each activated Steiner node $s \in \solS$ there existed a path from $s$ to $\sSinkR$ in $\Gext^{\solf}$. By the second statement, $(P_j,P_{j+1})$ is the only edge in $\Gext^{\solf}$ leaving $W$ showing that indeed a path
$P'' = \langle s, \dots, v, P_j,P_{j+1}, \dots, \sSinkR \rangle \in \Gext^{\solf}$ for $s \in \solS$ existed \emph{before} reduction of flow on $(P_j,P_{j+1})$.

By statement \ref{lem:IPCutSubC}, the prefix $\langle s,  \dots , P_j\rangle$ of path $P''$ still exists in $\Gext^{\solf}$ inducing that $P_j$ is reached by a positive flow. By Lemma~\ref{lem:limited-flow-conservation} flow conservation holds for all nodes $w \in W$, since by statement \ref{lem:IPCutSubA} $P_{j+1}$ is not included in $W$. As $\sSinkR$ is not included in $W$, there must exist a path $P' = \langle P_j,\dots ,\sSinkS \rangle \in \Gext^{\solf}$ with $P_i \in W$ for $1\leq i<m$. This shows the fourth statement \ref{lem:IPCutSubD} and shows that the \CHOOSE operation in Line~\ref{alg:decomp:choosePrimedP} is well-defined.

We will now prove the main statement of this lemma, namely that in Line~\ref{alg:decomp:endIfSteinerCut} the connectivity inequalities \ref{IP:CutSteiner} hold (again).
In Line~\ref{alg:decomp:repairFlow}, the flow along edge $(P_j,P_{j+1})$ is incremented again.
Assume for the sake of contradiction, that the reduction of flow along $(P'_1,P'_2)$ violates a connectivity inequality with node set $W'$ such that $\solf(\delta^+_{\Eext}(W')) = 0$ holds. By the same argument as used for proving statement~\ref{lem:IPCutSubA}, it is easy to see that $P'_1 \in W'$ and $P'_2 \notin W'$ must hold. However, by statement~\ref{lem:IPCutSubC}, after having reverted the flow reduction along $(P_j,P_{j+1})$, the path $\langle P_j,P_{j+1}, \dots, \sSinkR \rangle$ was re-established in $\Gext^{\solf}$. As flow along any of the edges contained in this path is greater or equal to one, $W'$ cannot possibly violate \ref{IP:CutSteiner} and contain $P_j \in W'$ as the super sink for the root $\sSinkR \notin W' \subseteq \VG$ may never be contained in $W'$.
\end{proof}

Using the above lemma, we can now prove Theorem~\ref{thm:induction-step}.

\begin{deferredproof}[Theorem \ref{thm:induction-step}]
Assume that the constraints of \ref{alg:MIP} hold with respect to $\solS,\solT,\solf,\solX$ before executing Line~\ref{alg:decomp:chooseTerminal}. By Lemma~\ref{lem:chooseP-is-correct} the \CHOOSE operation in Line~\ref{alg:decomp:chooseP} is well-defined as \ref{IP:FlowConservation} and \ref{IP:CutSteiner} hold by our assumption.
By Lemma~\ref{lem:IPCutSteinerInvariant} the path construction process in Lines~\ref{alg:decomp:reduceFlowP} through \ref{alg:decomp:endIfSteinerCut} is well-defined as initially \ref{IP:CutSteiner} holds. The execution of Lines~\ref{alg:decomp:chooseTerminal}-\ref{alg:decomp:includeInTree} is therefore well-defined.

To distinguish the state of the variables $\solS,\solT,\solf,\solX$ at Lines~\ref{alg:decomp:chooseTerminal} and \ref{alg:decomp:includeInTree} we will use primed variables $\solS',\solT',\solf',\solX'$ to denote the latter state. First note that \ref{IP:FlowConservation} holds by Lemma~\ref{lem:limited-flow-conservation}: As path $P$ must terminate in either $\sSinkS$ or $\sSinkR$ (see Lines~\ref{alg:decomp:chooseP},\ref{alg:decomp:choosePrimedP}),   Lemma~\ref{lem:limited-flow-conservation} reduces to $\f(\deltaOut{\Eext}(v)) - \f(\deltaIn{\Eext}(v)) = 0$ for all $v\in \VG$ for $j=|P|-1$ as neither of the super sinks are included in $\VG$.
The connectivity inequalities \ref{IP:CutSteiner} will also hold with respect to $\solS'$ and $\solf'$ as these are preserved by Lemma~\ref{lem:IPCutSteinerInvariant} and $\solS' \subseteq \solS$ holds.
Constraint~\ref{IP:FlowToSteiner} holds with respect to $\solS'$ as $\solS' \subseteq \solS$ and the flow along edges in $\EextSp$ is never reduced. As similarly flow along edges in $\EextTp$ is only reduced for the terminal being connected, Constraint~\ref{IP:FlowToTerminals} could only be violated by a node satisfying $t' \in \solT'$ but $t \notin \solT$. If such a node exists, then it must have been added in Line~\ref{alg:decomp:SbecomesT} and as \ref{IP:FlowToSteiner} has held for $\solS$, constraint~\ref{IP:CutTerminal} will hold for $t'\in \solS \cap \solT'$. Analogously, constraint~\ref{IP:CapSteiner} is not violated as setting $\solX(s)$ to zero for $s\in \solS$ implies that $s \notin \solS'$ (see Line~\ref{alg:decomp:SbecomesT}). Constraint \ref{IP:EnforceAbsorption} holds for $\solS'$ as the variable $\solX(s)$ is set to zero whenever the flow along an edge $(s,\sSinkS)$ is reduced to zero. Lastly, it is easy to observe that the capacity constraints \ref{IP:CapSteiner}, \ref{IP:CapEdge} cannot be violated as the flow is only reduced. \hfill $\qed$
\end{deferredproof}

Using Theorem~\ref{thm:induction-step} we can now prove that Algorithm~\ref{alg:decompose} terminates.

\begin{theorem}
\label{thm:termination}
Algorithm~\ref{alg:decompose} terminates.
\end{theorem}
\begin{proof}
By iteratively applying Theorem~\ref{thm:induction-step} the \CHOOSE operations of Algorithm~\ref{alg:decompose} are well-defined.
Note that by construction of the path $|P|$ (see Lines~\ref{alg:decomp:chooseP},\ref{alg:decomp:choosePrimedP}) flow variables which values are decremented must have been greater or equal to one before the reduction took place. Since the flow $\solf \in \mathbb{Z}_{\geq 0}$ is finite and is successively reduced during the process of path construction, the inner loop (see Lines~\ref{alg:decomp:beginFor}-\ref{alg:decomp:endFor}) must terminate. The outer loop must eventually terminate as well, because each node in $\solT$ (see Line~\ref{alg:decomp:chooseTerminal}) is handled exactly once and as a node $s \in \solS$ may be only moved only once into $\solT$ (see Line~\ref{alg:decomp:SbecomesT}).
\end{proof}

Using Theorem~\ref{thm:induction-step} and \ref{thm:termination} we can finally prove that Algorithm~\ref{alg:decompose} indeed constructs a feasible solution for A-CVSAP. 

\begin{theorem}
Algorithm~\ref{alg:decompose} constructs a feasible solution $\solTree \in \FeasibleACVSAP$ for A-CVSAP given a solution $(\solX,\solf)\in \FeasibleIP$. Additionally, $\cCVSAP(\solTree) \leq \cIP(\solX,\solf)$ holds.
\end{theorem}
\begin{proof}

To show that for $\solTree$ constructed by Algorithm~\ref{alg:decompose} $\solTree \in \FeasibleMCVSAP$ holds, we need to check \ref{def:CVSAP1}-\ref{def:CVSAP5} as well as \ref{def:VA1} and \ref{def:VA2}.We first give short arguments why in fact the conditions \ref{def:CVSAP1}-\ref{def:CVSAP5} hold:
\begin{longtable}{p{1.7cm}p{10.5cm}}
\ref{def:CVSAP1} & This constraint naturally holds due to Line~\ref{alg:decomp:initVA}. \\
\ref{def:CVSAP2} & Algorithm~\ref{alg:decompose} does not allow for connecting nodes to terminals. Thereby terminals are indeed leaves in $\solTree$ and \ref{def:CVSAP2} holds. \\
\ref{def:CVSAP3} & Each time another node is connected to the root $r$ the flow along $(r,\sSinkR)$ is decremented (see \ref{IP:CapRoot}). As the flow along this edge is bounded by $\uR$, the degree constraint \ref{def:CVSAP3} is satisfied by $\solTree$ \\
\ref{def:CVSAP4} & An analogue argument as for \ref{def:CVSAP3} applies. \\
\ref{def:CVSAP5} & As paths are constructed according to the flow variables $\solf$ that initially respect capacity constraints on edges \ref{IP:CapEdge}, and as $\solf$ is appropriately reduced on used edges, $\solTree$ satisfies the edge capacity constraint \ref{def:CVSAP5}.\\ 
\end{longtable} 

It remains to prove that $\solTree$ satisfies the conditions \ref{def:VA1} and \ref{def:VA2} given by in Definition~\ref{def:VA}. \ref{def:VA1} follows directly from Line~\ref{alg:decomp:includeInTree}. It remains to prove that $\solTree$ satisfies the connectivity requirements \ref{def:VA2}.

First note that $\solT = \emptyset$ holds when~\ref{alg:decompose} terminates.
We prove that $\solS = \emptyset$ equally holds, thereby showing that each node in $\solVtree \setminus \{r\}$ is connected to another node in $\solVtree$ in Line~\ref{alg:decomp:ifFlowReachedS}.
Assume that $\solS \neq \emptyset$ but $\solT = \emptyset$ holds. We show that this can never be the case using the invariant $s \in \solS \Rightarrow \solf(s,\sSinkS) \geq 1$ which directly follows from Theorem~\ref{thm:induction-step} as \ref{IP:EnforceAbsorption} holds. As this holds for all Steiner nodes, $\solf (\delta^+_{\EextSm}(\solS)) \geq |\solS|$ follows.
On the other hand, the amount of flow emitted by $\sSource$ equals $|\solS|$ as we assume $\solT = \emptyset$ to hold and by Theorem~\ref{thm:induction-step} the constraints  \ref{IP:FlowToSteiner} and \ref{IP:FlowToTerminals} must hold. Due to the flow conservation constraint \ref{IP:FlowConservation}, this implies $\solf(r,\sSinkR)\leq 0$  which immediately violates Constraint~\ref{IP:CutSteiner} by considering the node set $W=\VG$. As this contradicts the statement of Theorem~\ref{thm:induction-step}, we conclude that $\solS = \solT = \emptyset$ must hold when terminating, implying that for all included nodes (see Line~\ref{alg:decomp:initVA}) an edge was introduced in $\solEtree$ (see Line~\ref{alg:decomp:includeInTree}).

As each node (except for the root) has one outgoing edge, it remains to show that $\solTree$ does not contain cycles. This follows immediately from the order in which nodes are extracted from $\solT$. This order in fact defines a topological ordering on $\solVtree$ as a cycle containing nodes $u$ and $v$ would imply that $u$ was connected before $v$ and vice versa, that $v$ was connected before $u$. As this can never be the case, this concludes the proof that $\solTree \in \FeasibleACVSAP$ holds.

Lastly, $\cCVSAP(\solTree) \leq \cIP(\solX,\solf)$ is valid as costs associated with activating Steiner nodes are incurred in both objectives and~\ref{alg:decompose} uses only edges already accounted for in $\cIP(\solX,\solf)$. In fact $\cCVSAP(\solTree) < \cIP(\solX,\solf)$ may only be the case if the function \texttt{simplify} (see Line~\ref{alg:decomp:includeInTree}) truncated a path.
\end{proof}

To prove that our formulation \ref{alg:MIP} indeed computes an optimal solution, we need the following lemma showing that each solution to A-CVSAP can be mapped on a solution of \ref{alg:MIP} with equal cost:

\begin{lemma} 
\label{lem:A-CVSAP-feasible-implies-IP-CVSAP-feasible}
Given a network $\G = (\VG,\EG,\cE,\uE)$, a request $R_\G = (r,\sS,\sT,\uR,\cS,\uS)$ and a feasible solution $\solTree = (\solVtree,\solEtree,r,\solfPath)$ to the corresponding A-CVSAP. There exists a solution $(\solX,\solf) \in \FeasibleIP$ with $\cCVSAP(\solTree) = \cIP(\solX,\solf)$.
\end{lemma}
\begin{proof}
We define the solution $(\solX, \solf) \in \{0,1\}^S \times \pint^{\Eext}$ in the following way
\begin{itemize}
 \item $\solX_s = 1$ iff. $s \in \solVtree$ for all $s \in S$,
 \item $\solf_e \triangleq |(\solfPath(\solEtree))[e]|$ for all $e \in \EG$,
 \item $\solf_e = 1$ if $v \in \solVtree \setminus \{r\}$ for all $e=(\sSource,v) \in \EextSp \cup \EextTp$ and $\solf_e = 0$ otherwise,
 \item $\solf_e \triangleq \delta^-_{\solEtree}(s)$ for all $e = (s,\sSinkS) \in \EextSm$ and  $\solf(r,\sSinkR) \triangleq \delta^-_{\solEtree}(r)$.
\end{itemize}
Checking that $(\solX,\solf) \in \FeasibleIP$ and $\cCVSAP(\solTree) = \cIP(\solX,\solf)$ holds is straightforward.
\end{proof}

Finally, we can now prove that VirtuCast solves CVSAP to optimality.

\begin{theorem}
\label{thm:SolvingACVSAPToOptimality}
Algorithm VirtuCast, that first computes an optimal solution to \ref{alg:MIP} and then applies \ref{alg:decompose}, solves A-CVSAP to optimality.
\end{theorem}
\begin{proof}
We use \ref{alg:MIP} to compute an optimal solution $(\solX,\solf)\in \FeasibleIP$ and afterwards construct the corresponding $\solTree \in \FeasibleACVSAP$ via \ref{alg:decompose}. Assume for the sake of deriving a contradiction that $\solTree$ is not optimal and there exists $\tilde{\mathcal{T}} \in \FeasibleACVSAP$ with $\cCVSAP( \tilde{\mathcal{T}}) < \cCVSAP(\solTree)$. By Lemma~\ref{lem:A-CVSAP-feasible-implies-IP-CVSAP-feasible} any solution for A-CVSAP can be mapped on a feasible solution of \ref{alg:MIP} of the same objective value. This contradicts the optimality of $(\solX,\solf) \in \FeasibleIP$ and $\solTree$ must therefore be optimal.
\end{proof}

\subsection{Runtime Analysis for~\ref{alg:decompose}}
\label{sec:runtime-analysis}

We conclude our discussion of VirtucCast with stating that each \CHOOSE operation in \ref{alg:decompose} and checking whether connectivity inequalities \ref{IP:CutSteiner} hold can be implemented using depth-first search. Implementing \ref{alg:decompose} in this way and assuming that an optimal solution for \ref{alg:MIP} is given and that $\G$ does not contain zero-cost cycles, we can bound the runtime from above as follows.

\begin{theorem}
\label{thm:runtime}
Using depth-first search for choosing paths in Algorithm~\ref{alg:MIP} and for determining whether connectivity inequalities \ref{IP:CutSteiner} are violated, we can bound the runtime by $\mathcal{O} \left( |\VG|^2\cdot |\EG| \cdot (|\VG|+|\EG|) \right)$, given an optimal solution $(\solX, \solf) \in \FeasibleIP$ and assuming that graph $G$ does not contain zero-cost cycles.
\end{theorem}
\begin{proof}
We use depth-first search to separate the connectivity inequalities \ref{IP:CutSteiner} in a canonical manner which we only explain briefly. Given an activated Steiner node $s \in \solS$, we compute the set of all reachable nodes $R$ via depth-first search. If $\sSinkR$ is contained in $R$, then no set of nodes $W \subseteq \VG$ containing $s$ can violate \ref{IP:CutSteiner}. On the other hand, if $\sSinkR$ is not contained in $R$, then obviously $W \triangleq R$ violates \ref{IP:CutSteiner}.
Checking the connectivity inequalities in Line~\ref{alg:decomp:checkIP2} can  therefore be performed in time $\mathcal{O}(|\solS|\cdot (|\Vext|+|\Eext|))$.
The runtime for choosing a path in Line~\ref{alg:decomp:choosePrimedP} is clearly dominated by the runtime for checking the connectivity inequalities, and as the previous depth-first search provides a node set $W$, we do not consider these operations.

The length of any used path $P$ is bounded by $|\Eext|$ as otherwise $P$ would contain a cycle with positive cost. As this cycle can be removed (see function \texttt{simplify} in Line~\ref{alg:decomp:includeInTree}) yielding a better objective value while remaining feasible, this may never occur by the assumption that our solution is optimal.

 Thus, the runtime for the inner loop (Lines~\ref{alg:decomp:beginFor}-\ref{alg:decomp:endFor})  amounts to $\mathcal{O}(|\Eext|\cdot|\solS|\cdot (|\Vext|+|\Eext|))$ Lastly, the outer loop is performed at most $|\solS|+|\solT|$ many times and the runtime for choosing path $P$ in Line~\ref{alg:decomp:chooseP} is clearly dominated by the runtime of the inner loop. As $|\Eext| \in \Theta(\EG)$ (assuming $\EG$ to be connected), $|\Vext| \in \Theta(\EG)$ and $|\solS| + |\solT| \in \Theta(\VG)$ holds, the runtime of~\ref{alg:decompose} is bounded by
\[
\mathcal{O} \left( |\VG|^2\cdot |\EG| \cdot (|\VG|+|\EG|) \right)
\]
and our claim follows.
\end{proof}

\section{A Multi-Commodity Flow Formulation}
\label{sec:mcf-formulation}

This section introduces a naive multi-commodity flow (MCF) formulation (see \ref{alg:MCF-MIP}) to solve A-CVSAP. The formulation \ref{alg:MCF-MIP} models the virtual arborescence searched for rather directly, as it uniquely determines virtual links and paths for active Steiner nodes. This explicit representation comes at the price of a substantially larger model. In Section~\ref{sec:results-comparison-mcf} we provide a computational comparison showing the superiority of our compact formulation \ref{alg:MIP}.

\subsection{Notation}

For ease of representation of \ref{alg:MCF-MIP} we use a modified extended graph, which does not contain a super source but a single super sink.
\begin{definition}[Extended Graph for \ref{alg:MCF-MIP}]
\label{def:extended-graph-mcf}
Given a directed network $G = (\VG,\EG,\cE,\uE)$ and a request $R_G=(r,S,T,\uR,\cS,\uS)$ as introduced in Section~\ref{sec:problem-definition}, we define the \emph{extended graph}  $\GextM = (\VextM,\EextM)$ for the \ref{alg:MCF-MIP} formulation as follows:
\begin{description}
\item[\quad \namedlabel{def:EXT-MCF1}{\textnormal{(EXT-1-MCF)}}] $\VextM \triangleq \VG \cup \{ \sSink \}~,$ 
\item[\quad \namedlabel{def:EXT-MCF2}{\textnormal{(EXT-2-MCF)}}] $\EextM \triangleq \EG \cup \{(r,\sSink)\} \cup \EextSmM	~,$
\end{description}
where $\EextSmM \triangleq S \times \{\sSink\}$. \hfill $\qed$
\end{definition}

As already introduced in Lemma~\ref{lem:limited-flow-conservation} we use the Kronecker-Delta $\delta_{\foM{x},\foM{y}} \in \{0,1\}$, where $\delta_{\foM{x},\foM{y}} = 1$ holds iff.~ $\foM{x}=\foM{y}$. 

Flow variables corresponding to different commodities are distinguished by superscripts and we use $\f^{\foM{x}}(\foM{Y})$ to denote $\sum_{y \in \foM{Y}} f^{\foM{x}}(\foM{y})$.

We denote the set of feasible solution for \ref{alg:MCF-MIP} by $\FeasibleMCF$.

\newcommand{\tagMCF}{\refstepcounter{MCFIPnumber}\tag{{MCF-\theMCFIPnumber}}}
\newcommand{\tagMCFStar}{\refstepcounter{MCFIPnumber}\tag{{MCF-$\textnormal{\theMCFIPnumber}^\star$}}}
\newcounter{MCFIPnumber}

\begin{algorithm}[h!]
\footnotesize
\SetAlgoRefName{MIP-A-CVSAP-MCF}
\SetAlgoCaptionSeparator{}
\SetAlgorithmName{Mixed Integer Program}{}{{}}
%for algorithmic
	%\renewcommand{\thealgorithm}{IP-A-CVSAP}

\newcommand{\spaceIt}{\qquad\quad\quad}
\newcommand{\miniSpace}{\hspace{1.5pt}}

\caption{}
\label{alg:MCF-MIP}
\BlankLine
%\centering
\hspace{-30pt}
\begin{fleqn}[0pt]
\begin{alignat}{7}
\phantomsection	 \textnormal{minimize~}    && \quad \cIPM &\miniSpace=\miniSpace && \sum \limits_{e \in \EG} \cc_e (\f_e  + \sum \limits_{s \in S} \f_{s,e}) & \tag{MCF-OBJ} \label{IPM:obj} \\
										   && 			&                       && ~~~ + ~~\sum \limits_{s \in S}\cc_s \cdot \xS_s  \notag \\
\phantomsection   \textnormal{subject to}  && \quad \f^T(\deltaOut{\EextM}(v))  &\miniSpace=\miniSpace && \f^T(\deltaIn{\EextM}(v)) + |\{v\} \cap T| &  \forall~ v \in \VG \label{IPM:FlowConservationT} \tagMCF \\
\phantomsection	&& \quad \f^s(\deltaOut{\EextMS}(v))  &\miniSpace=\miniSpace && \f^s(\deltaIn{\EextMS}(v)) + \mathbf{\delta}_{s,v} \cdot \xS_s &	\forall~ s \in \sS, v \in \VG \label{IPM:FlowConservationS} \tagMCF \\
\phantomsection      &&  \f^T_e + \sum \limits_{s \in S} \f^s_e &\miniSpace\leq\miniSpace && 
\begin{cases}
\cu_s \xS_s ,~e = (s,\sSink), s \in \sS&\\
\cu_r ~~~~,~e = (r,\sSink) & \\
\cu_e ~~~~,~ e \in \EG &
\end{cases} & \forall e \in \EextM
  \label{IPM:CapAll} \tagMCF \\
 \phantomsection		 &&  -|S|(1-f^s_{\bar{s},\sSink}) &\miniSpace \leq \miniSpace && \priority_s - \priority_{\bar{s}} - 1 & \forall~ s,\bar{s} \in S \label{IPM:SetOrder} \tagMCF \\
\phantomsection		 &&  \f^s_{(\bar{s},\sSink)} &\miniSpace \leq \miniSpace && \xS_{\bar{s}} & \hspace{-24pt} \forall~ s \in S,\bar{s} \in S - s \label{IPM:IncomingFlowActivatesFlow} \tagMCFStar \\
\phantomsection		 &&  f^s_{s,\sSink} &\miniSpace = \miniSpace && 0 & \forall~ s \in S \label{IPM:DisableSendingFlowToOneself} \tagMCFStar \\
\phantomsection		 &&  f^s_{\bar{s},\sSink} + f^{\bar{s}}_{s,\sSink} &\miniSpace \leq \miniSpace && 1 & \forall~ s,\bar{s} \in S \label{IPM:SteinerPairsDoNotSendToEachOther} \tagMCFStar \\
\phantomsection		 &&  \xS_s &\miniSpace \in \miniSpace && \{0,1\} &\forall~ s \in S \label{IPM:VarX} \tagMCF \\
\phantomsection		 &&  \f^T_{e} &\miniSpace \in \miniSpace && \pint &\forall~ e \in \EextM \label{IPM:VarFT} \tagMCF \\
\phantomsection		 &&  \f^s_{e} &\miniSpace \in \miniSpace && \{0,1\} &\hspace{-24pt}\forall~ s \in \sS, e \in \EextM \label{IPM:VarFS} \tagMCF \\
\phantomsection		 &&  \priority &\miniSpace \in \miniSpace && [0,|\sS |-1] &\forall~ s \in S \label{IPM:VarP} \tagMCF 
\end{alignat}
\end{fleqn}
\end{algorithm}

\subsection{The MIP Model}

The formulation \ref{alg:MCF-MIP} uses one commodity for each Steiner site (see \ref{IPM:VarFS}) and a single commodity for the flow  originating at the terminals (see \ref{IPM:VarFT}).  Note that while $f^s$ defines a flow variable for each Steiner site $s \in \sS$ we use $f^T$ to denote a single commodity for all terminals. Furthermore note that flow variables $f^s$ corresponding to Steiner sites are binary whereas the aggregated flow variables $f^T$ from the terminals are defined to be integers.

We now briefly describe how a solution $(x,p,f^s,\f^T) \in \FeasibleMCF$ relates to a virtual arborescence $\solTree = (\solVtree, \solEtree, \hat{r}, \solfPath) \in \FeasibleACVSAP$. We naturally set $\solVtree \triangleq \{r\} \cup  \{s \in \sS| \solX_s \geq 1\} \cup \sT$ and $\hat{r} = r$. We continue by showing how $\solEtree$ and $\solfPath$ can be retrieved.

Constraints~\ref{IPM:FlowConservationT} and \ref{IPM:FlowConservationS}  specify flow preservation for the commodities such that terminal nodes emit one unit of flow in $f^T$ and activated Steiner nodes emit one unit of flow in $f^s$. Note that in Constraint~\ref{IPM:FlowConservationS} $\delta_{s,v}$ is a constant. As these constraints are specified for nodes $v \in \VG$, flows in $f^T$ and $f^s$ must terminate in $\sSink$ via edges in $\EextSmM$ or via $(r,\sSink)$.

If a Steiner node $s \in \sS$ is activated, $f^s$ defines a path $P^s$ from $s$ to $\sSink$. We therefore include $e=(s,P^s_{|P^s|-1})$ in $\solEtree$ and set $\solfPath(e)=\langle P^s_1, \dots, P^s_{|P^s|-1}\rangle$. As we use a single commodity for flow originating at the terminals, we have to first decompose $f^T$ into paths $\{P^t| t \in T\}$ such that $P^t$ originates at $t$ and terminates in $\sSink$. Due to the single destination, this can always be done using the standard $s-t$ flow decomposition~\cite{ahujaNetworkFlows}.

As the capacity constraints \ref{IPM:CapAll} are defined analogously to \ref{IP:CapSteiner}-\ref{IP:CapEdge}, we only need to establish the validity of connectivity condition \ref{def:VA2} to show that $\solTree \in \FeasibleACVSAP$ holds. As terminals and active Steiner nodes must be connected as discussed above, \ref{def:VA2} may only be violated by $\solTree$ if a cycle exists in $\solEtree$. To forbid such cycles, we adapt the well-known Miller-Tucker-Zemlim (MTZ) constraints~\cite{costa2009models} using continuous priority variables $\priority_s \in [0,|\sS|-1]$  in \ref{IPM:SetOrder}. The MTZ constraint \ref{IPM:SetOrder} enforces $f^s(\bar{s},\sSink) = 1 \Rightarrow p_s \geq p_{\bar{s}} + 1$, forbidding cyclic assignments containing only Steiner nodes. As terminals may not receive flow and the root may not send flow, this suffices to forbid cycles in $\solEtree$ overall and thus $\solTree \in \FeasibleACVSAP$ holds.

As formulations relying on MTZ constraints are comparatively weak~\cite{polzin2001comparison}, we introduce additional valid inequalities \ref{IPM:IncomingFlowActivatesFlow}, \ref{IPM:DisableSendingFlowToOneself} and \ref{IPM:SteinerPairsDoNotSendToEachOther} to strengthen the formulation. Constraint~\ref{IPM:DisableSendingFlowToOneself} disallows Steiner node $s \in \sS$ to absorb its own flow and \ref{IPM:SteinerPairsDoNotSendToEachOther} explicitly forbids cycles of length 2. Lastly, Constraint~\ref{IPM:IncomingFlowActivatesFlow} forces Steiner nodes receiving flow from another Steiner node to be activated. 

\section{Branch-and-Cut Solver}
\label{sec:implementation}

We have implemented VirtuCast based on \ref{alg:MIP} and \ref{alg:decompose}, which can be obtained from~\cite{rostSchmidWeb}. Our solver uses SCIP~\cite{achterberg2009scip} as underlying Branch-and-Cut framework with SoPlex~\cite{Wunderling1996} as LP solver. In Section~\ref{sec:separation} we shortly discuss our implementation of the separation procedures for Constraints \ref{IP:CutSteiner} and \ref{IP:CutTerminal}. Afterwards we present in Section~\ref{sec:heuristic} a primal heuristic to generate feasible solutions during the branch-and-bound search.

\subsection{Separation} 
\label{sec:separation}
Our solver generally follows the comprehensive work by Koch et al.~\cite{koch1998solving} and we assume the reader's familiarity with separation procedures (see e.g.~\cite{schrijver1998theory}). As the separation techniques used are well-known, we only sketch the most important features.

Instead of using a sophisticated maximal flow algorithm as~\cite{koch1998solving} proposes, we implemented the algorithm of Edmonds and Karp (see e.g.~\cite{ahujaNetworkFlows}). As choosing this simple algorithm only allows for constructing $s-t$ flows, we perform a single maximal flow computation for each $s\in \sS$ when separating connectivity inequalities \ref{IP:CutSteiner} and analogously perform $|T|$ many maximal flow computations when the valid inequalities of \ref{IP:CutTerminal} are to be separated. To improve performance for executing the maximal flow computations at each node, we use multithreading to speed up the computation.

Furthermore, we have implemented techniques that (empirically) \emph{improve}  the \emph{quality} of found violated inequalities for \ref{IP:CutSteiner} or \ref{IP:CutTerminal}. Following~\cite{koch1998solving} we implemented \emph{creep-flow} and \emph{nested-cuts}. We opted not to implement \emph{back cuts}, as in our formulation of \ref{IP:CutSteiner} violated node sets with respect to a given Steiner site $s\in \sS$ would probably be violated for other Steiner sites too. Adding back cuts for each of the violated node sets with respect to many $s \in \sS$ would in turn probably lead to many redundant constraints.

\subsection{Primal Heuristic} 
\label{sec:heuristic}
In Section~\ref{sec:results} we will show that the dual bound using the formulation \ref{alg:MIP} comes close to the optimal value within minutes of execution. Even though the branch-and-cut framework SCIP implements many primal heuristics~\cite{achterberg2009scip}, the heuristics of SCIP found effective for CVSAP are generally computationally expensive as they e.g. perform dive operations in the branch-and-bound tree (see Section~\ref{sec:results-main} for a discussion). Additionally, some of the heuristics implemented in SCIP, which are based on local search, already need a feasible solution as input. We therefore devloped Algorithm~\ref{alg:FlowDecoRound} for generating feasible solutions during the branch-and-bound process.

\subsubsection{Notational Remark.} In Algorithm~\ref{alg:FlowDecoRound}, we make use of the function \texttt{FlowDecomposition} with parameters $(G,f,v,t,D)$ to calculate a flow decomposition in graph $G$ from $t\in \VG$ to one of the nodes in $D\subseteq \VG$. The flow is given by $f:\EG \to \mathbb{R}_{\geq 0}$ and $v\in \mathbb{R}_{>0}$ specifies the amount of flow to decompose. The result of this function is a set of paths $\{(P_i,f_i)\} \in \mathcal{P}_G \times \mathbb{R}_{>0}$ with a value specifying the amount of flow carried by it, such that all paths start in $t$, terminate at one of the nodes of $D$, the sum of the carried flow amounts to $v$ and the sum of carried flow on each edge does not exceed the original flow $f$ on any edge.
We furthermore use the function \texttt{ShortestPath}$(G,c,t,D)$ to calculate the shortest paths in graph $G$ from $t\in \VG$ to one of the nodes in $D\subseteq \VG$ with respect to the partial cost function $c: \EG \to \mathbb{R}_{\geq 0}$. An edge for which no cost is specified, is assumed to be of zero cost.

\begin{algorithm}[hp]
\SetAlgoVlined

\newcommand{\FORALL}{\textbf{for all~}}
\newcommand{\FIND}{\textbf{find~}}
\newcommand{\IF}{\textbf{if~}}
\newcommand{\THEN}{\textbf{~then~}}
\newcommand{\STS}{\textbf{~s.t.~}}
\newcommand{\nL}{\nl}

\LinesNotNumbered
\DontPrintSemicolon
\SetAlgoRefName{FlowDecoRound}
\SetAlgoCaptionSeparator{}
\SetAlgorithmName{Algorithm}{}{{}}
\SetKwInOut{Input}{Input}\SetKwInOut{Output}{Output}
\SetKwFunction{FlowDecomposition}{FlowDecomposition}
\SetKwFunction{PruneSteinerNodes}{PruneSolution}
\SetKwFunction{ShortestPath}{ShortestPath}

\Input{Network $\G = (\VG,\EG,\cE,\uE)$, Request $R_\G = (r,\sS,\sT,\uR,\cS,\uS)$, \\~LP relaxation solution $(\solX, \solf)\in \FeasibleLP$ to~\ref{alg:MIP}}
\Output{Potentially a feasible Virtual Arborescence $\solTree$ for CVSAP}
\SetSideCommentRight
%\Function{Decompose}{Network $\G$, Request $R_G$, Solution $(\hat{x},\hat{f})\in \{0,1\}^S \times \pint^{\Eext}$}
%\Input{$(\hat{x},\hat{f})$ satisfies~\ref{IP:FlowConservation} -~\ref{IP:VarF}}
%\Output{$\solTree$ solves A-CVSAP}
    \BlankLine
    %\Begin{
\nL    \SET $\solS \triangleq \emptyset$ \AND $\solT \triangleq \emptyset$ \AND $U = \sT$\;
\nL	\SET $\solVtree \triangleq \{r\},~ \solEtree \triangleq \emptyset$ \AND $\solfPath: \solEtree \to \mathcal{P}_{\Gext}$ \;
\nL	\SET $u(e) \triangleq 	\begin{cases}	
								\uE(e) & \textnormal{, if } e \in \EG \\
								\uR(r) & \textnormal{, if } e = (r,\sSinkR) \\
								\uS(s) & \textnormal{, if } e = (s,\sSinkS) \in \EextSm \\
								1	   & \textnormal{, else}
							\end{cases}$ ~~ \FORALL $e \in \Eext $\;
\phantomsection \nL	\While{$ U \neq \emptyset$}{ \label{alg:flowdeco:beginPhase1}
\nL		\CHOOSE $t \in U$ uniformly at random \AND \SET $U \gets U - t$\;
\nL		\SET $\Gamma_t \triangleq \FlowDecomposition \left(\Gext, \solf, \solf(\sSource,t),t,\{\sSinkS, \sSinkR\} \right)$\;
\nL		\SET	 $\solf \gets \solf - \sum \limits_{(P,f) \in \Gamma_t, e \in P} f$\;
\nL		\SET $\Gamma_t \gets \Gamma_t \setminus \{ (P,f) \in \Gamma_t | \exists e \in P. u(e) = 0\}$\;
\nL		\SET $\Gamma_t \gets \Gamma_t \setminus \{ (P,f) \in \Gamma_t |  (\solVtree + t, \solEtree + (t,P_{|P|-1})) \textnormal{~is not acyclic~} \}$\;
\nL		\If{$\Gamma_t \neq \emptyset$}{
\nL			\CHOOSE $(P,f) \in \Gamma_t$ with probability $f/\left(\sum_{(P_j,f_j) \in \Gamma_t} f_j \right)$\;
\nL			\If{$P_{|P|-1} \notin \solVtree$}{
\nL				\SET $U \gets U + P_{|P|-1}$ \AND $\solVtree \gets \solVtree + P_{|P|-1}$\;
			}
\nL			\SET $\solVtree \gets \solVtree + t$ \AND $\solEtree \gets \solEtree + (t,P_{|P|-1})$ \AND $\solfPath(t,P_{|P|-1}) \triangleq P$\;
\nL			\SET $u(e) \gets u(e) - 1$ \FORALL $e \in P$\;
		}
	\phantomsection \label{alg:flowdeco:endPhase1}
	}
\nL	\SET $u(e) \gets 0$ \FORALL $e = (s,\sSinkS) \in \EextSm \textnormal{~\textbf{with}~} s \in \sS \wedge s \notin \solVtree $\;
\nL \SET	 $\unT \triangleq (\sT \setminus \solVtree) \cup (\{s \in \sS \cap \solVtree | \delta^+_{\solEtree}(s) = 0\})$\;
\nL	\phantomsection\For{$t \in \unT$}{  \label{alg:flowdeco:beginPhase2}
\nL		\CHOOSE $P \gets \ShortestPath \left( \Gext^{u},\cE,t, \{\sSinkS,\sSinkR\} \right)$\;
		\qquad \ST\ $(\solVtree + t, \solEtree + (t,P_{|P|-1}))$ is acyclic\;
\nL		\If{$P = \emptyset$}{
\nL			\KwRet{\textnormal{\texttt{null}}}	\;
		}
\nL		\SET $\solVtree \gets \solVtree + t$ \AND $\solEtree \gets \solEtree + (t,P_{|P|-1})$ \AND $\solfPath(t,P_{|P|-1}) \triangleq P$\;
\nL		\SET $u(e) \gets u(e) - 1$ \FORALL $e \in P$\;
\phantomsection 	 \label{alg:flowdeco:endPhase2}
	}
\nL	\For{$e \in \solEtree$}{
\nL		\SET $P \triangleq \solfPath(e)$\;
\nL		\SET $\solfPath(e) \gets \langle P_1, \dots, P_{|P|-1} \rangle $ \;
	}
\nL \SET $\solTree \triangleq \textnormal{Virtual Arborescence~} (\solVtree,\solEtree,r,\solfPath)$\;
\nL	\KwRet{$\PruneSteinerNodes{\ensuremath{\solTree}}$}\;
    %}
%\EndFunction
\caption{}
\label{alg:FlowDecoRound}
\end{algorithm}

\newpage 
\subsubsection{Synopsis of Primal Heuristic FlowDecoRound.}

Our primal heuristic \ref{alg:FlowDecoRound} uses the LP relaxation at the current node in the branch-and-bound tree as input. We denote by $\FeasibleLP$ the set of feasible solutions to \ref{alg:MIP} where the integrality restrictions on the flow (see \ref{IP:VarF}) and the decision variables for activating Steiner sites (see \ref{IP:VarX}) are relaxed to $\f_e \in \mathbb{R}_{\geq 0}$ and $\xS_s \in [0,1]$ for all $e \in \Eext$ and $s \in \sS$ respectively. Our heuristic works in the following three phases:
\begin{enumerate}
\item	In the first phase (see Lines~\ref{alg:flowdeco:beginPhase1} to \ref{alg:flowdeco:endPhase1}) for each terminal a flow decomposition is perfomed based on the flow values $\solf$ of the current LP solution $(\solX, \solf)\in \FeasibleLP$ such that the path may either terminate in $\sSinkS$ or $\sSinkR$. The flow decomposition returns a set of paths paired with an amount of flow carried by them. After discarding paths for which no capacity is left and paths that would lead to a cycle in the solution, one of the remaining paths is chosen uniformly according to the flow amount carried by it. If the path leads to an (inactive) Steiner site, then the aggregation node is opened and becomes itself a terminal to be connected during the first phase. If none of the paths returned by the flow decomposition is feasible, the terminal is not connected.
\item	In the second phase (see Lines~\ref{alg:flowdeco:beginPhase2} to \ref{alg:flowdeco:endPhase2}) the terminals (including Steiner nodes) that are still disconnected are connected using shortest paths under the restriction that these paths may not yield a cycle in the solution.
\item	If all terminals (including Steiner nodes) have been connected in the second phase, then a feasible solution has been constructed. Since in the first phase any Steiner node is activated if a path to it was selected, we try to reduce the cost of the solution by removing activated Steiner nodes from the solution in the third phase. This procedure is shown in Algorithm~\ref{alg:pruneSteinerNodes}. 
\end{enumerate}

\subsubsection{Synopsis of Local Search Algorithm PruneSteinerNodes.}
The activated Steiner nodes are put in the set $O$ (see Line~\ref{alg:prune:selectActiveSteinerNodes}). According to the ratio of cost for installing it divided by the number of nodes connected to it, the Steiner node maximizing this ratio is selected (see Line~\ref{alg:prune:chooseNode}). Together with all its incoming and outgoing edges, it is removed from the solution, yielding a (temporarily infeasible) solution $(\solVtree', \solEtree',r,\solfPath')$ and the remaining capacity $u': \EextM \to \mathbb{Z}$ (see Lines~\ref{alg:prune:beginConstructNew} to \ref{alg:prune:endConstructNew}). Thus, the objective value is decreased, giving an budget $b$ for reconnecting the disconnected nodes (see Line~\ref{alg:prune:budget}). Reconnecting the disconnected nodes is again done using shortest paths under the constraint that no cycles may be introduced to the solution (see Lines~\ref{alg:prune:beginReconnect} to \ref{alg:prune:endReconnect}). If a node cannot be connected or using the shortest path would exceed the budget (see Line~\ref{alg:prune:ifAbortCurrentNode}), the algorithm selects another activated Steiner node, if possible (see Line~\ref{alg:prune:doAbortCurrentNode}). If however all nodes could be reconnected and the budget was not exceeded, then a cheaper virtual arborescence  has been found and the process is restarted with all opened aggregation nodes (see Lines~\ref{alg:prune:acceptNewSolution} and ~\ref{alg:prune:resetActiveSteinerNodes}).

\subsubsection{General Remarks.}
The idea to use a flow decomposition to generate a set of possible paths and afterwards selecting one of the paths at random (according to the carried flow) was first proposed by Raghavan and Thomposon~\cite{Raghavan1985}. Our incentive to apply this scheme to CVSAP is twofold. Firstly, using this scheme only paths will be selected which have been (partially) accounted for in the objective of the LP relaxation. Secondly, it allows for an easy mechanism for deciding which Steiner nodes to activate. A Steiner node with much incoming flow is more likely to be opened than some Steiner node absorbing few flow. Furthermore we thereby circumvent the problem of deciding a priori which Steiner nodes should be activated or not. Section~\ref{sec:results-main} contains an evaluation of the performance of \ref{alg:FlowDecoRound}.

\begin{algorithm}[h!]
%\SetAlgoVlined

\newcommand{\FORALL}{\textbf{for all~}}
\newcommand{\FIND}{\textbf{find~}}
\newcommand{\IF}{\textbf{if~}}
\newcommand{\THEN}{\textbf{~then~}}
\newcommand{\STS}{\textbf{~s.t.~}}
\newcommand{\nL}{\nl}

\SetAlgoVlined
\LinesNotNumbered
\DontPrintSemicolon
\SetAlgoRefName{PruneSteinerNodes}
\SetAlgoCaptionSeparator{}
\SetAlgorithmName{Algorithm}{}{{}}
\SetKwInOut{Input}{Input}\SetKwInOut{Output}{Output}
\SetKwComment{doc}{//}{}

\Input{Network $\G = (\VG,\EG,\cE,\uE)$, Request $R_\G = (r,\sS,\sT,\uR,\cS,\uS)$, \\~Solution $\solTree \in \FeasibleACVSAP$ for A-CVSAP}
\Output{Feasible Virtual Arborescence $\solTree' \in \FeasibleACVSAP$ with $\cCVSAP(\solTree') \leq \cCVSAP(\solTree)$}
\SetSideCommentRight
%\Function{Decompose}{Network $\G$, Request $R_G$, Solution $(\hat{x},\hat{f})\in \{0,1\}^S \times \pint^{\Eext}$}
%\Input{$(\hat{x},\hat{f})$ satisfies~\ref{IP:FlowConservation} -~\ref{IP:VarF}}
%\Output{$\solTree$ solves A-CVSAP}
    \BlankLine
    %\Begin{
\nl \phantomsection	\SET $O \triangleq \sS \cap \solVtree$ \label{alg:prune:selectActiveSteinerNodes} \;
\nl	\While{$O \neq \emptyset $}{ 	\label{alg:prune:beginWhile}
\nl	\phantomsection	\CHOOSE $s \in O $ \textbf{maximizing~} $\cS(s) / |\delta^{-}_{\solEtree}(s)|$ \label{alg:prune:chooseNode} \;
\nl		\SET $O \gets O - s$ \label{alg:prune:remove-s}\;
\nl		\SET $U \triangleq \{ t | (t,s) \in \delta^{-}_{\solEtree}(s) \}$ \;
\nl		\SET $R \triangleq \delta^{+}_{\solEtree}(s) \cup \delta^{-}_{\solEtree}(s)$ \;
\nl		\SET $\mathcal{P}_R \triangleq \{\solfPath(e) | e \in R \}$ \;
\nl	\phantomsection	\SET $b \triangleq \cS(s) + \sum_{P \in \mathcal{P}_R} \cE(P)$ \label{alg:prune:budget}\;
\nl \phantomsection		\SET $\solVtree' \triangleq \solVtree \setminus (U \cup \{s\})$ \label{alg:prune:beginConstructNew} \;
\nl		\SET $\solEtree' \triangleq \solEtree \setminus (\delta^{-}_{\solEtree}(s) \cup \delta^{+}_{\solEtree}(s))$ \AND $\solfPath': \solEtree' \to \mathcal{P}_G$ \;
		\qquad \ST 	$\solfPath'(e) = \solfPath(e)$ \FORALL $e \in \solEtree'$ \;
\nl		\SET $u'(e) \triangleq 	\begin{cases}	
									\uE(e) - |\solfPath(\solEtree')[e]| & \textnormal{, if } e \in \EG \\
									\uR(r) - |\delta^{-}_{\solEtree'}(r)| & \textnormal{, if } e = (r,\sSinkR) \\
									\uS(s) - |\delta^{-}_{\solEtree'}(s')| & \textnormal{, if } e = (s',\sSinkS) \in \EextSm \\
									1	   & \textnormal{, else}
								\end{cases}$ ~~ \FORALL $e \in \Eext $\;
\nl	\phantomsection	\SET $u'(s,\sSinkS) \gets 0$ \label{alg:prune:endConstructNew}\;
\nl	\phantomsection	\For{$e=(t,s) \in \delta^{-}_{\solEtree}(s)$}{ \label{alg:prune:beginReconnect}
\nl			\CHOOSE $P \triangleq \ShortestPath \left( \Gext^{u'},\cE, t,\{\sSinkS,\sSinkR\} \right)$\;
			\qquad \ST\ $(\solVtree' + t, \solEtree' + (t,P_{|P|-1}))$ is acyclic\;
\nl	\phantomsection		\If{$P = \emptyset \vee b - \cE(P) \leq 0$ \label{alg:prune:ifAbortCurrentNode}}{
\nl	\phantomsection			\textbf{goto~}\ref{alg:prune:beginWhile} \label{alg:prune:doAbortCurrentNode}\;
			}
\nl			\SET $b \gets b - \cE(P)$\;
\nl			\SET $\solVtree' \gets \solVtree' + t$\;
\nl			\SET $\solEtree \gets \solEtree + (t,P_{|P|-1})$ \AND $\solfPath(t,P_{|P|-1}) \triangleq \langle P_1, \dots, P_{|P|-1} \rangle $\;
\nl	\phantomsection		\SET $u(e) \gets u(e) - 1$ \FORALL $e \in P$\;				
		
\label{alg:prune:endReconnect}		}
\nl	\phantomsection	\SET $\solTree \gets $ Virtual Arborescence $(\solVtree', \solEtree', r, \solfPath')$ \label{alg:prune:acceptNewSolution}\;
\nl \phantomsection			\SET $O \gets S \cap \solVtree$ \label{alg:prune:resetActiveSteinerNodes}\;
	}
\nl	\KwRet{$\solTree$}
	
     %\KwRet{$\solTree$}
    %}
%\EndFunction
\caption{}
\label{alg:pruneSteinerNodes}
\end{algorithm}

\section{Computational Evaluation}
\label{sec:results}

We investigate the applicability of VirtuCast with an empirical computational study on different problem instances using the solver that was  introduced in Section~\ref{sec:implementation}. We consider two different classes of instances, one being based on grid graphs and the other one being based on ISP topologies, see Section~\ref{sec:results-problem-classes}. In Section~\ref{sec:resuls-parameter-validation} we validate the choice of including the directed cut constraints \ref{IP:CutTerminal} as well as separation related optimizations in our implementation. While the main results considering our implementation are presented in Section~\ref{sec:results-main}, we present a computational comparison with the multi-commodity flow formulation \ref{alg:MCF-MIP} in Section~\ref{sec:results-comparison-mcf}. We conclude our computational study in Section~\ref{sec:results-analysis} with analyzing the runtime allocated by the different components of our solver to devise possible optimizations. Note that all problem instances used for our evaluation are available together with our solver from~\cite{rostSchmidWeb}.

\subsubsection{Technical Notes.}
All our experiments were conducted on machines equipped with an 8-core Intel Xeon L5420 processor running at 2.5 Ghz and 16 GB RAM. As we use $25$ instances for each problem class, we mainly use box plots to present our results. Note that all boxplots presented in this section use the standard $1.5*\mathit{IQR}$ whiskers of \texttt{R} (and not the 95th and 5th percentiles).

\subsection{Problem Classes} 
\label{sec:results-problem-classes}
We use two classes of problems for our experimal evaluation. One class is based on $n \times n$ grid graphs while the other is based on router-level  Internet topologies~\cite{igen2009}. Based on the inherent symmetry and supported by our computational results, problem instances based on grid graphs present hard instances that already for $n=20$ cannot be solved to optimality within reasonable time. On the other hand, we use Internet topologies to show the applicability to solve realistically sized instances close to optimality. 

\subsubsection{Grid Graphs.}
All $n \times n$ grid instances were generated according to the following parameters. Edge capacities are set to $3$ while the capacity of the root and Steiner sites is set to $5$. We use unit costs on edges and a cost of $20$ for activating Steiner nodes. The locations of terminals, Steiner sites and the root are chosen in a uniformly distributed fashion, such that $|\sS| \propto 20\% |\VG|$ and $|\sT| \propto 25 \% |\VG|$ holds. We consider problems based on $n \times n$ grids for $n=12,16,20$ and generated $25$ instances for each of these sizes. Table~\ref{tab:grid-sizes} summarizes the resulting number of graph sizes and the number of generated Steiner sites and terminals.

\begin{table}[b!]
\begin{center}
\begin{tabular}{|c|c|c|c|c|}
\hline		 	$n$   &	  $|V|$ 	 &	 $|E|$ 	 & 	$|\sS|$ & $|\sT|$ 	\\
\hline 			12	  &		144	 & 	528		&	29		& 36 	\\
\hline			16	  & 	 	256	 & 	960		& 	51		&	64	\\
\hline			~~~~20~~~~	  & 		~~~~400~~~~	 & 	~~~~1520~~~~		&	~~~~80~~~~		& 	~~~~100~~~~ \\
\hline
\end{tabular}
\caption{Size of graph, number of Steiner sites $|\sS|$ and terminals $|\sT|$ for grid instances.}
\label{tab:grid-sizes}
\end{center}
\vspace{-8pt}
\end{table}

\subsubsection{Internet Topologies.}

We have used the tool IGen~\cite{igen2009} to generate two Internet alike topologies, one having $1600$ and the other having $3200$ nodes. Nodes are distributed uniformly on a world map. Topologies are created by clustering nodes in a $20$ (horizontal) by $6$ (vertical) grid. Each of these clusters can be understood as a single autonomous system (AS). Within each AS a certain number of nodes are selected to become Points-of-Presences (PoP). 

\begin{figure}[t]
\centering
\includegraphics[width=0.9\textwidth]{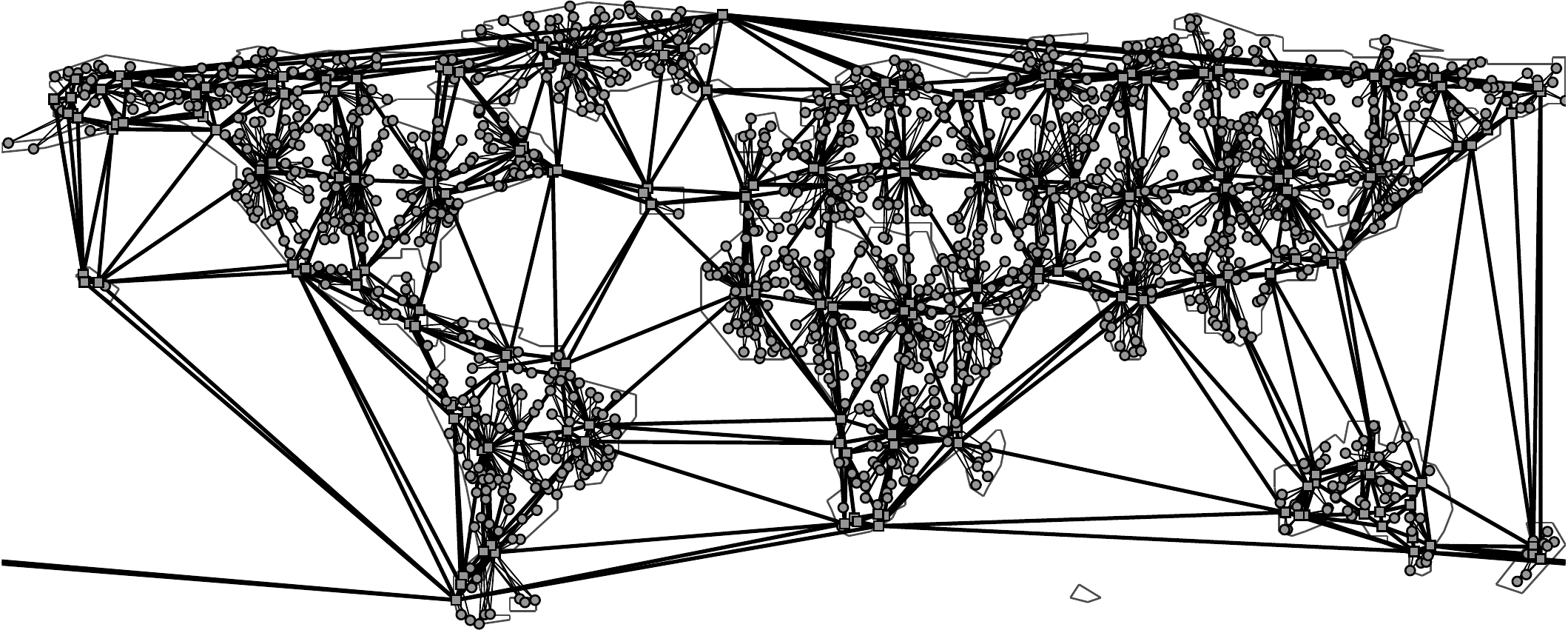}
\caption{Topology generated by IGen with $1600$ nodes. Nodes depicted as squares are PoPs while node depicted as disks are internal nodes.}
\label{fig:igen-topology-1600}
\end{figure}

While nodes within an AS are solely connected to a certain number of internal PoPs\footnote{using the sprint heuristic, see~\cite{igen2009} \label{igen:foot:sprint}}, PoPs are interconnected to provide global connectivity\footnote{using a delaunay triangulation, see~\cite{igen2009} \label{igen:foot:delaunay}}. In Figure \ref{fig:igen-topology-1600} the generated instance with $1600$ nodes is depicted. Table~\ref{tab:igen-param} gives an overview over the characteristics for the two topologies. Note that for IGen.3200 the number of PoPs per cluster as well as the number of links between internal nodes and PoPs has been increased.
For each of the topologies we have again generated 25 different instances according to the following parameters.  Steiner sites are only located at PoPs and Terminals may not be PoPs. The cost of using edges is given by the euclidean distance. Inter-PoP links have a capacity of $10$ while intra-AS links have a capacity of $2$. Activation costs for Steiner sites are chosen according to  $\mu(\cE)\cdot\mathcal{U}(25,75)$, where $\mu(\cE)$ denotes the average edge length. Steiner sites as well as the root have a capacity of $5$.

\vspace{-10pt}
 
\begin{table}[h!]
\begin{center}
\begin{tabular}{|c|c|c|c|c|c|c|c|}
\hline	Name 					& 	$|V|$		& 	$|E|$ 			& 	~~$|P|$~~ 	& 	~~$|I\to P|$~~ 	& ~~$|P \to P|$~~ 	& 	$|\sS|$ 		& $|\sT|$ \\
\hline 	~~\textbf{IGen.1600}~~ 	& 	~~$1600$~~ 	&	~~~~$6816$~~~~	&	~~$3$~~ 		& 	~~$2$~~			& 	~~$2$~~			& 	~~$200$~~	& ~~$300$~~	\\
\hline	~~\textbf{IGen.3200}~~	& 	$3200$		&	$19410$			& 	$4$			& 	$3$				&	$2$				&	$400$		& $600$	\\
\hline
\end{tabular}
\end{center}
\caption{Parameters for instances IGen.1600 and IGen.3200, with $|V|$ many nodes and $|E|$ edges. $|P|$ denotes the number of PoPs per cluster. $|I\to P|$ denotes the number of adjacent PoP nodes for internal nodes \textsuperscript{\ref{igen:foot:sprint}}. $|P \to P|$ denotes the number of peers\textsuperscript{\ref{igen:foot:delaunay}}.}
\label{tab:igen-param}
\end{table}

\subsection{Validation of Implementation Choices}
\label{sec:resuls-parameter-validation}
As shown in Section~\ref{sec:results-main} instances based on grid graphs are even for smaller instances hard to solve. As we believe these instances to be hard for CVSAP due to their inherent symmetry (see~\cite{rosseti2001new} for a discussion of hypercube graphs as hard instances for STP), we validate the choice of implementation parameters on this problem class. For this purpose we have chosen instances with $n=12$ as they are still throughout solvable to optimality within reasonable time.

Following the generation parameters as described in Section~\ref{sec:results-problem-classes}, our $12 \times 12$ grid instances contain $29$ possible Steiner sites and $36$ terminals that need to be connected. We consider the following four parameter settings to evaluate whether optimizations for the separation procedure \textbf{(S)} and the separation of the directed Steiner cuts \ref{IP:CutTerminal} for terminals \textbf{(T)} improve performance. 
\begin{center}
\begin{tabularx}{\textwidth}{cX}
\textbf{(T\&S)} 	& 	Constraints \ref{IP:CutTerminal} are separated and nested-cuts and creep-flow are used. \\
\textbf{(T)}	 	& 	Constraints \ref{IP:CutTerminal} are separated but neither nested-cuts nor creep-flow are used. \\
\textbf{(S)} 		& 	Constraints \ref{IP:CutTerminal} are not separated but nested-cuts and creep-flow are used. \\
\textbf{(-)} 		& 	Neither constraints \ref{IP:CutTerminal} are separated, nor are nested-cuts or creep-flow used.\\
\end{tabularx}
\end{center}

\begin{figure}[p]
\begin{subfigure}[t]{0.45\textwidth}
\includegraphics[width=1\textwidth]{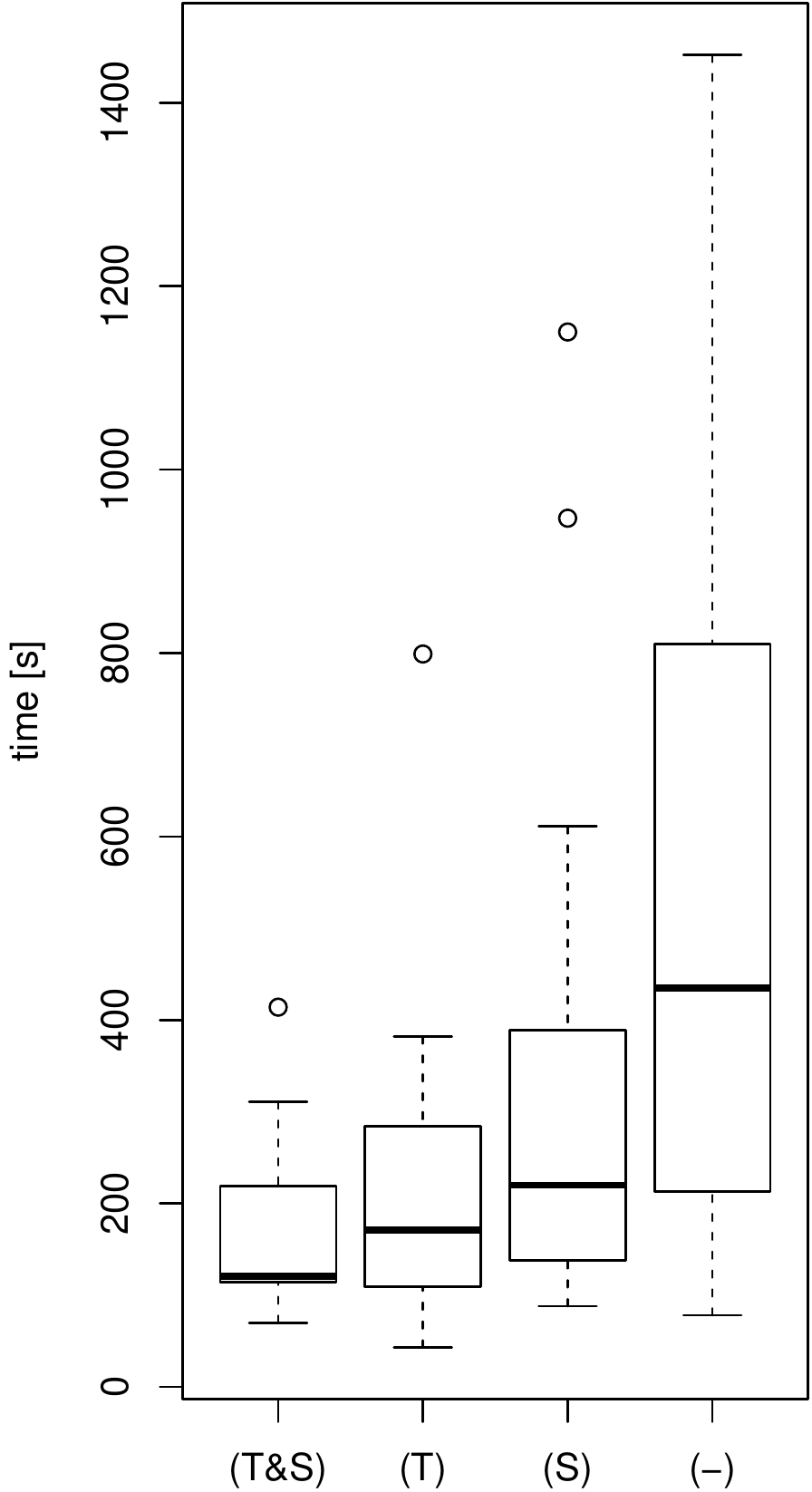}
\label{figure:grid-12x12-boxplots-runtime-comparison}
\end{subfigure}
\qquad
\begin{subfigure}[t]{0.45\textwidth}
\includegraphics[width=1\textwidth]{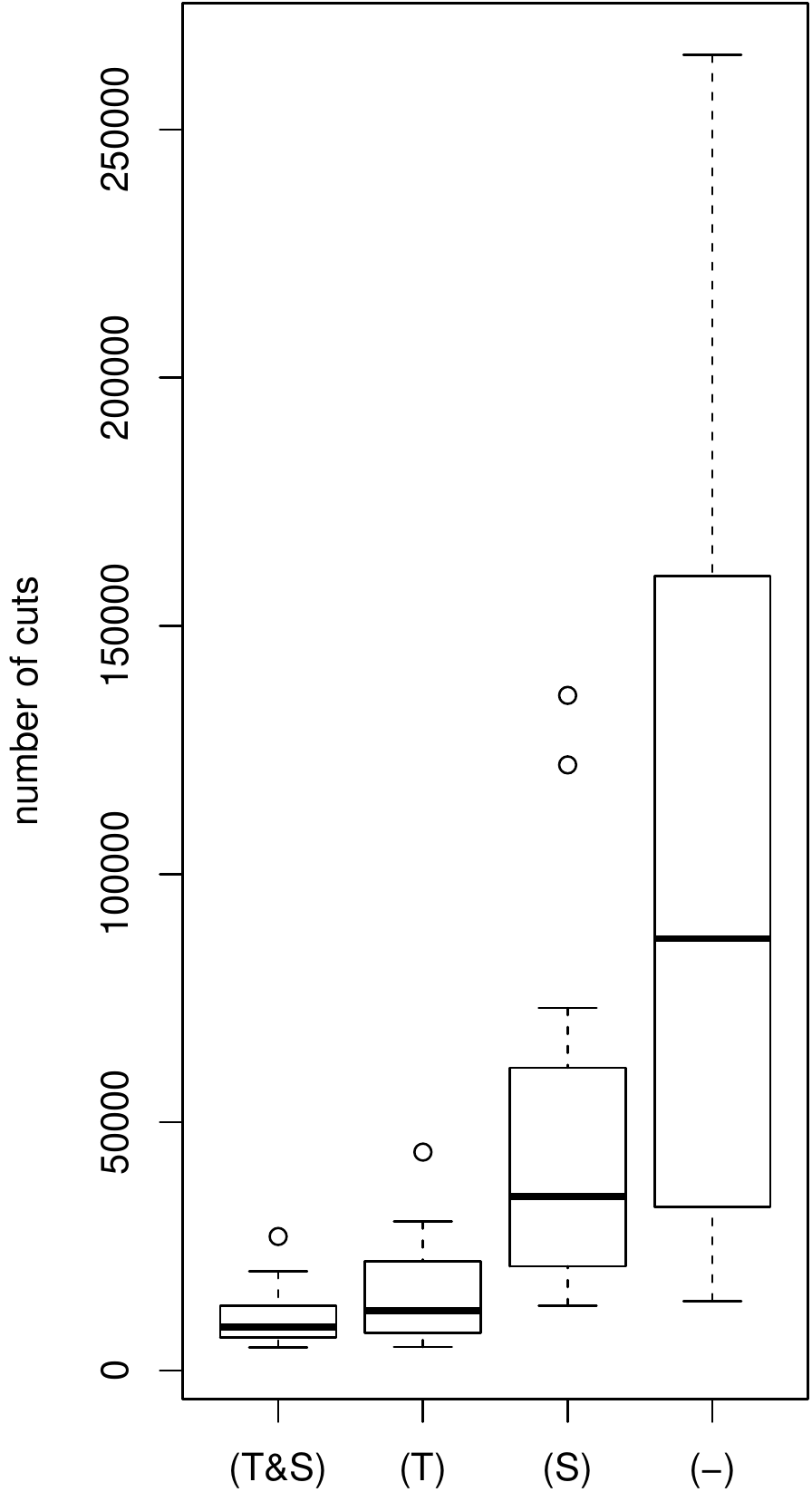}
\label{figure:grid-12x12-boxplots-cut-comparison}
\end{subfigure}
\caption{Runtimes in seconds and number of generated cuts for the $25$ runs on $12\times 12$ grids using different parameters.}
\label{fig:12-boxplot-runtimes-cuts}
\end{figure}
Figure~\ref{fig:12-boxplot-runtimes-cuts} plots the total runtime as well as the number of cuts introduced, for solving each of the $25$ instances to optimality.  Note that while \textbf{(T\&S)} provides the best runtime performance, the runtime distribution of \textbf{(S)} comes close to that of \textbf{(T)}. However, the number of cuts generated by \textbf{(S)} is substantially higher. This indicates, that even though the separation of \ref{IP:CutTerminal} comes at computational expenses, the found cuts of \textbf{(T)} are of high quality and drastically reduce the number of overall neded cuts. Clearly, \textbf{(-)} provides the worst performance.

\begin{figure}[p]
\includegraphics[width=1\textwidth]{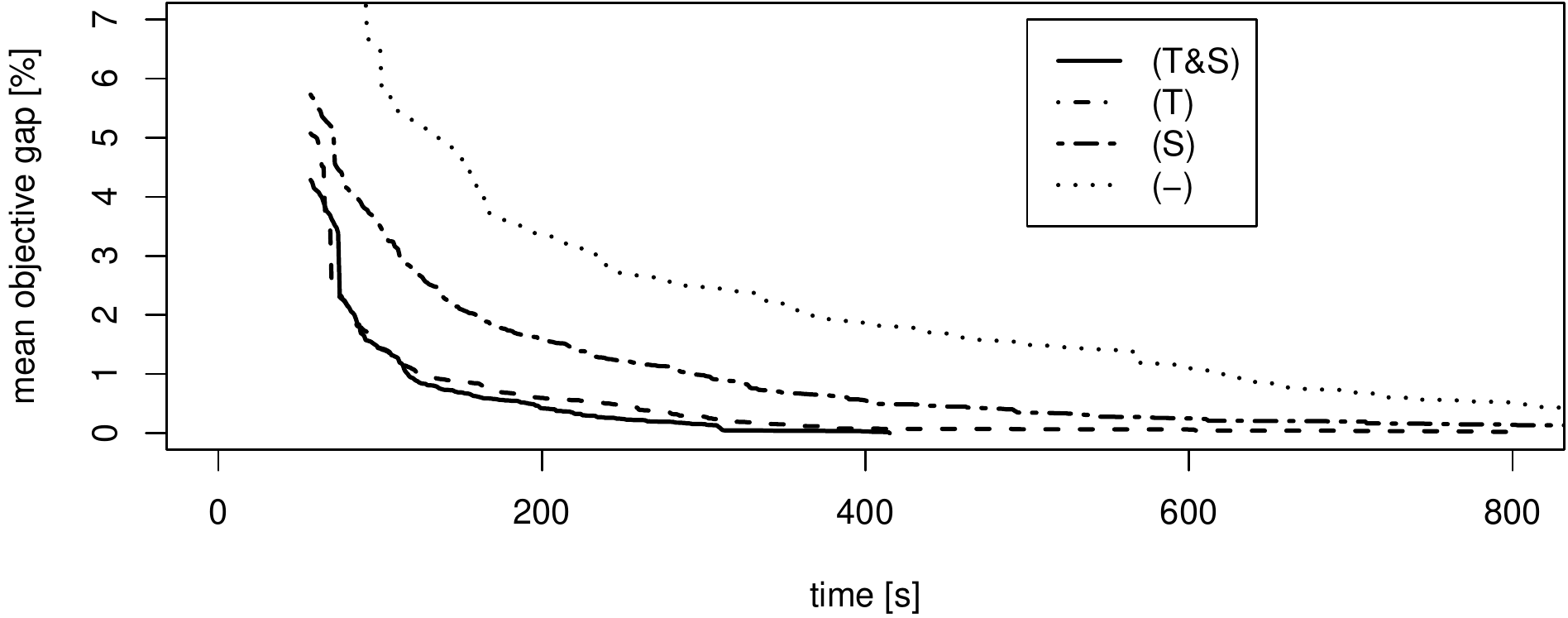}
\caption{The mean of the objective gap for the $25$ instances using the parameter settings \textbf{(T\&S)}, \textbf{(T)}, \textbf{(S)} and \textbf{(-)}.}
\label{fig:12-lines-gap-comparison}
\end{figure}

However, as one might favor an approach that generates \emph{provably} good solutions of high quality rather quickly but overall takes longer to solve a problem to optimality, we investigate the objective gap over time.  The objective gap is formally defined as $|P-D|/|D|$ where $P$ denotes the value of the best (primal) solution found and $D$ denotes the value of the (dual) lower bound. In Figure~\ref{fig:12-lines-gap-comparison} the mean objective gap over time for the $25$ instances is shown and Figure~\ref{fig:12-boxplots-gap-over-time-all} provides detailed boxplots for each of the parameters. Again, \textbf{(T\&S)} provides the best performance, while \textbf{(T)} comes very close to it. Importantly, the mean objective gap for \textbf{(S)} is noticeably higher than for \textbf{(T)}, even though the total runtimes are comparable.\\

\begin{figure}[p]
\begin{subfigure}[t]{\textwidth}
\includegraphics[width=1\textwidth]{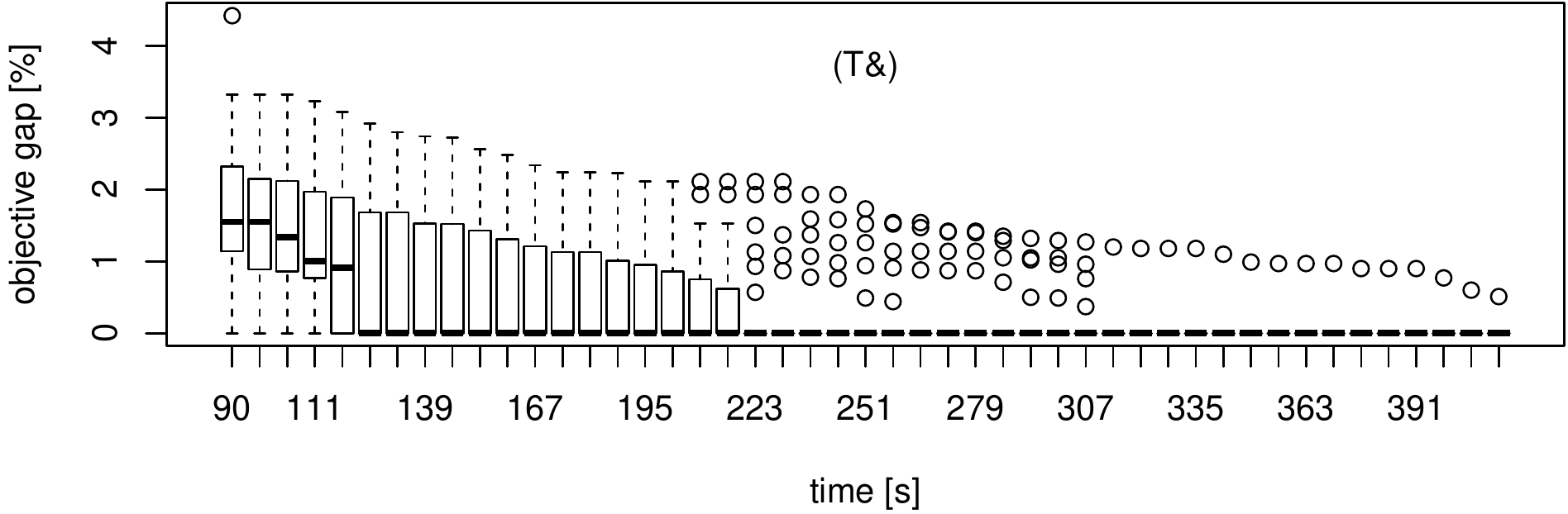}
\label{fig:12-boxplots-gap-over-time-ful-opt}
\end{subfigure}
\begin{subfigure}[t]{\textwidth}
\includegraphics[width=1\textwidth]{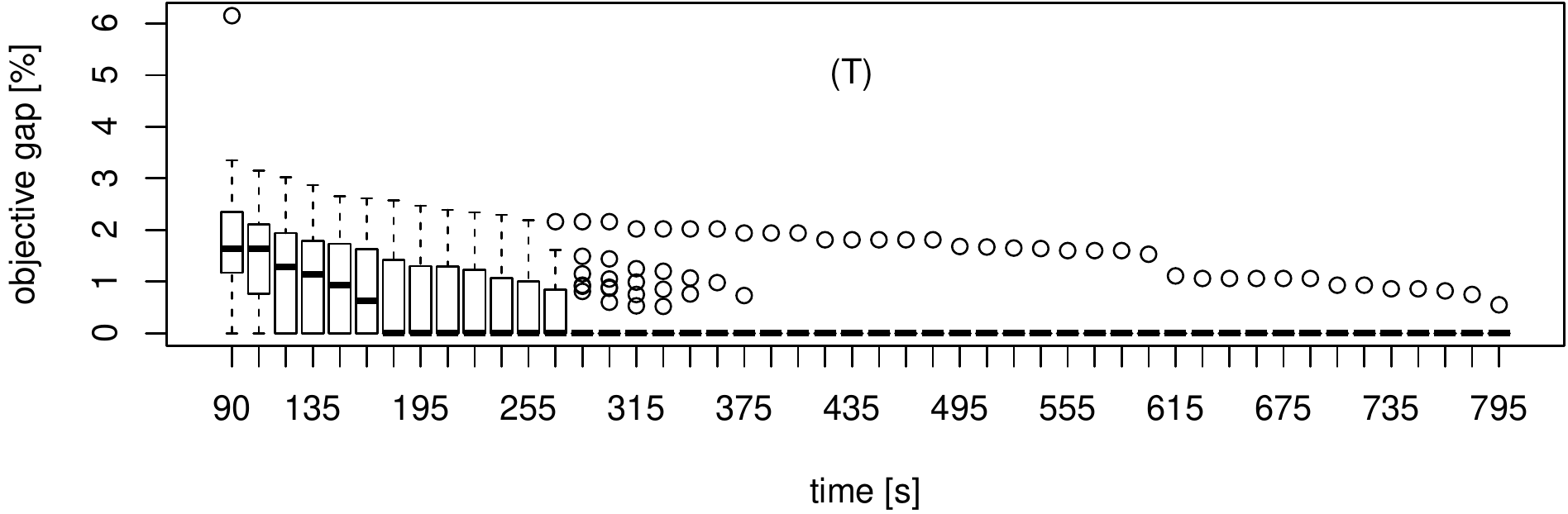}
\label{fig:12-boxplots-gap-over-time-wo-sep-opt}
\end{subfigure}
\begin{subfigure}[t]{\textwidth}
\includegraphics[width=1\textwidth]{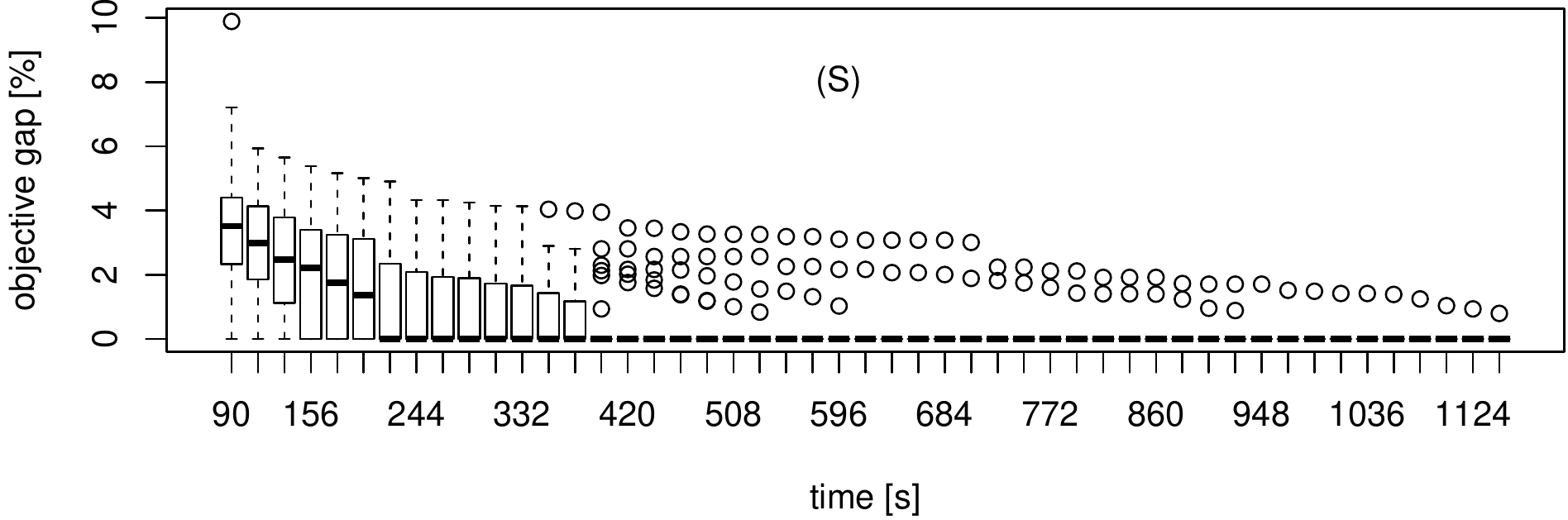}
\label{fig:12-boxplots-gap-over-time-wo-term-opt}
\end{subfigure}
\begin{subfigure}[t]{\textwidth}
\includegraphics[width=1\textwidth]{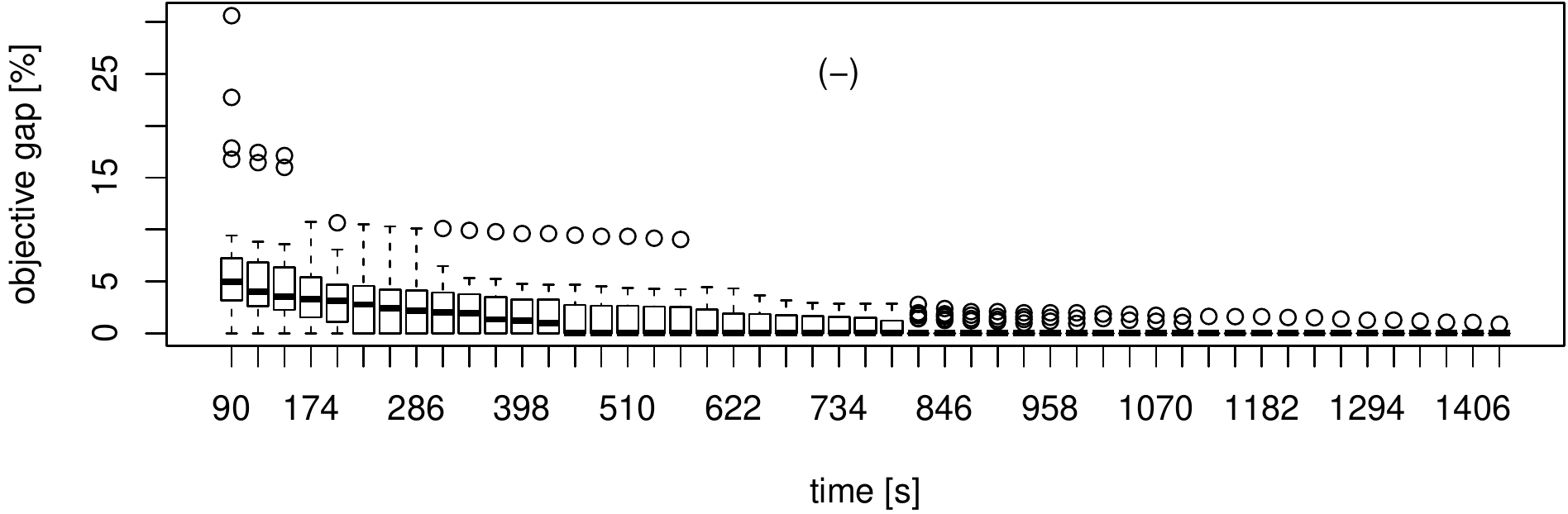}
\end{subfigure}
\caption{The objective gap over time for the same $25$ instances using $12 \times 12$ grids under the four different parameter settings \textbf{(T\&S)}, \textbf{(T)}, \textbf{(S)} and \textbf{(-)} measured at regular time intervals. Note the different time scales.}
\label{fig:12-boxplots-gap-over-time-all}
\end{figure}

We conclude by observing that CVSAP can be solved within minutes to optimality on $12 \times 12$ grids using parameter setting \textbf{(T\&S)} and that the separation of the valid inequalities \ref{IP:CutTerminal} dramatically improves performance. Furthermore, the usage of creep-flow and nested-cuts does improve performance slightly even though incurring further computational costs during the separation procedure. However, as the number of overall generated cuts is reduced, we have chosen to conduct all further experiments using the setting \textbf{(T\&S)}.

 \newpage
\subsection{Main Computational Results}
\label{sec:results-main} 

As shown in Section~\ref{sec:resuls-parameter-validation} instances based on grid graphs for $n=12$ can be solved to optimality within minutes. We will now consider instances based on the larger $n \times n$ grids with $n=16,20$ and as well as instances based on Internet Topologies IGen.1600 and IGen.3200. For each of these problem classes we present results for $25$ independently generated instances (see Section~\ref{sec:results-problem-classes} for parameters). As many of the instances cannot be solved to optimality anymore, we terminate experiments after $2$ hours.

\begin{figure}[p]
\begin{subfigure}[t]{\textwidth}
\includegraphics[width=1\textwidth]{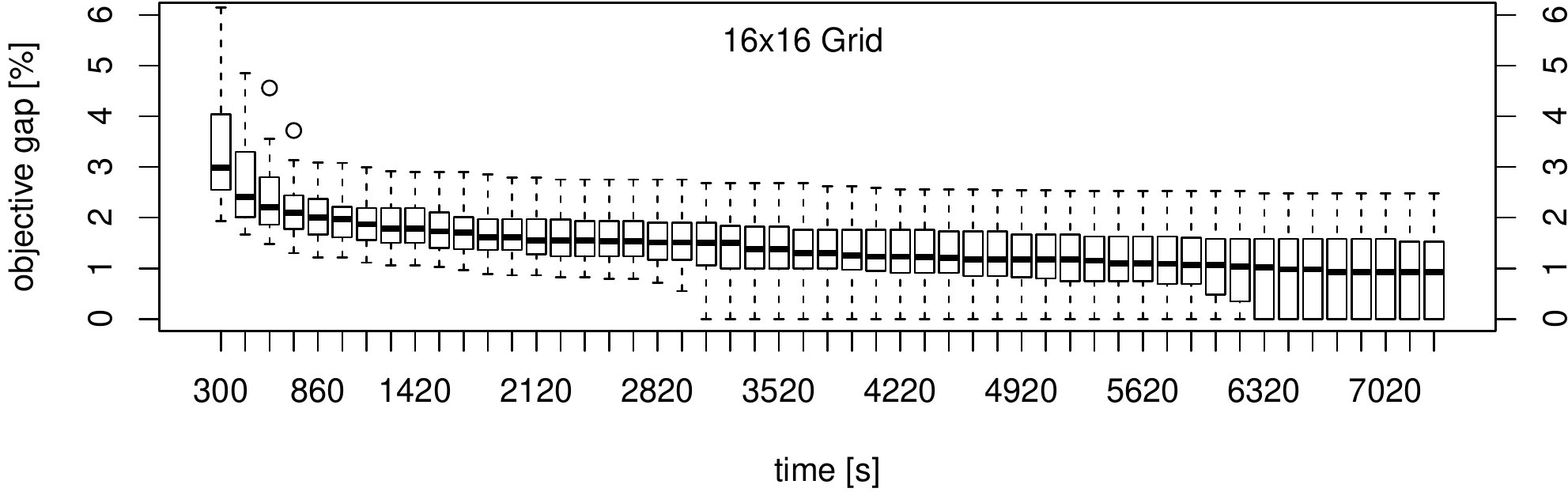}
\end{subfigure}

\bigskip

\begin{subfigure}[t]{\textwidth}
\includegraphics[width=1\textwidth]{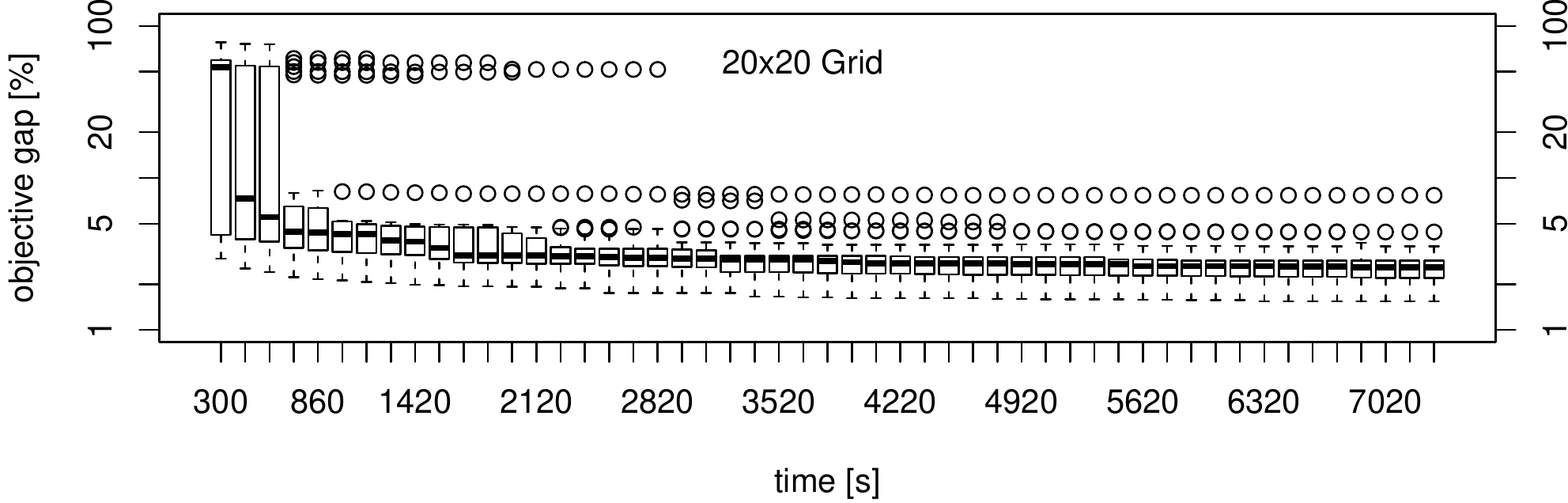}
\end{subfigure}

\bigskip
\begin{subfigure}[t]{\textwidth}
\includegraphics[width=1\textwidth]{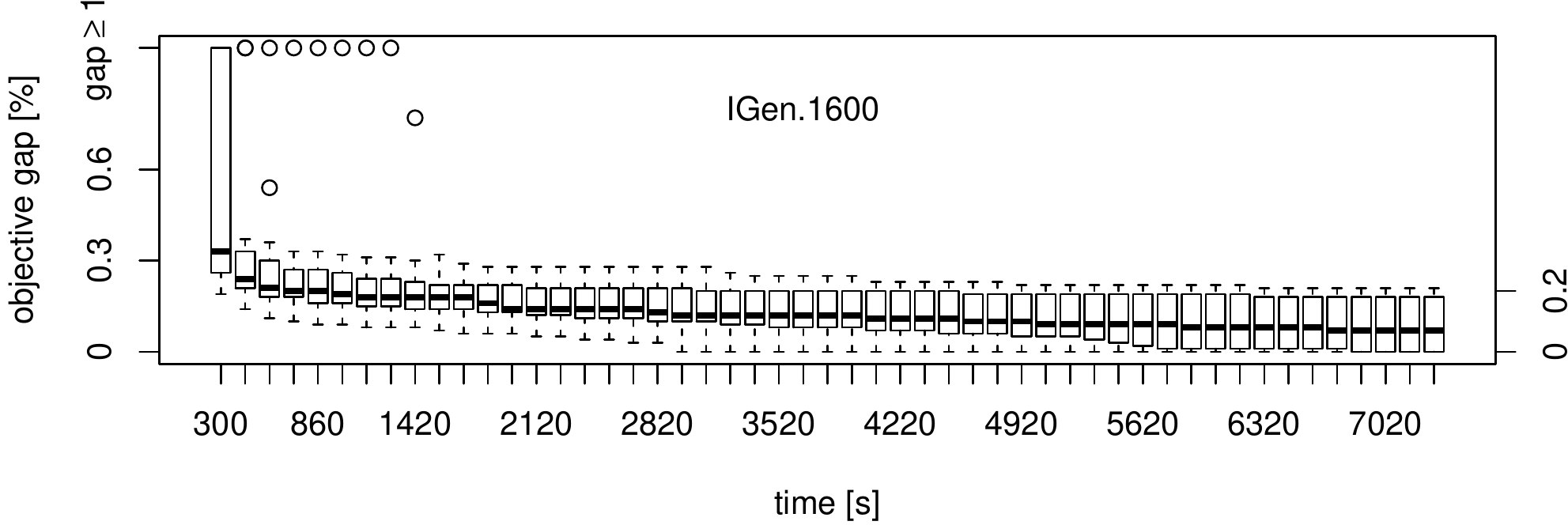}
\end{subfigure}

\bigskip
\begin{subfigure}[t]{\textwidth}
\includegraphics[width=1\textwidth]{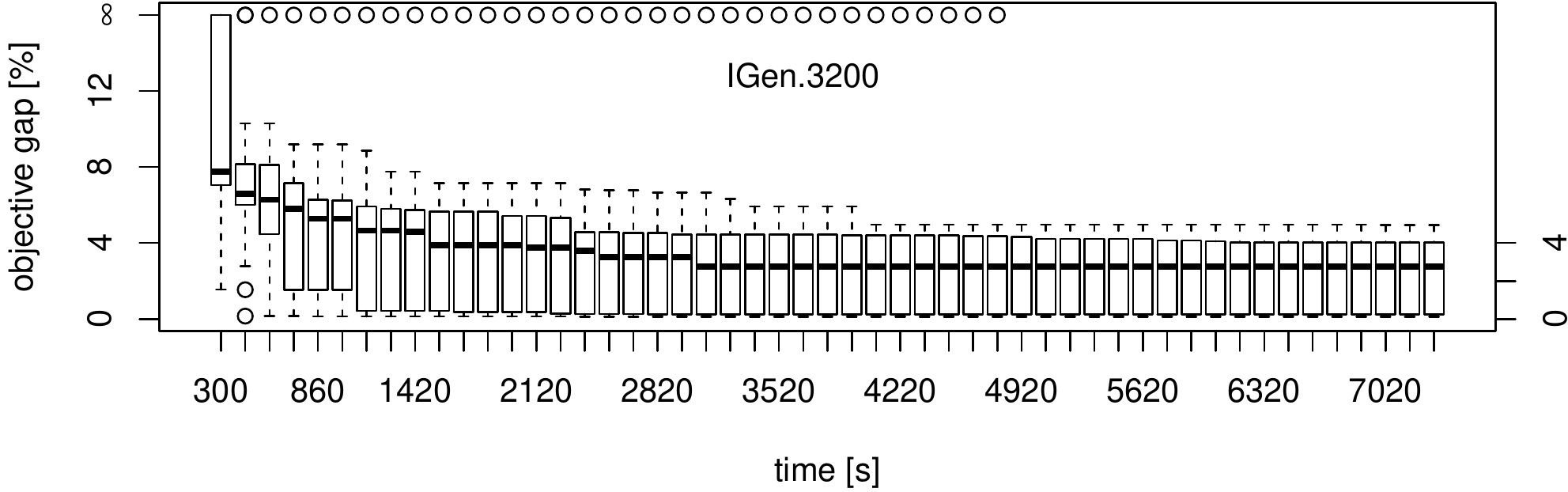}
\end{subfigure}

\caption{The objective gap of $25$ instances for $n\times n$ grids with $n=16,20$ as well as for IGen.1600 an IGen.3200 measured at regular time intervals. Note the logarithmic y-axis for $n=20$. For IGen.1600, objective gaps above $1\%$ are subsumed. For IGen.3200 an objective gap of $\infty$ expresses, that no primal solution has been found.}
\label{fig:main-results-gap}
\end{figure}

\begin{figure}[p]
\begin{subfigure}[t]{\textwidth}
\includegraphics[width=1\textwidth]{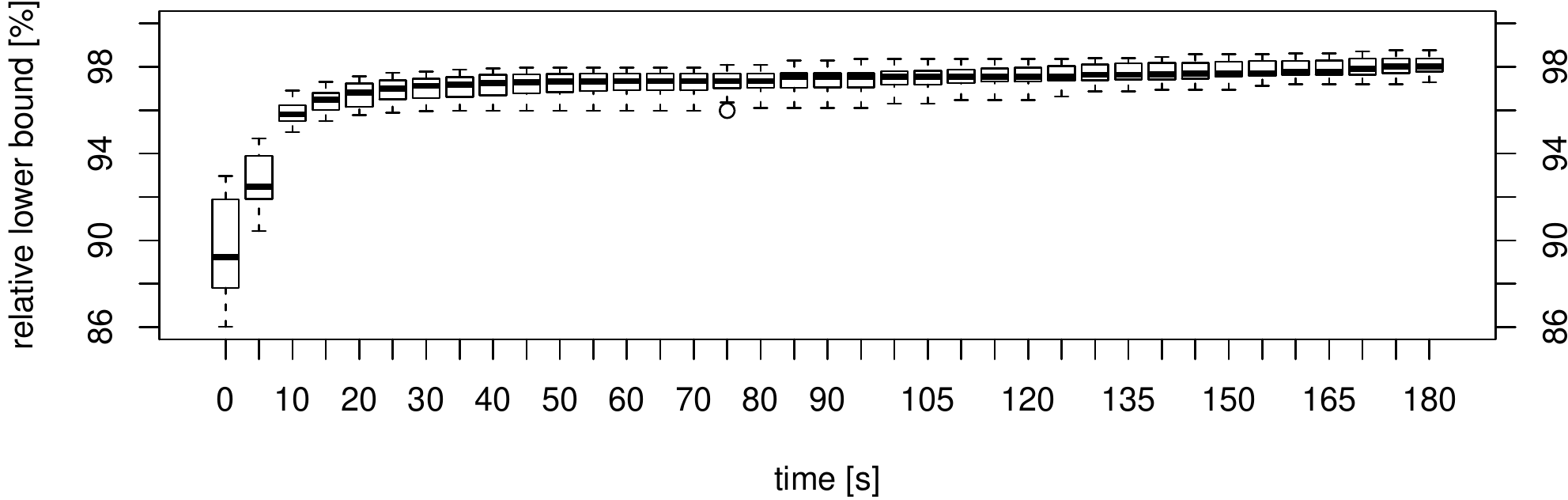}
\end{subfigure}

\bigskip

\begin{subfigure}[t]{\textwidth}
\includegraphics[width=1\textwidth]{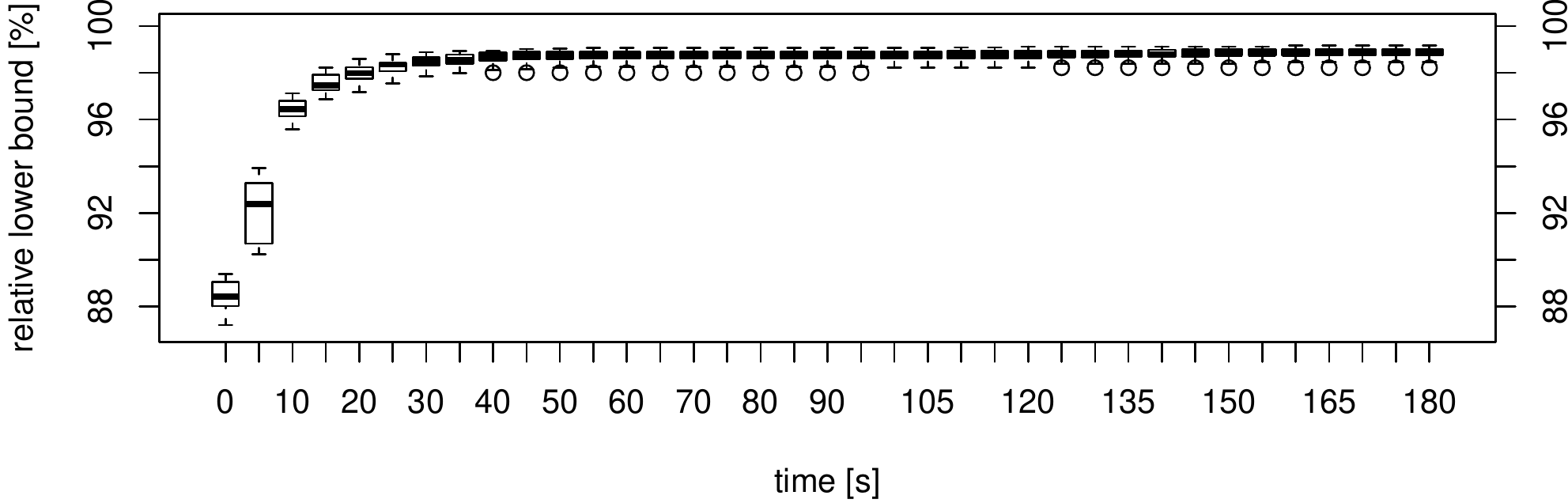}
\end{subfigure}

\bigskip

\begin{subfigure}[t]{\textwidth}
\includegraphics[width=1\textwidth]{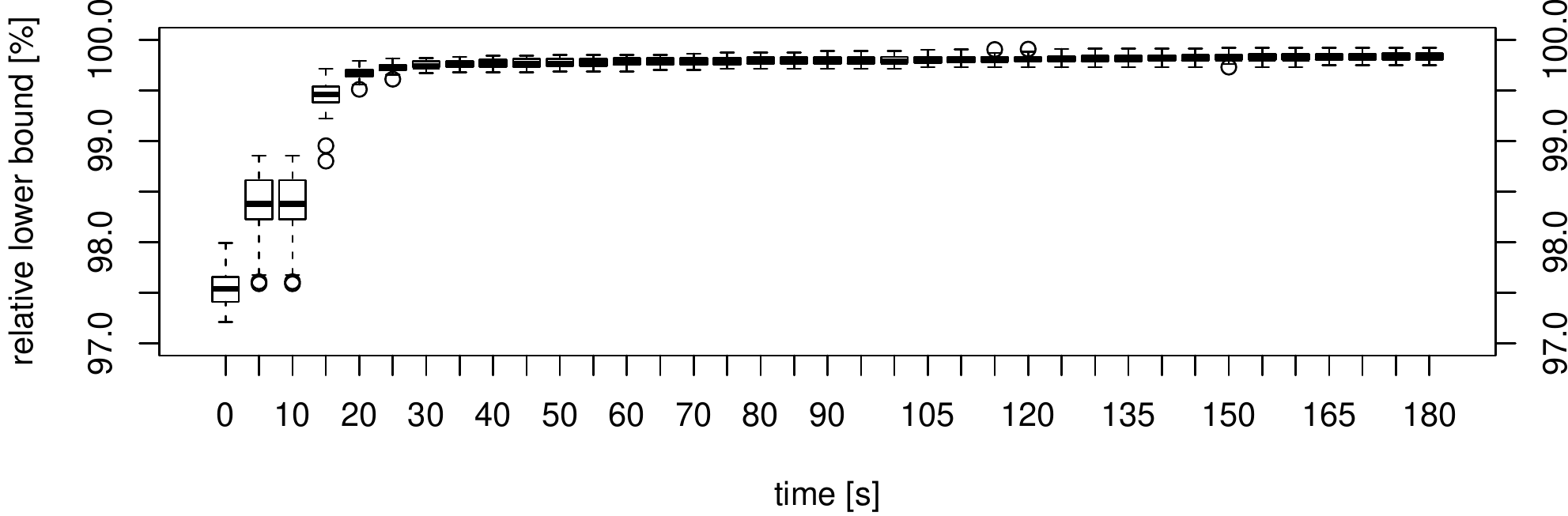}
\end{subfigure}

\bigskip

\begin{subfigure}[t]{\textwidth}
\includegraphics[width=1\textwidth]{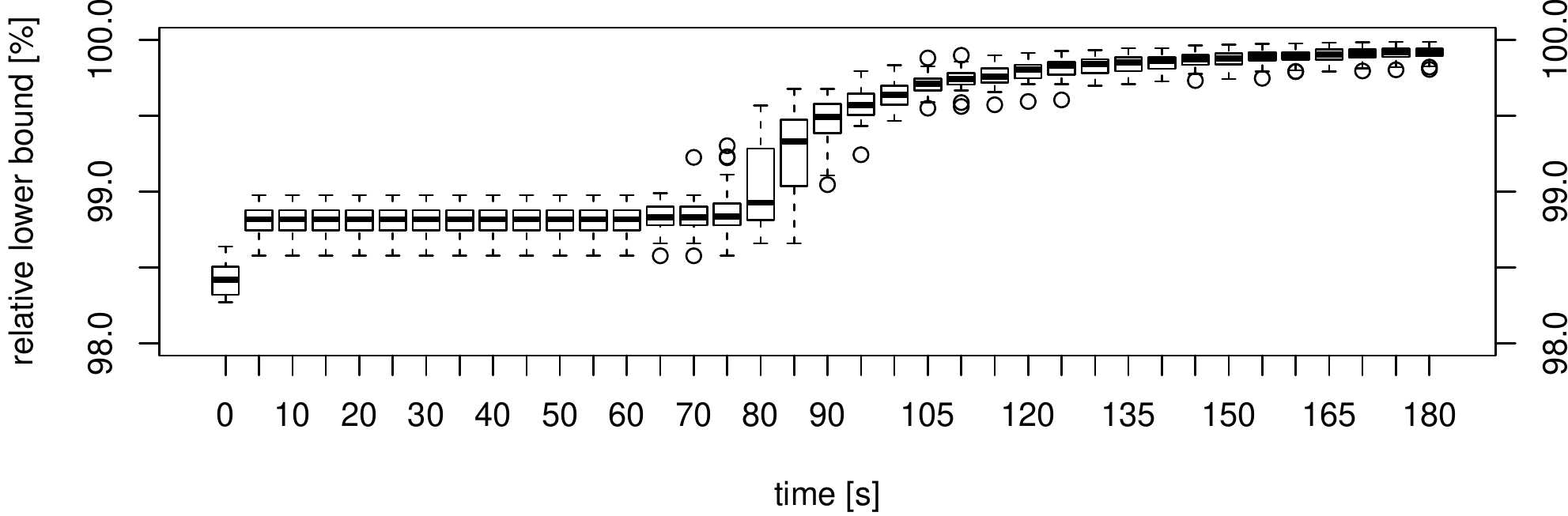}
\end{subfigure}

\caption{The relative lower bound, i.e. the lower bound compared to the best lower bound found after 2 hours, of the $n \times n$ grid graphs for $n=16,20$ as well as for IGen.1600 and IGen.3200 measured at regular time intervals. Please note the log scale for $n=20$.}
\label{fig:main-results-dual}
\end{figure}
\begin{figure}[p]
\centering
\begin{subfigure}[t]{\textwidth}
\includegraphics[width=1\textwidth]{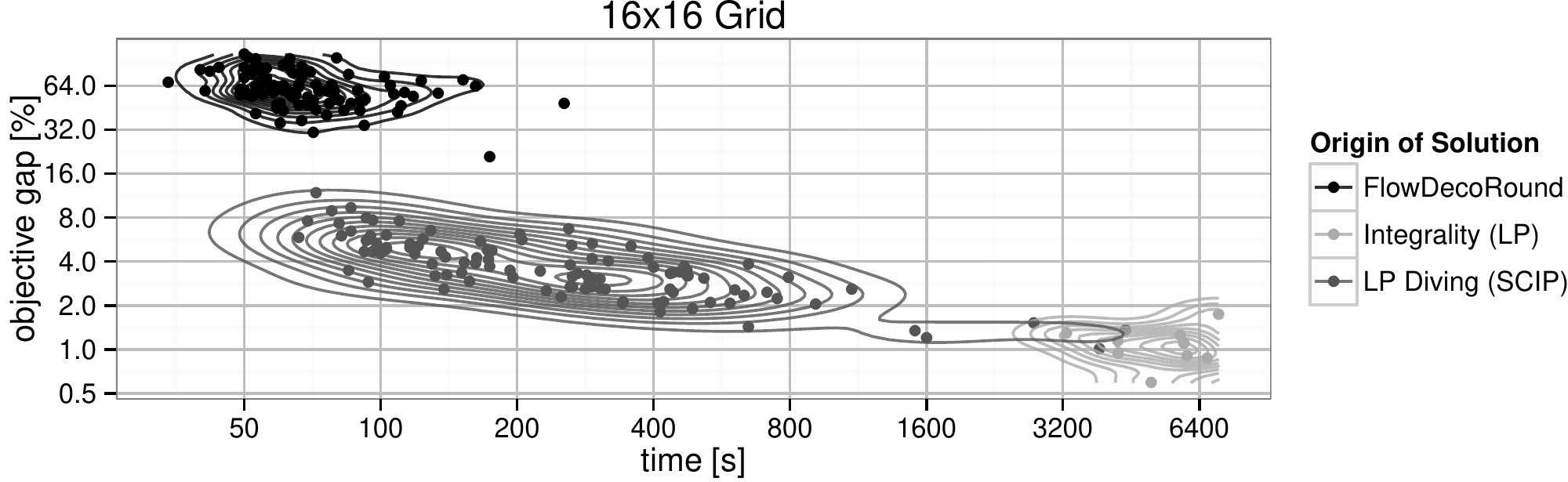}
\end{subfigure}
\medskip

\begin{subfigure}[t]{1\textwidth}
\includegraphics[width=1\textwidth]{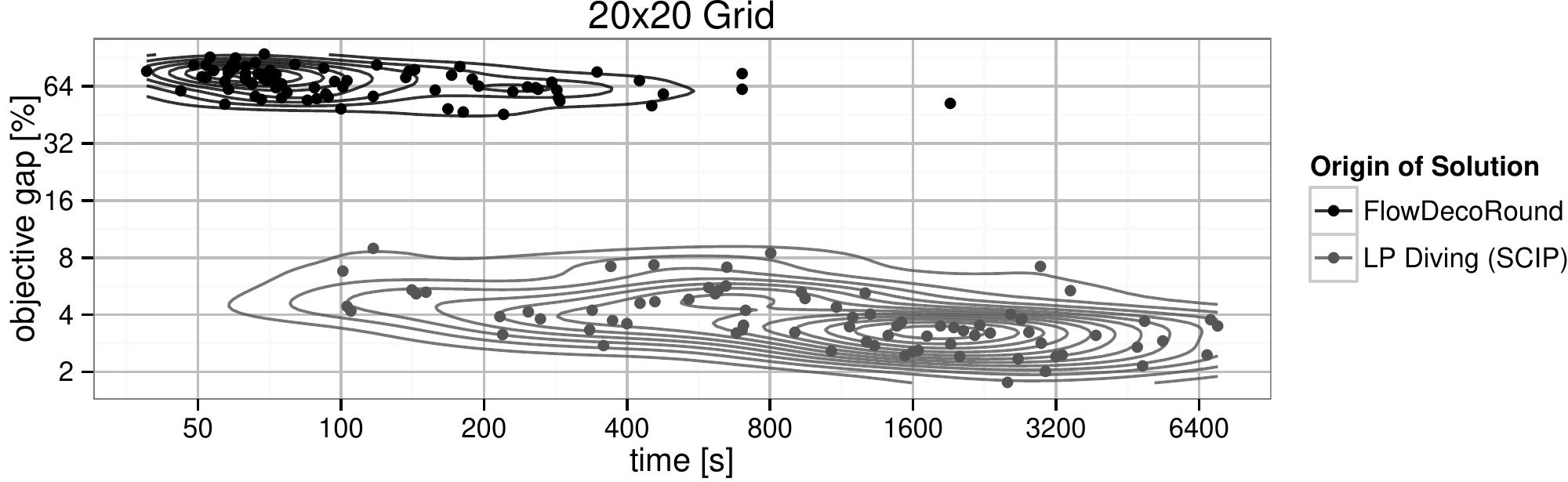}
\end{subfigure}
\medskip

\begin{subfigure}[t]{1\textwidth}
\includegraphics[width=1\textwidth]{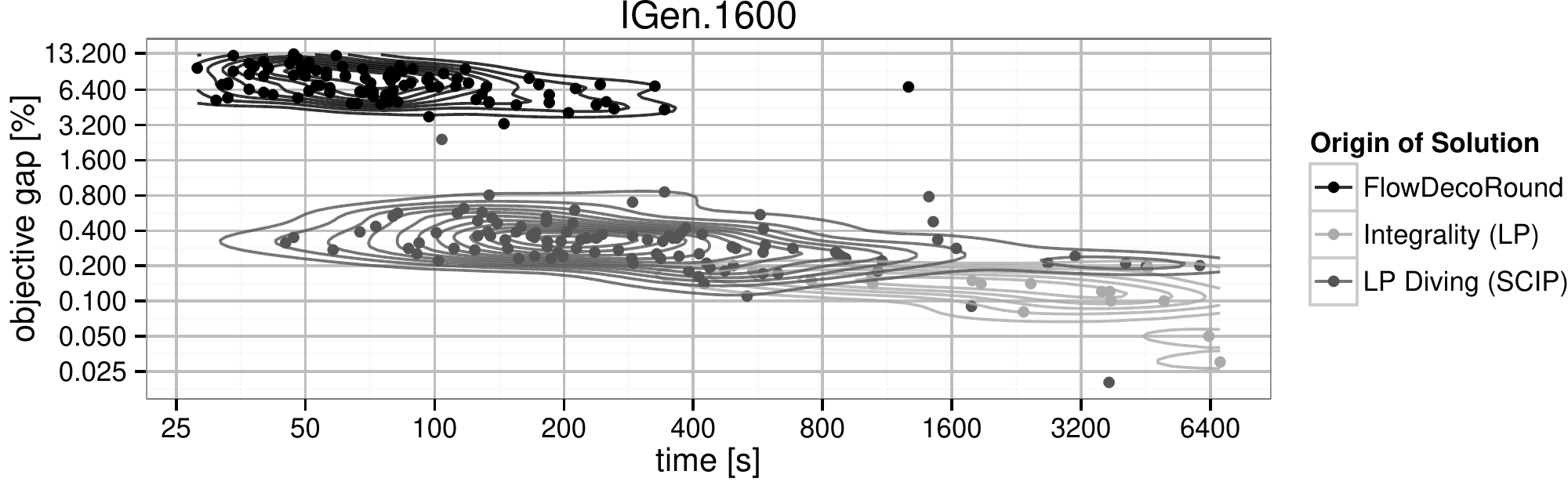}
\end{subfigure}
\medskip

\begin{subfigure}[t]{1\textwidth}
\includegraphics[width=1\textwidth]{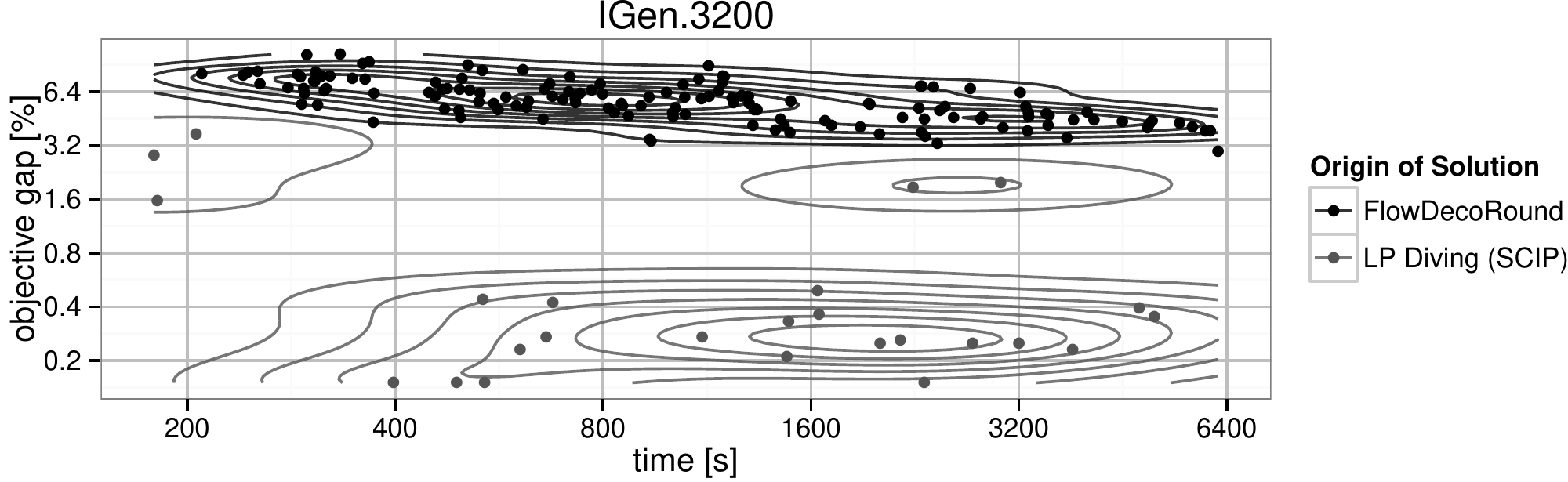}
\end{subfigure}
\caption{Primal solutions found by our solver for $n\times n$ grids with $n=16,20$ and the IGen topologies IGen.1600 and IGen.3200 over time with a borad classification of the solution's origin. The depicted gap corresponds to the lower bound at the time the solution was found. Note the logarithmic x- and y-axis.}
\label{fig:main-results-primal}
\end{figure}

Figure~\ref{fig:main-results-gap} depicts the objective gap over time. First note that, independent of the problem class, the gap decreases substantially during the first $30$ minutes while decreasing only slightly during the last $90$ minutes of execution. Furthermore, while for grid gaphs with $n=16$ some solutions can be solved, this is not the case for $n=20$. Considering the Internet topologies, IGen.1600 instances can be solved very close to optimality within few tens of minutes. While some of the IGen.3200 instances exhibit a similarly low gap, the median gap is generally a magnitude higher.

To better understand the progress in the objective gap, we will investigate both the progress in the lower as well as the primal bound. In Figure~\ref{fig:main-results-dual} the progress of the lower bound during the first $3$ minutes is depicted. As measure we use the relative lower bound, i.e. the lower bound with respect to the best lower bound achieved after $2$ hours. Note that at point in time $0$ the respective lower bound is the root relaxation without any separated cuts stemming from \ref{IP:CutSteiner} and \ref{IP:CutTerminal}. Across the board, the relative lower bound after $3$ minutes is within a margin of less than $0.02\%$. Thus in the remaining $117$ minutes of executions, the lower bound only improves sligthly.

To investigate the primal bound, i.e. the objective value of the best solution found so far, we depict in Figure~\ref{fig:main-results-primal} the origin and quality of solutions found. Primal solutions may either be generated by our heuristic \ref{alg:FlowDecoRound}, a bundled heuristic of SCIP or integral solutions of the LP relaxation. Importantly, all but less than $10$ solutions found by SCIP's heuristics were generated using LP diving, such that we only depict these. Considering grid graphs, our heuristic \ref{alg:FlowDecoRound} does generally construct solutions very quickly with an objective gap of $30\%-100\%$ whereas LP diving heuristics find solutions with a gap of an order less. The objective gap of solutions found by SCIP's diving heuristics are essentially the same for $n=16$ and $n=20$, while for $n=20$ solution generation is clearly delayed. Considering the Internet topologies, the heuristic developed by us achieves a gap of $3\%-13\%$, while LP diving heuristics provide again the best solutions within a gap of less than $0.8\%$. Finally, note that solutions found by integral linear relaxations yield the best objective values but are only found very late in the branch-and-cut process and only for the smaller instances of $16 \times 16$ grids and IGen.1600. We conclude the analysis of our results with the following observations. 
\begin{enumerate}
\item Using formulation \ref{alg:MIP} we can solve realistically sized instances on Internet topologies near optimally. Considering IGen.3200 instances, with a gap of as much as $4\%$. We believe this to be due to the lack of high-quality solutions rather than due to low-quality lower bounds. 
\item The heuristic \ref{alg:FlowDecoRound} presented in Section~\ref{alg:FlowDecoRound} performs quite well on Internet topologies in which costs are not chosen uniformly.
\item $n \times n$ grid instances with uniform costs as presented in Section~\ref{sec:results-problem-classes} seem to be hard to solve even for comparatively small sizes of $n=16,20$.
\end{enumerate}

\subsection{Comparison with MIP-A-CVSAP-MCF}
\label{sec:results-comparison-mcf}
Having presented the computational results for formulation~\ref{alg:MIP} in the above section, we will now examine the performance of~\ref{alg:MCF-MIP} on $20 \times 20$ grid graph instances as well as on the smaller IGen.1600 instances. Already for IGen.1600 the formulation~\ref{alg:MIP} induces around one million binary variables. Therefore, we have chosen to use the commercial MIP solver CPLEX~\cite{cplex} instead of SCIP to solve these instances.  To generate the instances for CPLEX we have modeled~\ref{alg:MCF-MIP} in the GNU Mathematical Programming Language~\cite{glpk}. Both the model files and the corresponding data files are available at~\cite{rostSchmidWeb}. We ran CPLEX with standard parameters making 12 GB of RAM available to it. Again, we terminate the execution after $2$ hours. Log files of the experiments running CPLEX are also obtainable from~\cite{rostSchmidWeb}. Note that CPLEX generally performs better than SCIP (see~\cite{koch2011miplib} for a comparison) in direct comparisons. The experiments are therefore a priori biased in favor of~\ref{alg:MCF-MIP}.

\begin{figure}[b]
\includegraphics[width=1\textwidth]{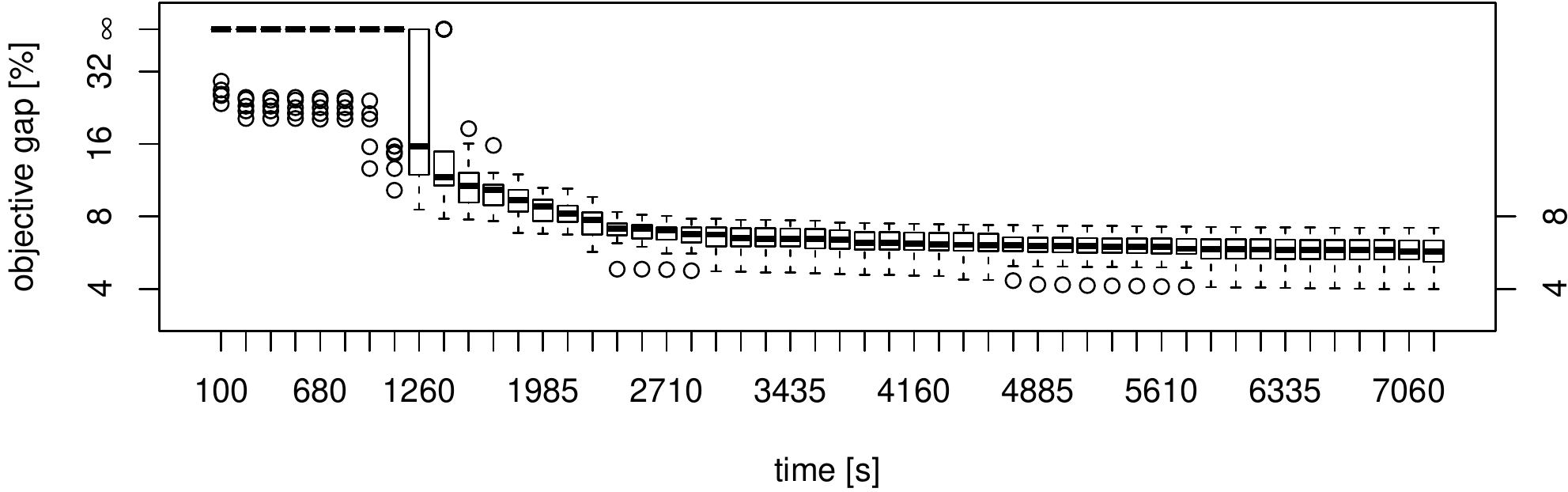}
\caption{Objective gap over time for the $20 \times 20$ grid instances using formulation~\ref{alg:MCF-MIP} which is being solved by CPLEX. A gap of $\infty$ denotes that no primal solution has been found.}
\label{fig:grid-20-cplex-gap}
\end{figure}

In Figure~\ref{fig:grid-20-cplex-gap} the objective gap of the $20 \times 20$ grids using formulation~\ref{alg:MCF-MIP} is shown over time. Compared to the corresponding plot for $n=20$ in Figure~\ref{fig:main-results-gap}, we notice that our approach yields solutions much quicker and that the final gap after $2$ hours using~\ref{alg:MIP} lies approximately $3\%$ below the gap found by CPLEX. CPLEX could only find solutions for $3$ of the $25$ instances of IGen.1600. We therefore do not provide a plot of the gap over time for these instances.

To computationally compare the strength of formulations~\ref{alg:MIP} and~\ref{alg:MCF-MIP}, we will compare the lower bounds by the following metric.

\begin{definition}[Relative Improvement of Formulations]
The relative improvement of formulation $\foM{A}$ over formulation $\foM{B}$ for a minimization problem is defined as 
\[
	I^{\textnormal{rel}}_{\textnormal{dual}}(\foM{A},\foM{B}) = \frac{D_{\foM{A}} - D_{\foM{B}}}{P_{\textnormal{best}} - D_{\foM{B}}}~,
\]
where $D_\foM{A}$ denotes the dual bound of $\foM{A}$, $D_\foM{B}$ denotes the dual bound of $\foM{B}$ and $P_{\textnormal{best}}$ is the objective value of the best solution found overall.
\end{definition}

Note that $I^{\textnormal{rel}}_{\textnormal{dual}}(\foM{A},\foM{B})=1$ holds iff. $D_{\foM{A}} = P_{\textnormal{best}}$ and therefore formulation $\foM{A}$ has proven the optimality of $P_{\textnormal{best}}$. Otherwise,  $I^{\textnormal{rel}}_{\textnormal{dual}}(\foM{A},\foM{B})$ measures the improvement of the dual bound using formulation $\foM{A}$ compared to the absolute gap of formulation $\foM{B}$. Furthermore note that this metric is independent of the process of determining primal solutions, since the best overall solution is chosen indepently of which formulation provided it and when it was found.

Figure~\ref{fig:relative-improvement-cplex-scip} depicts $I^{\textnormal{rel}}_{\textnormal{dual}}(\textnormal{\ref{alg:MIP}},\textnormal{\ref{alg:MCF-MIP}})$ over time. As it takes CPLEX up to $1200$ seconds to determine the root relaxation for IGen.1600 instances, we can only compare the relaxations after this point in time. Clearly, for both problem classes the relative improvement is substantial. We therefore conclude by stating that~\ref{alg:MIP} yields distinctively better lower bounds than~\ref{alg:MCF-MIP}.
\begin{figure}[t]
\begin{subfigure}[t]{\textwidth}
\includegraphics[width=1\textwidth]{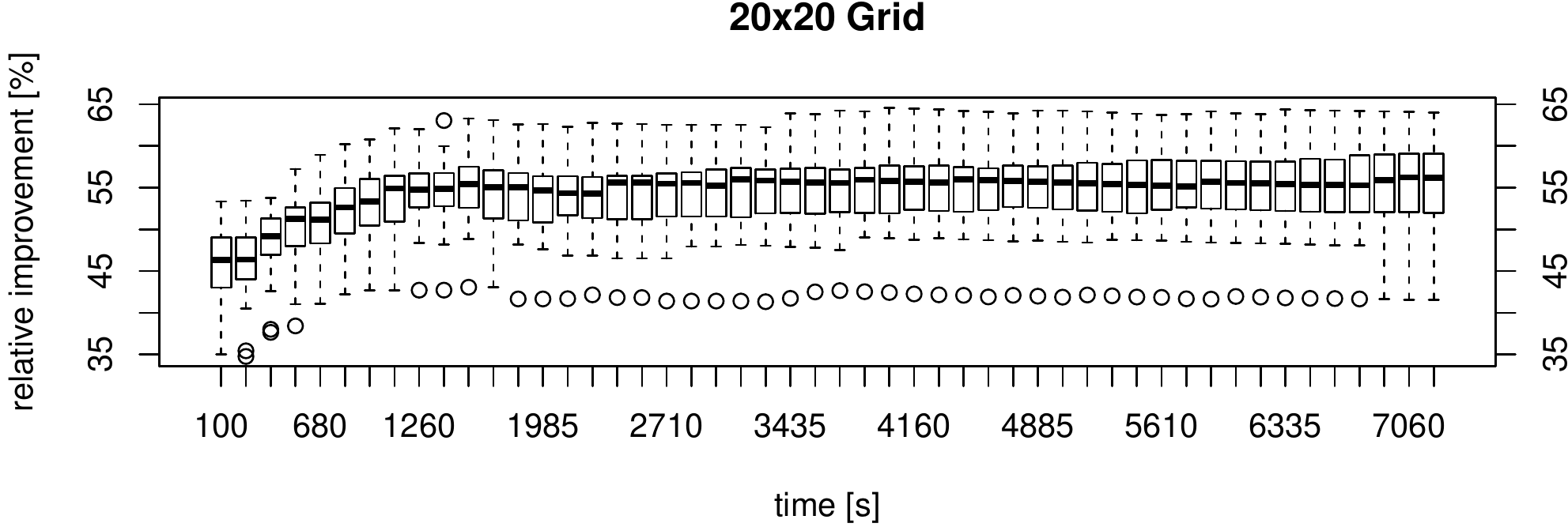}
\end{subfigure}

\medskip

\begin{subfigure}[t]{\textwidth}
\includegraphics[width=1\textwidth]{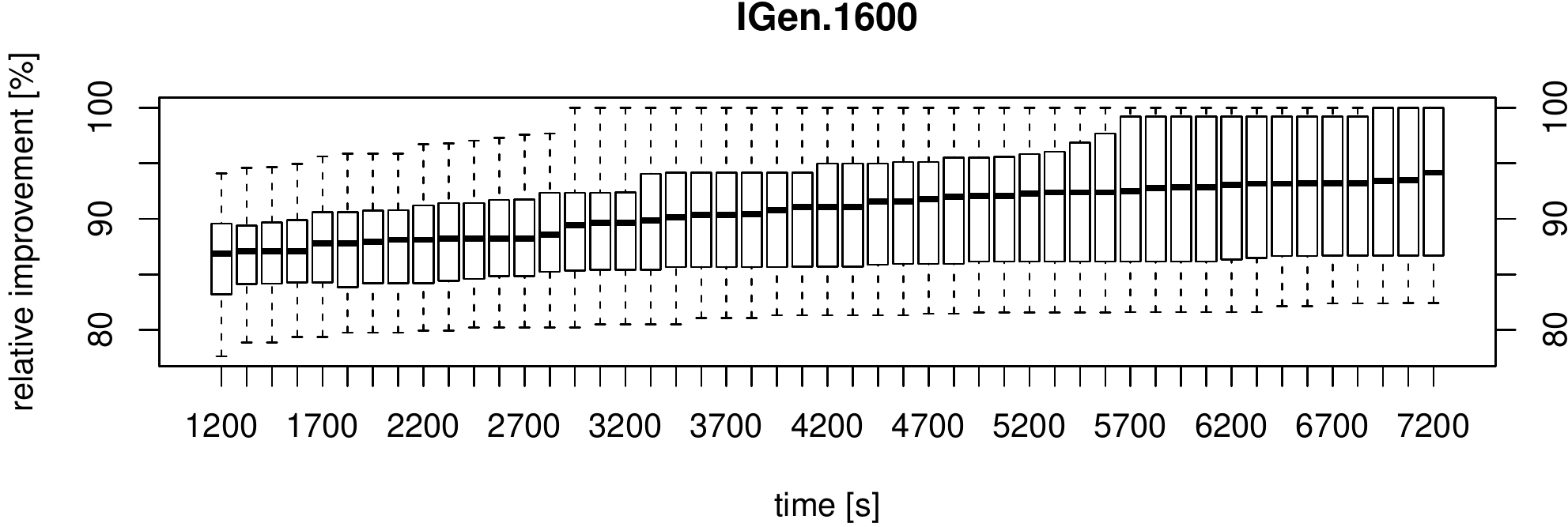}
\end{subfigure}
\caption{$I^{\textnormal{rel}}_{\textnormal{dual}}(\textnormal{\ref{alg:MIP}},\textnormal{\ref{alg:MCF-MIP}})$ on IGen.1600 and $20 \times 20$ grid instances over time.}
\label{fig:relative-improvement-cplex-scip}
\end{figure}

\subsection{Analysis of Runtime Allocation}
\label{sec:results-analysis}
We will now discuss the runtime allocation of different subroutines, to identify computational bottlenecks. In Figure~\ref{fig:runtime-allocation} the runtime allocation for the different problem classes is depicted according to the following classification:

\newenvironment{descriptionWorkaround}[1]% #1 is widest label
     { \list{}%
           {\settowidth\labelwidth{{#1}}%
            \leftmargin\labelwidth
            \advance\leftmargin\labelsep}%
}{\endlist}

\begin{descriptionWorkaround}{total}
\item[\textbf{B}] Time spent in branching procedures. This represents the time needed for deciding which variables to branch on. In case that, e.g., strong-branching~\cite{achterberg2009scip} is applied, this is a computationally expensive subroutine.\vspace{2pt}
\item[\textbf{H}] Time spent in heuristics. While this includes the execution time of our heuristic~\ref{alg:FlowDecoRound}, its contribution is negligible compared to the time spent in diving heuristics. It must be noted that diving heuristics require the solving of  linear relaxations, but do not trigger separation procedures. \vspace{2pt}
\item[\textbf{L}] Time spent in solving linear relaxations. \vspace{2pt}
\item[\textbf{S}] Time spent in separation procedures to check the validity of ~\ref{IP:CutSteiner} and~\ref{IP:CutTerminal} and to generate cuts where necessary
\end{descriptionWorkaround}

\begin{figure}[b!]
\centering
\begin{subfigure}[t]{0.1762\textheight}
\includegraphics[width=1\textwidth]{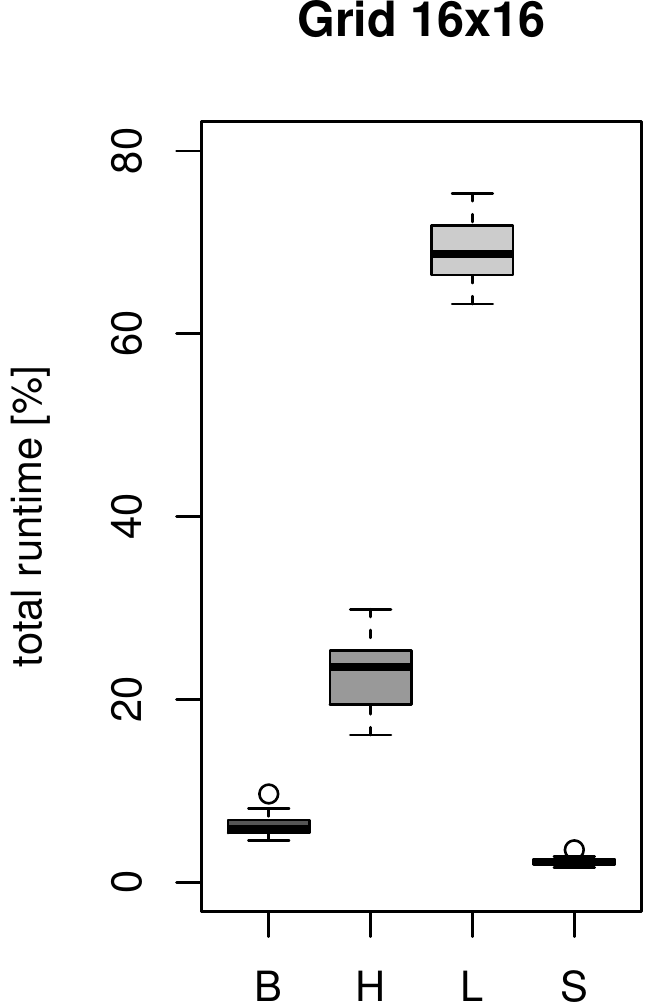}
\end{subfigure}
\begin{subfigure}[t]{0.23\textwidth}
\includegraphics[width=1\textwidth]{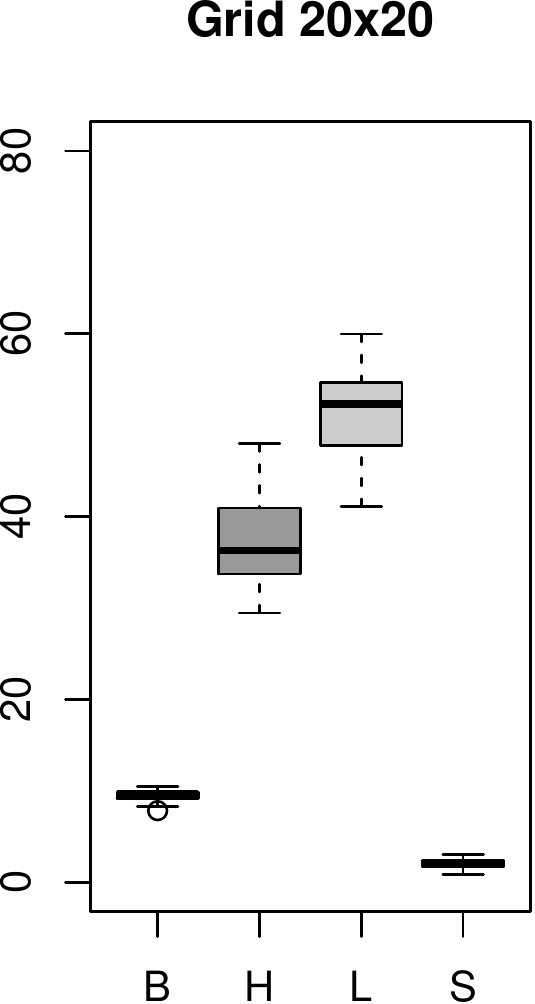}
\end{subfigure}
\begin{subfigure}[t]{0.23\textwidth}
\includegraphics[width=1\textwidth]{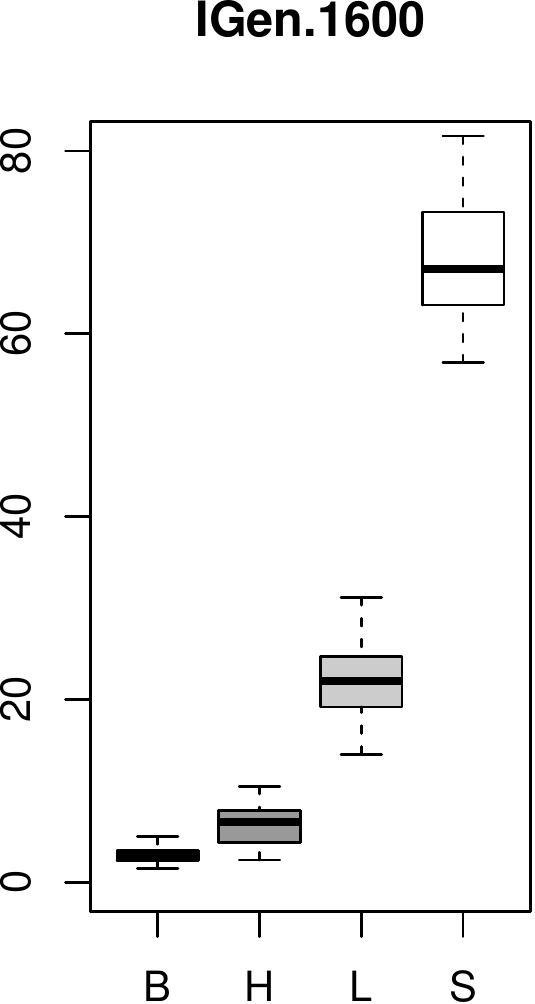}
\end{subfigure}
\begin{subfigure}[t]{0.23\textwidth}
\includegraphics[width=1\textwidth]{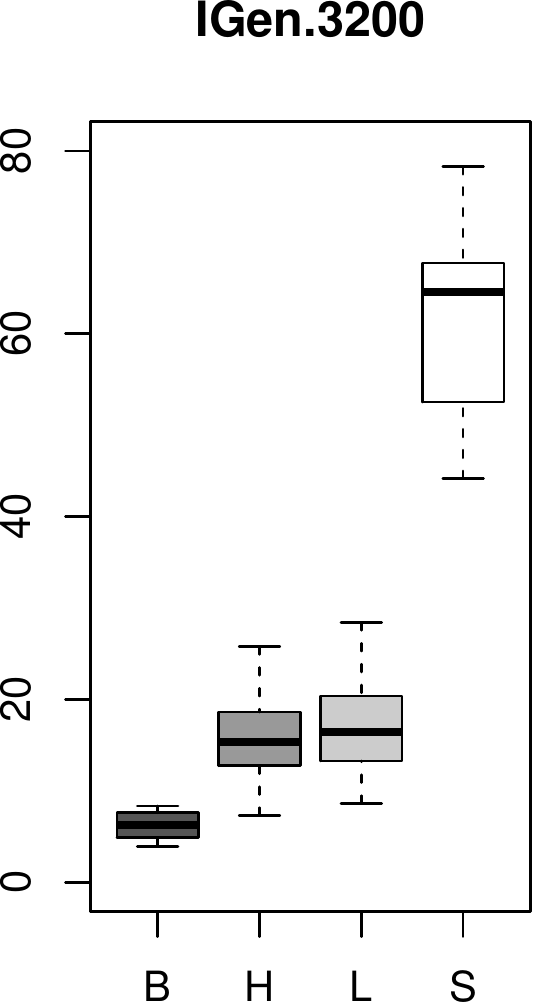}
\end{subfigure}
\caption{Runtime allocation distinguished according to the following categories: execution time spent for branching (B), for heuristics (H), for solving of the linear relaxations (L) and the execution time needed for separating~\ref{IP:CutSteiner} and~\ref{IP:CutTerminal} (S).}
\label{fig:runtime-allocation}
\end{figure}

The runtime allocation varies to a great extent between grid and Internet topology instances. Considering grid experiments, due to the small number of terminals and Steiner sites and as the graphs are comparatively small, the time spent in separation procedures is marginal. In fact, the time spent in solving linear relaxations and executing heuristics clearly dominates the runtime. In stark contrast, the runtime allocation for Internet topologies IGen.1600 and IGen.3200 shows the major influence of the separation procedures as the underlying gaphs are much larger and the number of terminals and Steiner sites has increased by a magnitude. Note that for the larger instances, i.e. grids with $n=20$ and IGen.3200, the time spent in heuristics and the time spent in solving linear relaxations almost equalizes.

The above observations allow the following two conclusions:
\begin{enumerate}
\item	Our choice of using Edmonds-Karp to separate inequalities~\ref{IP:CutSteiner} and~\ref{IP:CutTerminal} limits performance for larger instances with many terminals and Steiner sites. It should therefore be either optimized or replaced by another algorithm, e.g. the one of Hao and Orlin (see~\cite{koch1998solving} for a discussion).
\item 	While LP diving heuristics often find the best heuristic solutions (see Section~\ref{sec:results-main}), their overall runtime comes close to the time spent for solving hte linear relaxations. By devising more advanced combinatorial heuristics to replace these generic heuristics, performance could furthermore be improved.
\end{enumerate}

\section{Related Work} 
\label{sec:related-work}

The CVSAP problem differs from many models studied in the context of IPTV~\cite{BLTJ:BLTJ20217}, sensor networks~\cite{stacy,in-net-sensor}, or fiber-optical transport~\cite{BLTJ:BLTJ20217}, to just name a few, in that the number and placement of processing locations is subject to optimization as well.
The problem is complicated further by the fact that the communication between sender and receiver may be in-network processed \emph{repeatedly}. The result from~\cite{molnar2011hierarchies} on multi-constrained multicast routing also applies to CVSAP: any algorithm limited to (directed) acyclic graphs cannot
solve the problem in general. Generally, while there exist heuristic and approximate algorithms for related problem variants,
we are the first to consider exact solutions.

The two closest models to CVSAP are studied in~\cite{shi2001proposal} and~\cite{olivieraPardalos2011MathematicalAspectsOfNetworkRouting}. However, while~\cite{shi2001proposal} already showed the applicability of selecting only a few processing nodes for multicasting, no concise formalization is given and
the described heuristics do not provide performance guarantees. In a series of publications, Oliviera and Pardalos consider the Flow Streaming Cache Placement Problem (FSCPP)~\cite{olivieraPardalos2011MathematicalAspectsOfNetworkRouting}. Unfortunately, their FSCPP definition is inherently flawed as it does not guarantee connectivity see (Lemma \ref{lem:PardalosIncorrect} in Section~\ref{sec:appendix:proofs}). Interestingly, the authors also provide a correct approximation algorithm, which however only considers the weak model which ignores traffic. 

\subsubsection{Other related problems and algorithms.} The CVSAP is related to several classic problems. For example, CVSAP generalizes the \emph{light-tree} concepts (e.g.,~\cite{cai2005improved})
in the sense that ``light splitting'' locations can be chosen depending on the repeatedly processed traffic; our approach can directly be used to optimally solve
the light-tree problem. In the context of wave-length assignment, Park et al.~\cite{park2010virtual} show that a small number of virtual splitters can
be sufficient for efficient multicasting.
Our formalism and the notion of hierarchy is based on the paper by Molnar~\cite{molnar2011hierarchies} who studies the structure of the so-called multi-constrained multicast routing problem. Unlike in the CVSAP, an edge appears at most once in the solution.
If the cost of in-network processing is zero and all nodes are possible Steiner sites, the CVSAP boils down to a classic Steiner Tree Problem~\cite{goemans1993catalog} and its degree-bounded variants~\cite{lee1993strong}. A closer look shows that CVSAP can be easily modified to
generalize the standard formulation of prize-collecting Steiner trees~\cite{Johnson:2000:PCS:338219.338637} where used edges entail costs, and connected terminals may come with a benefit. However, CVSAP does not directly generalize other STP variants where disconnected nodes yield penalties~\cite{Johnson:2000:PCS:338219.338637} or which need to support anycasts~\cite{Demaine:2009:NST:1577399.1577429}.
Lastly, CVSAP generalizes the standard facility location problem~\cite{gollowitzer2011mip}.

\subsubsection{Mathematical programs.} For a good overview
on Integer Programs and avoiding inefficient relaxations
that lead to high runtimes (e.g.~\cite{karget2000building}), we refer the reader
to~\cite{optint}.
Our mathematical program builds upon the separation approach
by Koch et al.~\cite{koch1998solving} (see also~\cite{goemans1993catalog}).
Unlike~\cite{koch1998solving}, for CVSAP it is not sufficient to
only compute the cuts from the senders, but also additional
cuts are introduced depending on whether aggregation nodes are opened.

\section{Conclusion} 
\label{sec:conclusion}

This paper presented VirtuCastto optimally solve the CVSAP.
We rigorously proved that, although the computed IP solution may contain directed cyclic structures and flows may be merged repeatedly, there
exists an algorithm to decompose the solution into individual routes.
Since the CVSAP problem is related to several classical optimization problems, we believe that our approach is of interest beyond the specific model studied here.

Complementing our theoretical results, we have undertaken an extensive computational evaluation to study the performance of our solution approach. We have validated that  the introduction of Directed Steiner Cuts, nested cuts and creep-flow can speed up computations significantly. Considering the generation of primal solutions, we have introduced a primal heuristic that generates feasible solutions far quicker than generic heuristics that are bundled together with SCIP. Our heuristic performs especially on the non-symmetric IGen instances well, yielding solutions with an objective gap below $14\%$. Lastly, our computational evaluation has demostrated that VirtuCast allows for solving realistically sized instances to within $4\%$ of optimality, while naive multi-commodity formulations perform significantly worse or cannot be used at all on large problem instances.

An interesting direction for future research regards the design of approximation algorithms as an efficient alternative to the rigorous optimization approach proposed in this paper. While in its general form CVSAP cannot be approximated, we believe that there exist good approximate solutions, e.g., for uncapacitated variants or bi-criteria
 models where capacities can be violated slightly.

{
  \bibliographystyle{abbrv} \bibliography{in-net_arxiv}
}

%\newpage

\begin{appendix}

\section{Deferred Lemmas} 
\label{sec:appendix:proofs}

\begin{lemma} \label{lem:IP3StrengthensFormulation}
The directed Steiner cuts (see~\ref{IP:CutTerminal}) can strengthen the formulation~\ref{alg:MIP}, i.e. improve the objective value of its LP relaxation.
\end{lemma}
\begin{proof}
Consider the simple example in which the whole network $\G = (\VG,\EG)$ consists only of three nodes on a line: $\VG=\{r,s,t\}$ and $\EG = \{(t,s),(s,r)\}$. We consider the following capacities and costs $\uS(s) = 10$, $\uR(r)=1$, $\cE(t,s) = \cE(s,r) = 1$, $\cS(s)=5$ and that  $r$ is the root, $s$ is the single Steiner location and $t$ the only terminal.
The optimal solution of~\ref{alg:MIP} without~\ref{IP:CutTerminal} and relaxing the constraints~\ref{IP:VarF} and~\ref{IP:VarX} to $\f \in \mathbb{R}^{\Eext}_{\geq 0}$ and $\xS \in [0,1]^{\sS}$ is $\f(t,s)=1$, $\xS_s = 1/10$ and $\f(s,r) = 1/10$ yielding an objective value of $5\cdot \xS_s + \f(t,s) + \f(s,r) = 1.6$
By introducing Constraint~\ref{IP:CutTerminal} $\f(s,r)$ must equal 1 and therefore, the solution obtained by introducing this constraint yields the integral solution $\f(t,s)=f(s,r)=1$ and $\xS_s = 0$ with objective value 2, therefore strengthening the model.
Note that we did not give values for the edges from and to the super sinks which are introduced in $\Eext$ as these do not influence the objective value.
\end{proof}

\begin{lemma}\label{lem:PardalosIncorrect}
The formulation for the Flow Cache Placement Problem (FCCP) given in~\cite{olivieraPardalos2011MathematicalAspectsOfNetworkRouting} is incorrect, as it does not ensure connectivity as claimed by the authors.
\end{lemma}
\begin{proof}
In a series of works, the latest being~\cite{olivieraPardalos2011MathematicalAspectsOfNetworkRouting}, Oliviera and Pardalos have introduced the Tree Streaming Cache Placement Problem (TSCPP) and the Flow Streaming Cache Placement Problem (FSCPP). The TSCPP can directly be understood as M-CVSAP in which flow must be routed along a tree. Using our notation, they require $G^{\f}_{\Eext}$, which included each used edge, to be a tree.

They introduce FSCPP as a generalization of TSCPP in which the constraint that $G^{\f}_{\Eext}$ must be a tree is simply removed. Their Integer Program therefore reduces to~\ref{alg:MIP} in which~\ref{IP:CutSteiner} (and~\ref{IP:CutTerminal}) are removed. However, it is easy to see that  $G^{\f}_{\Eext}$ might indeed be disconnected such that there exist caches (activated Steiner nodes) that do not receive any flow and we only informally describe a counterexample. Consider a graph with one source (i.e., root), one possible cache (a Steiner site) and two demand nodes (terminals). Then if one terminal sends flow to the root and the other sends flow to the Steiner node such that their paths do not cross, then all constraints of~\ref{alg:MIP} (except~\ref{IP:CutSteiner} and~\ref{IP:CutTerminal}) are satisfied (assuming appropriate capacities) and one of the terminals is not connected to the root. This contradicts the connectivity requirement.
\end{proof}

\end{appendix}

\end{document}